\documentclass[sigconf]{acmart}
\AtBeginDocument{%
  }

\setcopyright{acmlicensed}
\copyrightyear{2026}
\acmYear{2026}
\acmDOI{XX.XXXX}
\acmConference[ACM Conf. 'XX]{ACM Conf.}{June 03--05,
  2018}{Woodstock, NY}

\acmISBN{978-1-4503-XXXX-X/2018/06}

\usepackage{algorithmic}
\usepackage{algorithm}
\usepackage{bbm}
\usepackage{wrapfig}
\usepackage{epstopdf}
\usepackage{graphicx}
\usepackage{textcomp}
\usepackage{xcolor}
\usepackage{booktabs}
\usepackage{subfigure}
\usepackage{enumitem}
\newtheorem{thm}{Theorem}

\begin{document}

\title{Bipartite Randomized Response Mechanism for Local Differential Privacy}


\author{Shun Zhang}
\author{Hai Zhu}
\email{szhang@ahu.edu.cn}
\email{e23301231@stu.ahu.edu.cn}
\affiliation{%
  \institution{Anhui University}
  \city{Hefei}
  \country{China}
}


\author{Zhili Chen} 
\affiliation{%
  \institution{East China Normal University}
  \city{Shanghai}
  \postcode{200062}
  \country{China}
}
\email{zhlchen@sei.ecnu.edu.cn}

\author{Haibo Hu}
\affiliation{%
  \institution{The Hong Kong Polytechnic University}
  \city{Hung Hom}
  \country{Hong Kong}
}
\email{haibo.hu@polyu.edu.hk}


\begin{abstract}
  With the increasing importance of data privacy, Local Differential Privacy (LDP) has recently become a strong measure of privacy for protecting each user’s privacy from data analysts without relying on a trusted third party. In this paper, we consider the problem of high-utility differentially private release. Given a domain of finite integers $\{1,2,\cdots,N\}$ and a distance-defined utility function, our goal is to design a differentially private mechanism that releases an item with the global expected error as small as possible. The most common LDP mechanism for this task is the Generalized Randomized Response (GRR) mechanism that treats all candidate items equally except for the true item. In this paper, we introduce Bipartite Randomized Response mechanism (BRR), which adaptively divides all candidate items into two parts by utility  rankings given priori item. In the local search phase, we confirm how many high-utility candidates to be assigned with high release probability as the true item, which gives the locally optimal bipartite classification of all candidates. For preserving LDP, the global search phase uniformly selects the smallest number of dynamic high-utility candidates obtained locally. In particular, we give explicit formulas on the uniform number of dynamic high-utility candidates. The global expected error of our BRR is always no larger than the GRR, and can offer a decrease with a small and asymptotically exact factor. Extensive experiments demonstrate that BRR outperforms the state-of-the-art methods across the standard metrics and datasets.
\end{abstract}



\received{20 February 2007}
\received[revised]{12 March 2009}
\received[accepted]{5 June 2009}

\maketitle

\section{Introduction}

In recent years, as concerns over data privacy have intensified~\cite{TM22}, Differential Privacy (DP) has emerged as a vital approach for protecting individual information~\cite{DKT2025}. DP, first introduced by Dwork in $2006$~\cite{DMN2006}, addresses the issue of privacy leakage caused by minor changes in the data source. With rigorous theoretical guarantees, centralized DP ensures that the output information remains influenced by any single record only within a specified threshold, thus preventing third parties from inferring sensitive information based on output variations. Users have to provide original data to a central server, which acts as a trusted entity to preserve centralized DP. In contrast, Local Differential Privacy (LDP) \cite{KLNRS11} allows users to send data with added noise directly to the central server. 
This ensures that individual privacy is protected before the  data upload.
\begin{figure}[tb]
\centering
\subfigure{
 \begin{minipage}{0.47\linewidth}
 \includegraphics[width=1.00\textwidth,trim=0.4cm 0.5cm 0.35cm 0.5cm, clip]{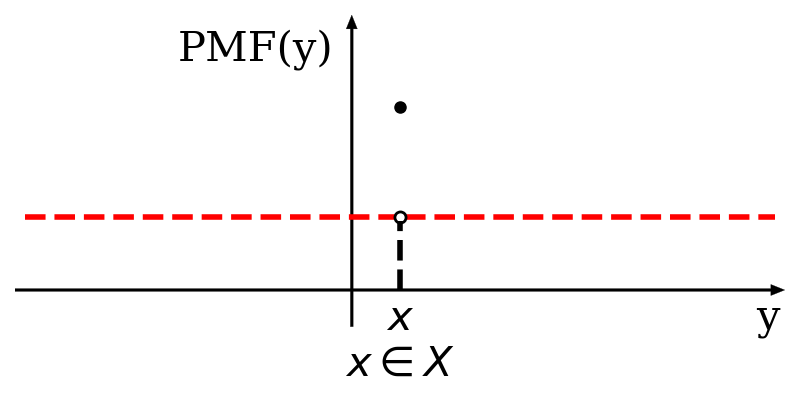} 
 \centerline{\fontsize{9}{10}\selectfont(a) GRR mechanism \cite{KOV2016}
 }
 \end{minipage}
}
\subfigure{
 \begin{minipage}{0.47\linewidth}
 \includegraphics[width=1.00\textwidth,trim=0.4cm 0.5cm 0.35cm 0.5cm, clip]{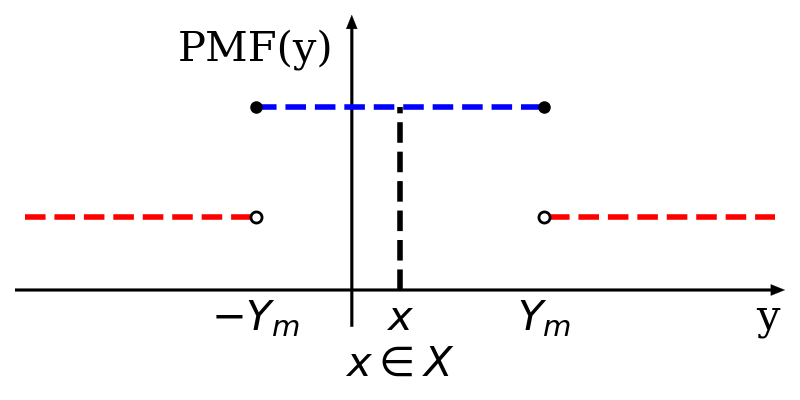} 
 \centerline{\fontsize{9}{10}(b)   BRR mechanism}
 \end{minipage}
}
\caption{\fontsize{9}{10} Probability mass 
function (PMF) for $GRR$ and $BRR$.}
\label{fig:GRR_BRR_PDF}
\end{figure}

One foundational mechanism in LDP 
is the Randomized Response (RR) mechanism introduced by Warner~\cite{W1965}. 
While RR is effective for protecting binary data, its applicability is limited to scenarios where the answer domain consists of only two values. To address this, the Generalized Randomized Response mechanism (GRR, also known as $k$-RR)~\cite{KOV2016} 
is proposed. GRR extends the RR mechanism to accommodate larger answer domains consisting of $N$ possible candidates as shown
in Fig. \ref{fig:GRR_BRR_PDF}(a). Under GRR, the data publisher selects and reports the true value with a fixed probability $p$, while assigning an equal probability to each of the remaining $N-1$ values with probability $\frac{1-p}{N-1}$. This extension retains the privacy-preserving properties of RR while broadening its applicability to use cases involving multi-valued categorical data, such as user preferences.
It can be used directly as a competitive mechanism for a range of tasks, including frequency estimation \cite{FCLG2023,WLJ2019}, 
data collection ~\cite{YHMZ2019}, decision tree ~\cite{MZCY2023}, and  deep learning ~\cite{MPBKLCA2020}.
The GRR mechanism is order optimal in the low privacy regime (large $\epsilon$) for a series of utility functions \cite{KOV2016}.
However, GRR has some limitations on the scenarios as follows: 

\begin{itemize}
\item Candidates within the answer domain are often numerous. With the increasing number of candidates in the equidistant point domain,  the quality loss of GRR grows explosively as shown
in Fig. \ref{fig:bQ}(a). 

\item  Candidates often have different degrees of utility. 
This factor is not taken into account by the GRR mechanism that treats all candidates equally except the true item.

\item 

Users sometimes require high privacy levels (small $\epsilon$) for protecting sensitive data. However, GRR can not achieve optimal utility in this case \cite{KOV2016}. 

\end{itemize}

To address these limitations, one has to adjust the probability distribution on candidates for achieving higher utility. This raises the open question of how to make the adjustment efficient with theoretical guarantees.
In this paper, we propose the \textit{Bipartite Randomized Response} (BRR) mechanism as an alternative to the GRR
mechanism for the task of differentially private release.
BRR incorporates the utility of candidates in the answer domain and consists of the local and global search phases.
Starting with the GRR framework, BRR adjusts the probabilities assigned to each candidate based on their utility to the true answer. 



Let \( X \) denote the space of true values and \( Y \) the space of released values, where \( Y \) is identical to \( \mathcal{X}=\{1,2,,\ldots,N\}\). When perturbing a specific \( x_i \), BRR assigns higher publishing probabilities to the \( m \) candidates most close to \( x_i \) (with higher utility) as shown
in Fig. \ref{fig:GRR_BRR_PDF}(b).
By adjusting the distribution of publishing probabilities, BRR increases the likelihood of publishing data that are closer to the true item while decreasing the likelihood for data that are farther. This reallocation mechanism renders the publishing probabilities more rational.

We can theoretically demonstrate that BRR significantly reduces utility loss relative to GRR, cf. Theorem \ref{thm:global_ratio_qloss} in the equidistant point
domain.
Fig.~\ref{fig:bQ}(a) shows that BRR has a significantly lower query loss ($QLoss$) compared to GRR with varying $N$. In addition, when $N$ tends to infinity, the asymptotically exact  quality-loss (and global expected-error) ratio of BRR to GRR is verified in Fig.~\ref{fig:bQ}(b).


BRR enjoys the same desirable properties
of the GRR mechanism stated above, and its local and global expected errors are never higher in any privacy region, but can be exponential times (dependent on  privacy budget $\epsilon$) lower than those of the GRR mechanism as long as the count $N$ of candidates is sufficiently large. Furthermore, we show that in reasonable
settings no better mechanism exists: the BRR mechanism is  optimal in a reasonable ``global'' sense,  independent of the prior distribution.

The \textbf{main contributions} of this work are summarized as follows:
\begin{itemize}
\item We propose the \textit{Bipartite Randomized Response} (BRR) mechanism as an alternative to the Generalized Randomized Response (GRR)
mechanism for the task of LDP release while aiming to obtain global high utility under strict LDP
constraints.


\item  BRR mechanism is conducted through a two-phase approach. In the first phase, given priori item,   we sort the candidates according to their utilities and divided into two parts for maximizing local utility. 
In the second phase, for enforcing LDP the global search selects uniformly the smallest count among dynamic high-utility parts obtained locally.

\item  For the BRR mechanism in the equidistant point domain, we present explicit formulas for the global uniform number of elements that make up the dynamic high-utility part. The expected error of BRR is always no more than that of GRR, and offers a decrease with a small and asymptotically exact ratio as elaborated later in Theorem \ref{thm:global_ratio_qloss}. 


\item  A series of experiments are executed for comparisons and applications on public datasets. The significant utility improvements over the state-of-the-art in both discrete and continuous domains have impacts on practical deployments of local differential privacy and the privacy-utility trade-off in social service platforms.

\end{itemize}

\textbf{Roadmap.} In the Section 2, some preliminaries are presented. Section 3 expounds our  BRR mechanism in the equidistant point domain.
Section 4 discusses the exploration of BRR in general domains.
Section 5
presents some comparative experiments on BRR and the state-of-the-arts response mechanisms. Section 6 reviews the related work. Finally, we offer a comprehensive conclusion of this paper.

\section{Preliminaries}
In this section, we delve into the concept of $\epsilon$-differential privacy, local differential privacy and associated noise-adding mechanisms, including Generalized Randomized Response (GRR), the Exponential Mechanism, and the Laplace Mechanism. These methods constitute pivotal techniques for implementing differential privacy, and each will be elucidated in detail below. Table~\ref{notation} summarizes the main notations in this paper.

\begin{table}[tb]
\caption{Summary of Notations}
\label{notation}
        \centering
		\begin{tabular}{p{0.5cm}l}
			\toprule
			Symbol  & \quad Definition \\
			\midrule
        $X$ &  {the domain of priori values or points}\\
        $Y$ &  {the domain of reported values or points, identical to $X$ }\\
        $N$ &  {the cardinality of domain $Y$, i.e., $N=\lvert Y \rvert$}            \\
        $Y_m$& {the set of m high probability $e^\epsilon$-weighted candidates }\\
        $\epsilon$ & privacy budget        \\ 
        $p$ &{probability of publishing the true item in GRR}\\
        $q$ &{probability of publishing non-true item in GRR}\\
        $m$ &  {the unified count of high-weighted terms}            \\
        $p_m^\ast$&{probability of reporting each top-$m$ utility item in BRR} \\
        $q_m^\ast$&{probability of reporting each low utility item in BRR}    \\
        $k$ & { real term currently being processed}\\
       $\lambda_i$  & {the $i$-th highest utility of items  given the true item}   \\
        $\lambda^{(k)}_i$ & {the $i$-th highest utility of items  for the $k$-th priori point}  \\
       $s_i$  & {weight of reporting the $i$-th item, with initialization $s_1=e^\epsilon$}   \\
        $s^{(k)}_i$ & {weight of reporting the $i$-th item in the $k$-th priori  point}  \\
       $w_i$  & {normalized probability weight for the $i$-th item, $\Sigma_i{w_i}=1$}   \\
        $w^{(k)}_i$&{normalized weight of the $i$-th item in the $k$-th point}\\
        $Q$ & {local utility  function (expected error) defined by $Q=\Sigma_i{\lambda_i w_i}$} \\
        $\pi$ &  the prior probability distribution on the domain\\
        $Q^{(k)}$ &{local expected error at the priori point $k$  by Eq.  \eqref{Q_k}}\\
        $Q_g$ & {global utility function (expected error) defined by Eq.  \eqref{Q_g}}  \\
        $QLoss$ & {global service quality loss defined by Eq. \eqref{QLoss}}  \\
	 \bottomrule
    \end{tabular}
\end{table}

\subsection{$\epsilon$-Differential Privacy and Local Differential Privacy}
DP and LDP are both privacy-preserving mechanisms used to protect sensitive data. While both methods aim to ensure privacy, they differ in their implementation and application scenarios~\cite{LLSY2017}.
Let $D \sim D'$ denote  the neighboring relation between $D$ and $D'$, that is, the datasets \( D \) and \( D' \)  differ on adding or deleting only one element for unbounded DP \cite{LLSY2017}. We first recall the  classical concept of Centralized DP \cite{DMN2006} as follows.

\begin{definition}[$\epsilon$-Differential Privacy] \label{def:DP}
     A randomized mechanism \( \mathcal{M} : \mathcal{D} \to \mathcal{R}\) satisfies \( \epsilon \)-differential privacy (\( \epsilon \)-DP), where \( \epsilon \geq 0 \), if and only if for any neighboring datasets \( D \) and \( D' \), we have
\begin{equation}
 \forall S \subseteq \mathcal{R}: \Pr[\mathcal{M}(D) \in S] \leq e^{\epsilon} \cdot \Pr[\mathcal{M}(D') \in S],   
\end{equation}
where \(\mathcal{R}\) denotes the set of all possible outputs of the mechanism \( \mathcal{M} \).
\end{definition}

$\epsilon$-DP provides strong privacy guarantees under the assumption of a trusted data collector. However, in many practical scenarios, such as mobile or edge computing environments, users may not be willing to trust a central server holding their raw data. To address this concern, a stronger notion, Local Differential Privacy (LDP), is formulated in \cite{KLNRS11}.  We will use the definition of LDP given by Erlingsson et al. \cite{EPK14}.

\begin{definition}[$\epsilon$-Local Differential Privacy]\label{def:LDP}
    A randomized mechanism \( \mathcal{M} \) satisfies \( \epsilon \)-LDP, where \( \epsilon \geq 0 \), if and only if for any input \( x ,\; x' \in X\), we have 
\begin{equation}
 \forall y \in \mathcal{M}(X) \; :\; \Pr[\mathcal{M}(x) = y] \leq e^{\epsilon} \cdot \Pr[\mathcal{M}(x') = y].
\end{equation}
Here, \(\mathcal{M}(X)\) denotes the set of all possible outputs of \(\mathcal{M}\).
\end{definition}

A key quantity in designing mechanisms that satisfy either $\epsilon$-DP or $\epsilon$-LDP is the sensitivity of a scoring function~\cite{DMN2006}. It captures the maximum influence that a single individual can have on the result of a numeric query. 
\begin{definition}[Sensitivity]\label{def:q}
    Let $D \sim D'$. The sensitivity of a scoring function $q$, denoted by $\Delta q$, is given by
\begin{equation}
    \Delta q = \max_{D \sim D'} | q(D) - q(D') |.
\end{equation}
\end{definition}

\subsection{Laplace Mechanism}
The Laplace Mechanism~\cite{DMN2006,DJW2013}  is the most widely used obfuscation mechanism for DP, which accomplishes privacy protection by adding noise drawn from a Laplace distribution to the query results. 
\begin{definition}[The Laplace Mechanism]
    For a given query function $q$, the noisy output is defined as:
\begin{equation}
\mathcal{M}(D) = q(D) + \text{Lap}\left(\frac{\Delta q}{\epsilon}\right).
\end{equation}
Here, $\text{Lap}(\frac{\Delta q}{\epsilon})$ denotes a Laplace distribution with a mean of $0$ and scale parameter $\frac{\Delta q}{\epsilon}$.
\end{definition}

\subsection{Exponential Mechanism}
The Exponential Mechanism, introduced by McSherry and Talwar~\cite{MT2007}, is a versatile DP mechanism suitable for arbitrary query ranges. The fundamental principle involves selecting outputs from the candidates space based on a score function. 
\begin{definition}[Exponential Mechanism]
    For any quality function $q:\mathbb{D}\: \times \:O \to \mathbb{R}$, $o \in O$, the Exponential Mechanism $\mathcal{M}(D)$ outputs $o$ with probability,
\begin{equation}
\Pr[\mathcal{M}(D) = o] \propto \exp\left(\frac{\epsilon q(D, o)}{2 \Delta q}\right).
\end{equation}
\end{definition}

\subsection{Generalized Randomized Response (GRR)}
Local privacy dates back to Warner~\cite{W1965}, who
introduced the Randomized Response (RR, also known as W-RR) mechanism 
to deal with binary responses.  Generalized Randomized Response (GRR, also known as $k$-RR)
is a generalization of
the RR to handle multiple choice problems by Kairouz et al.~\cite{KOV2016}.

Let $\mathcal{X}$ denote the input domain consisting of $N \geq 2$ elements. GRR maps $\mathcal{X}$ stochastically onto itself (i.e., $\mathcal{Y}=\mathcal{X}$). For any $ x\in \mathcal{X}$ and $y\in \mathcal{Y}$, the perturbation probability distribution is defined by
\begin{equation}\label{GRR}
\Pr[y|x]=
\begin{cases}
\frac{e^\epsilon}{ e^\epsilon +N-1}& \text{if}\  y=x,\\
\frac{1}{ e^\epsilon +N-1}&  \text{if}\  y\neq x. 
\end{cases}
\end{equation}
If $N = 2$, this gives the definition of RR mechanism.

\begin{figure*}[tb]
\centering
 \includegraphics[scale=0.52,trim=0.7cm 9cm 0.1cm 1.8cm, clip]{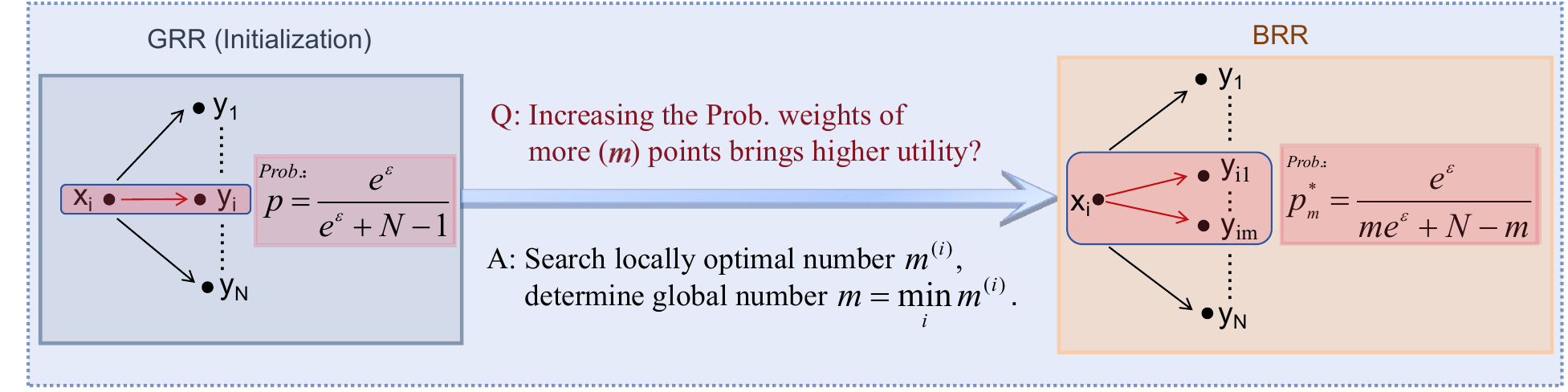} 
 \caption{Framework of BRR mechanism. Beginning with GRR, we maximize the pointwise prob. weight (up to $e^\epsilon$) of high-quality points by Alg. \ref{alg1} for obtaining locally optimized $k=m^{(i)}$ of each priori point $x_i$. Then the minimal $m^{(i)}$ presents the globally uniform  $m$ of $e^\epsilon$-weighted candidates (with prob. $p^*_m$), while the others are $1$-weighted, for the final BRR, cf. Alg. \ref{alg2}.
 }
 \label{fig:BRR}
\end{figure*}

\subsection{Design Consideration}\label{sect: prob_formu}

We consider a scenario where the user wants to protect the privacy of her/his actual value $x_i$ from a specified discrete domain $\mathcal{X}=\{1,2,,\ldots,N\}$,
by reporting a pseudo-value $y_i$
in a finite set $\mathcal{Y}$, where usually $\mathcal{Y}=\mathcal{X}$.
It is desirable to develop an obfuscation mechanism that satisfies $\epsilon$-differential privacy and generates perturbed values with effective performance. 

Naturally, beginning with GRR by \eqref{GRR}, we first search for locally optimized  high-quality points whose number $m^{(i)}$ is due to the starting priori point $x_i$. Then we confirm the uniform number $m$ of high probability $e^\epsilon$-weighted candidates for the desired BRR mechanism. 

On the construction of BRR, we have to overcome four challenges as follows:

\begin{itemize}
\item In the local search phase, why are the optimal distributed probability weights  simply bipartite, only $1$ and $e^\epsilon$?

\item
In the global search phase, how shall we confirm the uniform number $m$ of high probability $e^\epsilon$-weighted candidates that is independent of the priori distribution?

\item Can we compute the uniform number $m$ by some explicit formula? 

\item What is the impact of our proposed BRR mechanism on theory and applications? Specifically, to what extent does the BRR outperform the GRR in the utility of the mechanism?
\end{itemize}

\section{Design of BRR in equidistant point domain}

In this section, we consider the one-dimensional domain that consists of $N$ equidistant points on the real line. Without loss of generality, we assume $\mathcal{X}=\{1,2,\ldots, N\}$ and denote the data utility 
by the expected error/distance between the true and reporting/pseudo points. The longer distance gives the lower data utility. 

In the following, we first introduce the basic expression of BRR as illustrated by Fig.~\ref{fig:BRR}. Then we focus on determining the number \( m \) of high weighted candidates  in the BRR mechanism. This includes the local and global search phases for $m$.

\subsection{Bipartite Randomized Response (BRR)}

Given a priori point $x$, privacy budget $\epsilon$ and the uniform number $m$ of high probability $e^\epsilon$-weighted candidates, the set of $m$ points is denoted by \( Y_m \). Clearly, \( Y_m\subset Y \) depends heavily on the priori point $x$.

As shown in Fig.~\ref{fig:BRR}, we introduce Bipartite Randomized Response (BRR) mechanism whose  perturbation probability
distribution is defined by  
\begin{equation}
\Pr[y|x]=
\begin{cases}
\frac{e^\epsilon}{ me^\epsilon +N-m}& \text{if}~  y\in{Y_m },  \\
\frac{1}{ me^\epsilon +N-m}&  \text{if}~y\notin{Y_m}.
\end{cases}\label{def:BRR}
\end{equation}

When \( m = 1 \),  BRR becomes the initial  trivial mechanism GRR. In addition, following such a definition, the BRR satisfies differential privacy. Under this restriction, we will investigate how to search the optimal number $m$  for making the utility of the BRR mechanism as high as possible. This approach includes two phases: the local and global search phases.



\begin{algorithm}[tb]
  \renewcommand{\algorithmicrequire}{\textbf{Input:}}
  \renewcommand{\algorithmicensure}{\textbf{Output:}}
  \caption{Local highest-utility response search algorithm}
  \label{alg1}
  \begin{algorithmic}[1]
   \REQUIRE  quality loss\ $0\leq \lambda_1 \leq \lambda_2 \leq \cdots \leq \lambda_N$,\ \ privacy budget\ $\epsilon$\\

    \STATE Initialize: \ $s_1 = e^\epsilon$,\ $s_2= \cdots= s_N=1$,\ $w_1=\frac{e^\epsilon}{e^\epsilon+N-1}$,\\ $w_2= \ldots= w_N= \frac{1}{e^\epsilon+N-1}$,\ $Q=\Sigma_i{\lambda_i w_i}$, $m=1$
    \FOR {$i$ = 2 \textbf{to} $N$}
      \STATE{Compute $\frac{\partial Q}{\partial s_i}=\frac{\Sigma_j(\lambda_i-\lambda_j)s_j}{(\Sigma_js_j)^2}$}
      \STATE \textbf{if}\ {$\frac{\partial Q}{\partial s_i}<0$} \textbf{then}
        $s_i = e^\epsilon$ and $m=i$,
      \textbf{otherwise}, break
    \ENDFOR
    \STATE{$w_i=\frac{s_i}{\Sigma_j{s_j}}$ 
   \textbf{for} $i = 1$ \textbf{to} $N$}
  \ENSURE $(w_1,\ldots,w_N)$ and $m$
  \end{algorithmic}
\end{algorithm}

\subsection{Determining Locally Optimal $m$ for BRR}\label{sect: local_BRR}

For local publishing, given the actual point $x$ from the domain $\mathcal{X}$, let \(\lambda_i\) denote the data utility of reporting each point in $\mathcal{X}$. The utility can be defined adaptively in various scenarios such as the Jaccard similarity discussed in Section \ref{sect:jaccard}. Currently, we assign simply the Euclidean distance as the metric of data utility, which can also be called service quality loss. Then,  the closer reporting point to the actual point $x$ gives higher utility and smaller quality loss.
 After sorting, the quality loss of all points are increasingly like
$$0=\lambda_1 \leq \lambda_2 \leq \cdots \leq \lambda_N,$$
with the pairs $(y_1,\lambda_1),\cdots, (y_N,\lambda_N)$.
Accordingly, the initial probability distribution weights \(s_i\) can satisfy
$$e^\epsilon=s_1\geq s_2\geq \cdots\geq s_N=1.$$
 By normalization of the sum, define  \(w_i=s_i/\sum_j s_j\) such that $\Sigma_i{w_i}=1$. The local expected error of reporting is defined by $Q=\Sigma_i{\lambda_i w_i}$,
 i.e., 
 \begin{equation}\label{utility_local}
 Q=\frac{\Sigma_i{\lambda_i s_i}}{\Sigma_j s_j}.
 \end{equation} 
 
 To further explore the effect of \(s_i\) on 
\(Q\), we take the partial derivative of \(Q\) with respect to \(s_i\) as
\begin{equation}\label{Qloss_partial}
\frac{\partial Q}{\partial s_i}=\frac{\lambda_i\Sigma_js_j-\Sigma_j{\lambda_js_j}}{(\Sigma_j s_j)^2}=\frac{\Sigma_j(\lambda_i-\lambda_j)s_j}{(\Sigma_js_j)^2}.
\end{equation}

It can be observed that the monotonicity of \(Q\) with respect to \(s_i\) is independent of the value of \(s_i\). As illustrated in Algorithm \ref{alg1},  we perform the local search to minimize the expected error \(Q\).

\textbf{Step 1:}  Given the actual point $x$, initialize the search algorithm as the GRR mechanism, i.e., $s_1=e^\epsilon$,\ and the other, $s_2,\ldots, s_N$, are uniformly assigned the minimum weight value \(1\). 

\textbf{Step 2:}  Increase $s_2$ to the maximum value $s_1=e^\epsilon$, then successively increase $s_3, s_4, \ldots$, until some \(s_i\) satisfies $\frac{\partial Q}{\partial s_i}\geq 0$. Stop iterating at this point $y_i$, fix $m=i$ and $Y_m=\{y_1,\ldots,y_m\}$ and remain
$s_i,\ldots, s_N$ unchanged to be \(1\). 

\textbf{Step 3:}  
Update all weights $w_j$ by
\eqref{def:BRR} and finish the local search phase. 


Following this, different priori points could give different locally optimal $m$ and further different perturbation probability distributions by employing \eqref{def:BRR}. This will certainly violate $\epsilon$-differential privacy guarantee. Then, we have to search for a uniform number $m$, even independent of the prior distribution over the domain $\mathcal{X}$ in the next section.

\subsection{Determining Globally Uniform $m$ for BRR}

For global publishing with differential privacy guarantee, 
given each real occurrence \( k\in\mathcal{X} \), i.e., \(1\le k\le N\),
we first sort the quality losses  $\lambda_i^{(k)}$ and  initial weights $s_i^{(k)}$ of all points as follows,
$$0=\lambda_1^{(k)} \leq \lambda_2^{(k)} \leq \cdots \leq \lambda_N^{(k)},\ \
e^\epsilon=s_1^{(k)}\geq s_2^{(k)}\geq \cdots\geq s_N^{(k)}=1.$$
The weights $s_i^{(k)}$ are normalized to generate $w_i^{(k)}=\frac{s_i^{(k)}}{\Sigma_j s_j^{(k)}}$ with satisfying $\Sigma_i w_i^{(k)}=1$. Note that for different occurrence $k$, the sequence of pairs $\{(y_1,\lambda_1^{(k)}),\cdots, (y_N,\lambda_N^{(k)})\}$ varies.

For the occurrence item \(k\), the initial weight vector \(s^{(k)}\) is the same as the GRR, where the first weight \(s_1^{(k)}=e^\epsilon\) (for the point $k$ itself) and the remaining weights are $1$. 
In particular, the first pair is \((\lambda_1^{(k)},s_1^{(k)})=(0,e^\epsilon)\).

The local expected error (local quality loss) for the occurrence item \(k\) is denoted by 
\begin{equation}\label{Q_k}
Q^{(k)} = \sum_{i} \lambda^{(k)}_i w_i^{(k)}.  
\end{equation}

By averaging uniformly the local utility of all points, the global expected error (global utility loss function) is denoted by, with uniform average,
\begin{equation}\label{Q_g}
Q_g =\frac 1N\sum_k Q^{(k)} = \frac 1N \sum_k\sum_{i} \lambda^{(k)}_i w_i^{(k)}.
\end{equation}

 As illustrated in Algorithm \ref{alg2},  we perform the global search to minimize the quality loss \(Q_g\). This gives the sequence $(w_1,\ldots,w_N)$ uniformly for any priori $k$, then $w_i^{(k)}$ can be simply rewritten as $w_i$ in \eqref{Q_k} and \eqref{Q_g}. The search steps are summarized as follow.

\textbf{Step 1:}  Search locally optimal $m^{(k)}$ for each occurrence $k$ by Algorithm \ref{alg1}, cf. Lines 1-8.

\textbf{Step 2:}  Search globally uniform $m=\min_k m^{(k)}$, cf. Line 9.

\textbf{Step 3:}  
Update all weights $w_j$ by
\eqref{def:BRR} where $Y_m$ is replaced by $Y_m^{(k)}$ that contains the first $m$ points closest to the point (occurence) $k$,
and finish the global search phase, cf. Line 10.

Starting from the second weight \(s^{(k)}_2\), each weight is adjusted sequentially. During the adjustment process, the partial derivative of the utility function \(Q^{(k)}\) with respect to \(s^{(k)}_i\) is calculated to determine whether increasing the current weight can improve the utility.  

Once the weights of all occurrence items \(k\) have been optimized, we compare the optimal $m^{(k)}$ derived from each local search phase and select the minimum value as the global final threshold.
This allows us to construct a unified global BRR mechanism. In this global BRR mechanism:
\begin{itemize}
\item  The number \(m\) of high weighted candidates is chosen as the smallest  $m^{(k)}$ across all occurrences. This ensures that the obfuscation mechanism achieves higher utility than GRR at each occurrence point \(k\). If \(m\) increases, the local utility of publishing some occurrences would certainly decrease. 
In addition, this selection is independent of the prior distribution $\pi$ over the domain $\mathcal{X}$.

\item For each local occurrence \(k\), the first \(m\) points closest to the actual occurrence point are assigned a probability weight of \(e^\epsilon\), while the remaining points with weight of \(1\) before normalization. Thus, the so-called set $Y_m$ depends on the actual occurrence point $k$.

\end{itemize}

The core idea of BRR mechanism demonstrated by Alg. \ref{alg2}, is to dynamically adjust the weight distribution under the constraints of privacy protection, i.e., achieving an optimal trade-off between utility and privacy. 
In essence, the algorithms above aims to give the number $m$ of high weighted candidates.
Next, we present some explicit
formulas of the uniform number $m$.

\begin{algorithm}[tb]
  \renewcommand{\algorithmicrequire}{\textbf{Input:}}
  \renewcommand{\algorithmicensure}{\textbf{Output:}}
  \caption{Bipartite Randomized Response (BRR)}
  \label{alg2}
  \begin{algorithmic}[1]
   \REQUIRE  quality loss matrix $\{\lambda^{(k)}_j\}_{N\times N}$\ with $0\leq\lambda_1^{(k)} \leq \lambda_2^{(k)} \leq \cdots \leq \lambda_N^{(k)}$, privacy budget\ $\epsilon$\\

    \STATE Initialize: \ $s_1^{(k)} = e^\epsilon$,\ $s_2^{(k)}= \cdots= s_N^{(k)}=1$,\ \ $m^{(k)}=1$ 
    \FOR {$k$ = 1 \textbf{to} $N$}
    \STATE $w_1=\frac{e^\epsilon}{e^\epsilon+N-1}$,\ $w_2= \ldots= w_N= \frac{1}{e^\epsilon+N-1}$,\ \ $Q^{(k)}=\Sigma_i\lambda_i^{(k)}w_i$
     \FOR {$i$ = 2 \textbf{to} $N$}
      \STATE{Compute $\frac{\partial Q^{(k)}}{\partial s_i^{(k)}}=\frac{\Sigma_j(\lambda_i^{(k)}-\lambda_j^{(k)})s_j^{(k)}}{\left(\Sigma_js_j^{(k)}\right)^2}$}
      \STATE \textbf{if}\ $\frac{\partial Q^{(k)}}{\partial s_i^{(k)}}<0$ \textbf{then}
        $s_i^{(k)} = e^\epsilon$, $m^{(k)}=i$,
      \textbf{otherwise}, break
    \ENDFOR
    \ENDFOR
    \STATE Compute $m=\min_k m^{(k)}$
    \STATE Update all weights $w_j$ for $j$ = 1 \textbf{to} $N$ by
  \begin{align*}
      w_j=
\begin{cases}
\frac{e^\epsilon}{ me^\epsilon +N-m},& \text{if}\ j\,\le\,m\\
\frac{1}{ me^\epsilon +N-m},&  \text{if}\ j\,>\, m
\end{cases}  
    \end{align*}
  \ENSURE $(w_1,\ldots,w_N)$ and $m$
  \end{algorithmic}
\end{algorithm}

\subsection{Explicit
Formulas for $m$ and Utility}

\begin{figure}[tb]
\centering
 \centerline{\includegraphics[trim=0.1cm 0.32cm 0.1cm 0.1cm, clip,width=\linewidth]{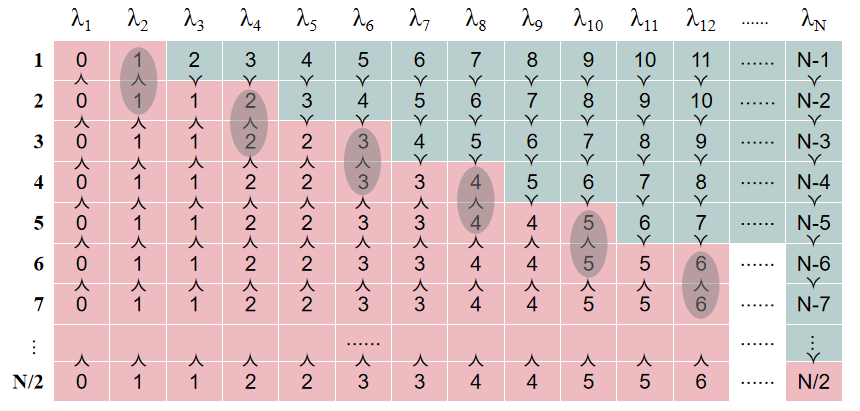}}
 \caption{Monotonicity characteristics of BRR with Euclidean distance utility
 for the array from $1$ to even $N$. The symbol $l
\prec l'$ in the column labeled by  $\lambda_i$ means that  the partial derivative $\frac{\partial Q}{\partial s_i}$
for the row of $l'$ is higher than that for $l$.}
 \label{fig3}   
\end{figure}
As above in Section 3,
the utility $\lambda$ is denoted by the Euclidean distance between the actual and reported numbers, i.e., simply the absolute value of their difference \(|x-y|\). 

As illustrated in Fig.~\ref{fig3},
after sorting $\lambda$ of all candidates for reporting, we have $0=\lambda_1\leq\lambda_2\leq\cdots\leq\lambda_N\le N$. For the numbers $1\sim N$, the distance array $[\lambda]$ for item $1$ is $[0,1,2,3,\cdots,N-1]$; the distance array for item $2$ is $[0,1,1,2,3,\cdots,N-2]$. 
For even $N$, we need only to consider the case that the priori number $k\le N/2$ due to the symmetry.

Since the distances are sorted in an ascending order, for any adjacent priori numbers $k$ and $k'=k-1$, 
by comparing their sorted distance arrays (rows) $[\lambda]$ and $[\lambda']$, 
we can find that there exists a single $j_0$, such that
$\lambda_1'=\lambda_1,\ \lambda_2'=\lambda_2,\ \cdots,\ \lambda_{j_0}'=\lambda_{j_0}$, while for the remaining $N-j_0$ distances, $\lambda_{j_0+1}'=\lambda_{j_0+1}+1,\ \lambda_{j_0+2}'=\lambda_{j_0+2}+1,\ \cdots,\ \lambda_N'=\lambda_N+1$. Each column number $j_0$ is indicated by the column involving a gray ellipse area between the two rows.

During the local search phase by Alg.~\ref{alg1}, one can mainly mention the 
sign of the derivative $\frac{\partial Q}{\partial s_i}$ for search the optimal local $m$.
For the symbol $l
\prec l'$ in the column labeled by  $\lambda_i$,  their involved arrays are denoted by $[\lambda]$ and $[\lambda']$, respectively, and the inequality $\frac{\partial Q}{\partial s_i}\Big|_{[\lambda]} <\frac{\partial Q}{\partial s_i}\Big|_{[\lambda']}$ holds true based on the setting $s_1=\ldots=s_{i-1}=e^{\epsilon},\ s_{i+1}=\ldots=s_N=1$.
This indicates that the array $[\lambda']$ achieves the 
non-negative sign on the partial derivatives
no later than array $[\lambda]$ and their locally optimal number of high-weighted candidates $m'\leq m$. 
This claim will be discussed in the proof of Theorem~\ref{thm:BRR}.

For the detailed explicit formulas on global $m$, we collect them regularly in Algorithm~\ref{alg:sequence}.
Due to the inequalities marked by $\prec$ between the partial derivatives at the adjacent elements in the same column,
we derive that the global minimal $m$ can occur only in the row with the extreme or middle priori point. By comparing the two candidates, we prove that when the privacy budget $\epsilon \geq 1$, the first row (from Fig.~\ref{fig3}) for the extreme points, $1$ and $N$ (each as priori point), give the global $m$ for BRR (Lines 2-3). On the other hand, the last row (from Fig.~\ref{fig3}) for the middle priori points, the explicit formulas are presented on Lines 4-5.

Following this, we obtain a series of main results as follows.

First, for the BRR mechanism in one-dimensional equidistant-point domain, the explicit formulas are mainly included in Theorem~\ref{thm:BRR}, and more detailed  formulas  are presented by
Algorithm~\ref{alg:sequence}.

\begin{algorithm}[tb]
  \renewcommand{\algorithmicrequire}{\textbf{Input:}}
  \renewcommand{\algorithmicensure}{\textbf{Output:}}
  \caption{BRR in one-dimensional equidistant point domain}
  \label{alg:sequence}
  \begin{algorithmic}[1]
\REQUIRE 
Privacy budget $\epsilon$, equidistant points count $N$

\STATE $m_1 \gets \Bigg\lfloor \frac{\sqrt{N^2e^\epsilon + \frac{1}{4}(1 - e^\epsilon)^2} - (N - \frac{e^\epsilon}{2} + \frac{1}{2})}{e^\epsilon - 1} \Bigg\rfloor$

\STATE \textbf{if}\ $\epsilon \geq 1$, \textbf{then} $m = m_1$ and go to Line 6

\STATE \textbf{if}\ $N$ is even \textbf{then} $i = \left\lfloor \frac{N}{e^{\frac{\epsilon}{2}} + 1} + 1 \right\rfloor$, \textbf{otherwise}, \\ 
$i = \left\lfloor \frac{\sqrt{e^\epsilon(N^2 - 1) + 1} - N}{e^\epsilon - 1} + 1 \right\rfloor$ 

\STATE \textbf{if}\ $i$ is odd \textbf{then}
        $m_2 = i$,
      \textbf{otherwise}, $m_2 = i + 1$
                
        \STATE $m \gets \min(m_1, m_2)$
    \STATE Define all weights $w_j$ for $j$ = 1 to $N$ by
\begin{align*}
      w_j=
\begin{cases}
\frac{e^\epsilon}{ me^\epsilon +N-m},& \text{if}\ j\,\le\,m\\
\frac{1}{ me^\epsilon +N-m},&  \text{if}\ j\,>\, m
\end{cases}  
\end{align*}
    
  \ENSURE $(w_1,\ldots,w_N)$ and $m$
\end{algorithmic}
\end{algorithm}

\begin{thm}\label{thm:BRR}
Given privacy budget $\epsilon$, let $m$ be the number of items endowed with high probability in BRR mechanism defined by \eqref{def:BRR}.

1) Whatever the selections of $m_1$ and $m_2$, we have $m\ge1$ and
\[
\lim_{N\rightarrow\infty}\frac{m}{N}=\frac{1}{e^{\frac{\epsilon}{2}} +1}.
\]

2) Assume $\epsilon \geq 1$, then we have, for the BRR mechanism, 
\[
m = \Bigg\lfloor \frac{\sqrt{N^2e^\epsilon + \frac{1}{4}(1 - e^\epsilon)^2} - (N - \frac{e^\epsilon}{2} + \frac{1}{2})}{e^\epsilon - 1} \Bigg\rfloor.
\]

3) BRR satisfies $\epsilon$-local differential
privacy.
\end{thm}

\begin{proof}
We begin with the illustration by Fig.~\ref{fig3}.
For each adjacent priori number $k$, the distances for all candidates are sorted in an ascending order.
For any two adjacent rows with priori numbers $k$ and $k'=k-1$, respectively,
we  compare their sorted distance arrays (rows), $[\lambda]$ and $[\lambda']$. 
Then there exists a single $j_0$, such that
$\lambda_1'=\lambda_1,\ \lambda_2'=\lambda_2,\ \cdots,\ \lambda_{j_0}'=\lambda_{j_0}$, while for the remaining $N-j_0$ items, $\lambda_{j_0+1}'=\lambda_{j_0+1}+1,\ \lambda_{j_0+2}'=\lambda_{j_0+2}+1,\ \cdots,\ \lambda_N'=\lambda_N+1$. Such a number $j_0$ is indicated by the column number involving a gray ellipse area between the two rows.

Now we consider to increase the probability weight $s_i$ in order. 
We first consider the case $i\leq j_0$. When considering the weight increase on $s_i$, we have fixed  $s_1=\cdots=s_{i-1}=e^{\epsilon}$ for both rows, and for the numerator part of \eqref{Qloss_partial} about $\frac{\partial Q}{\partial s_i}$,

\[
\begin{split}
\sum_j (\lambda_i' - \lambda_j') s_j
 =& \sum_{j=1}^{j_0} (\lambda_i' - \lambda_j') s_j + \sum_{j=j_0+1}^N (\lambda_i' - \lambda_j') s_j \\
=& \sum_{j=1}^{j_0} (\lambda_i - \lambda_j) s_j + \sum_{j=j_0+1}^N (\lambda_i - (\lambda_j + 1)) s_j \\
=& \sum_j (\lambda_i - \lambda_j) s_j - \sum_{j_0+1}^N s_j 
\leq \sum_j (\lambda_i - \lambda_j) s_j,
\end{split}
\]
where the equality on the last line holds only when $j_0=N$. In fact, $j_0\leq N-1$ holds for any two adjacent row.
This shows the monotonicity relationship within the same column in the lower triangular red region, specifically indicating that the numerator factor of the partial derivative of utility loss for the item below is greater than the corresponding factor for the item above when considering increasing the probability weight.

Next, we consider the case \(i > j_0\), then \(\lambda_i' = \lambda_i + 1\). When \(j \leq j_0\), \(\lambda_j' = \lambda_j\); and when \(j > j_0\), \(\lambda_j' = \lambda_j + 1\). Then, for $i > j_0$, we have
\[
\begin{split}
\sum_j (\lambda_i' - \lambda_j') s_j =& \sum_{j=1}^{j_0} (\lambda_i' - \lambda_j') s_j + \sum_{j=j_0+1}^N (\lambda_i' - \lambda_j') s_j \\
=& \sum_{j=1}^{j_0} (\lambda_i + 1 - \lambda_j) s_j + \sum_{j=j_0+1}^N (\lambda_i + 1 - (\lambda_j + 1)) s_j \\
=& \sum_{j=1}^{j_0} (\lambda_i - \lambda_j) s_j + \sum_{j=j_0+1}^N (\lambda_i - \lambda_j) s_j + \sum_{j=1}^{j_0} s_j \\
=& \sum_j (\lambda_i - \lambda_j) s_j + \sum_{j=1}^{j_0} s_j 
> \sum_j (\lambda_i - \lambda_j) s_j.
\end{split}
\]
This demonstrates the monotonicity relationship within the same column in the upper triangular green region, specifically indicating that the numerator factor of the partial derivative of utility loss for the item below is less than the corresponding factor for the item above.

Based on Fig.~\ref{fig3}, if the threshold value \(i\) for changing the sign of the partial derivative appears before (no more than) \(j_0\) (in the red region), the lower row will always reach the threshold before the upper one for any two adjacent rows. In this case, starting from the midpoint \(N/2\) as the prior, the bottom-most row achieves the smallest \(m\). On the other hand, if the threshold appears after (larger than) \(j_0\) (in the green region), the upper row \(k\) will reach the threshold before the lower \(k'\). In this scenario, starting from the extreme points $1$ or \(N\) as the prior, the top-most row achieves the smallest \(m\). Fig.~\ref{fig3} depicts the situation where \(N\) is even; the case for odd \(N\) is similar and will not be elaborated further.

Afterwards, it can be seen that the minimum $m$ is obtained when the true item is at the extreme point or the middle point. If it is at the extreme point, we have
\begin{equation}
 \begin{split}
  &\sum_{j=1}^{i-1}(\lambda_i-\lambda_j)e^\epsilon + \sum_{j=i+1}^N(\lambda_i-\lambda_j)1 \\
  =& \frac{(1+i-1)(i-1)}{2} e^\epsilon - \frac{(1+N-i)(N-i)}{2} \\
  =& \frac{1}{2}(e^\epsilon - 1)i^2 + \left(N - \frac{1}{2}e^\epsilon + \frac{1}{2}\right)i - \frac{1}{2}N^2 - \frac{1}{2}N,
 \end{split}
\end{equation}
which is a quadratic function in $i$, where $\Delta = N^2e^\epsilon + \frac{1}{4}(1 - e^\epsilon)^2 > 0$. Since the integer $i$ is positive, we get 
\[
m_1 = i = \left\lfloor \frac{\sqrt{N^2e^\epsilon + \frac{1}{4}(1 - e^\epsilon)^2} - (N - \frac{e^\epsilon}{2} + \frac{1}{2})}{e^\epsilon - 1} \right\rfloor.
\]

Since the arrangement of the median depends on $N$, if the minimum $m$ is obtained at the middle point, it is necessary to consider the two different cases when $N$ is odd or even. In the median, there are many paired values, making the paired $i$ always odd. Here, only the case where $i$ is odd needs to be considered.

For even $N$:
\begin{equation}
 \begin{split}
   &\sum_{j=1}^{i-1}(\lambda_i-\lambda_j)e^\epsilon + \sum_{j=i+1}^N(\lambda_i-\lambda_j)\times 1 \\
   =& \frac{(e^\epsilon - 1)}{4}i^2 + \frac{N - (e^\epsilon - 1)}{2}i + \frac{(e^\epsilon - 1)-N^2-2N}{4}, 
 \end{split}
\end{equation}
where $\Delta = \frac{e^\epsilon}{4}N^2$. We have $i = \left\lfloor \frac{N}{e^{\frac{\epsilon}{2}} + 1} + 1 \right\rfloor$.

For odd $N$:
\begin{equation}
 \begin{split}
 &\sum_{j=1}^{i-1}(\lambda_i-\lambda_j)e^\epsilon + \sum_{j=i+1}^N(\lambda_i-\lambda_j)\times 1 \\
 =& \frac{(e^\epsilon - 1)}{4}i^2 + \frac{N - (e^\epsilon - 1)}{2}i + \frac{(e^\epsilon - 1)-2N-N^2+1}{4},
 \end{split}
\end{equation}
where $\Delta = \frac{e^\epsilon}{4}(N^2 - 1) + \frac{1}{4}$. We have $i = \left\lfloor \frac{\sqrt{e^\epsilon(N^2 - 1) + 1} - N}{e^\epsilon - 1} + 1 \right\rfloor$. If $i$ is odd, then $m_2 = i$; if $i$ is even, then $m_2 = i + 1$.

So we just need to compare the $m$ values obtained at the extreme point and the middle point, and take the smaller one as the final $m = \min\{m_1, m_2\}$.

For ease of discussion, let 
\[z = \sum_{j=1}^{i-1} (\lambda_i - \lambda_j) e^\epsilon + \sum_{j=i+1}^N (\lambda_i - \lambda_j) 1.\]
If the minimum $m$ is obtained at the extreme point, we have 
\[z_1(i) = \frac{(e^\epsilon - 1)}{2} i^2 + \left( N - \frac{e^\epsilon}{2} + \frac{1}{2} \right) i - \frac{N^2}{2} - \frac{N}{2}.\] 
If the minimum $m$ is obtained at the middle point, then for even $N$, we have \[z_2(i) = \frac{(e^\epsilon - 1)}{4} i^2 + \frac{N - (e^\epsilon - 1)}{2} i + \frac{(e^\epsilon - 1)}{4} - \frac{N}{2} - \frac{N^2}{4},\] 
and for odd $N$, we have 
\[z_3(i) = \frac{(e^\epsilon - 1)}{4} i^2 + \frac{N - (e^\epsilon - 1)}{2} i + \frac{(e^\epsilon - 1)}{4} - \frac{N}{2} - \frac{N^2}{4} + \frac{1}{4}.\]
It can be observed that $z_3$ can be obtained by shifting $z_2$ upwards by $\frac{1}{4}$ units. Therefore, the left root of $z_3$ is greater than the left root of $z_2$, and the right root of $z_3$ is smaller than the right root of $z_2$. The left root of $z_3$ is $i_{3l} = 1 - \frac{N + \sqrt{e^\epsilon (N^2 - 1) + 1}}{e^\epsilon - 1} < 1$, thus the left root of $z_2$ is also less than $1$.

Considering the selections of $m_1$ and $m_2$, we have
\[
\lim_{N\rightarrow\infty}\frac{m_1}{N}=\lim_{N\rightarrow\infty}\frac{m_2}{N}=\frac{1}{e^{\frac{\epsilon}{2}} +1}.
\]
This gives the claim 1).

To analyze which of $m_1$ or $m_2$ is larger, it is necessary to compare the right root of $z_1$, $i_{1r}$, with the right root of $z_3$ when $N$ is odd, $i_{3r}$.
\begin{equation}
 \begin{split}
 i_{1r} &= \frac{\sqrt{N^2 e^\epsilon + \frac{1}{4} (1 - e^\epsilon)^2} - (N - \frac{e^\epsilon}{2} + \frac{1}{2})}{e^\epsilon - 1} \\
 &< \frac{N e^{\frac{\epsilon}{2}} + \frac{1 - e^\epsilon}{2} - (N - \frac{e^\epsilon}{2} + \frac{1}{2})}{e^\epsilon - 1} = \frac{N e^{\frac{\epsilon}{2}} - N}{e^\epsilon - 1}, \\
 i_{3r} &= \frac{\sqrt{e^\epsilon (N^2 - 1) + 1} - (N - e^\epsilon + 1)}{e^\epsilon - 1} \\
 &> \frac{e^{\frac{\epsilon}{2}} (N - 1) - (N - e^\epsilon + 1)}{e^\epsilon - 1}.
 \end{split}
\end{equation}
Let $a = N e^{\frac{\epsilon}{2}} - N$, $b = e^{\frac{\epsilon}{2}} (N - 1) - (N - e^\epsilon + 1)$, then $b - a = e^\epsilon - e^{\frac{\epsilon}{2}} - 1 = (e^{\frac{\epsilon}{2}} - \frac{1}{2})^2 - \frac{5}{4}$. Setting $b - a > 0$, we solve for $\epsilon > 2 \ln \left( \frac{\sqrt{5} + 1}{2} \right) \approx 0.9624$. Therefore, when $\epsilon \geq 1$, we always have $i_{3r} > i_{1r}$, and $m_1 = \lfloor i_{1r} \rfloor$, $m_2 = \lfloor i_{3r} \rfloor$, hence $m_1 \leq m_2$. This shows that when $\epsilon \geq 1$, the minimum $m$ is obtained at the extreme point. This gives the claim 2).

Given any released candidate $y\in\mathcal{Y}$, for any possible priori items $x_1,\, x_2\in\mathcal{X}$, we have
\[\frac{\Pr[y|x_1]}{\Pr[y|x_2]}\leq \frac{e^\epsilon}{ me^\epsilon +N-m}
\bigg/
\frac{1}{ me^\epsilon +N-m}=e^\epsilon.
\]
This gives the claim 3).
 \end{proof}

Afterwards, for the quality loss of BRR in  equidistant point domain, we confirm theoretically that BRR outperforms GRR and obtain the exact upper and lower asymptotic bounds on the local quality loss ratio of BRR to GRR at any given priori number $k$ as well as the global quality loss ratio in Theorem~\ref{thm:BRR_qloss}.

\begin{thm}\label{thm:BRR_qloss}
For any priori point $k$, the local expected error of BRR is no more than that of GRR; the asymptotic upper and lower bounds for the local expected-error ratio of BRR to GRR are 
\begin{align*}
   \mathop{\lim}\limits_{N \to \infty}\sup_{k}\frac{Q^{(k)}(BRR,N)}{Q^{(k)}(GRR,N)}&=
   \frac{2}{e^{\frac{\epsilon}{2}} +1},\\
\mathop{\lim}\limits_{N \to \infty}\inf_{k}\frac{Q^{(k)}(BRR,N)}{Q^{(k)}(GRR,N)}&=e^{\frac{\epsilon}{2}}-\frac{e^{\frac{\epsilon}{2}} -1}{e^{\frac{\epsilon}{2}} +1}\left(\sqrt{e^{\epsilon} +2e^{\frac{\epsilon}{2}} +2} +1\right), 
\end{align*}
respectively, and their gap is 
\[\text{\rm sup-inf}=\frac{e^{\frac{\epsilon}{2}} -1}{e^{\frac{\epsilon}{2}} +1}\left(\sqrt{e^{\epsilon} +2e^{\frac{\epsilon}{2}} +2} +e^{\frac{\epsilon}{2}} +1\right)^{-1}<\frac12\left(e^{\frac{\epsilon}{2}} +1\right)^{-1}.\]
\end{thm}

\begin{proof}
 The detailed proof is presented in Appendix \ref{app:proof_part2}. 
Motivated by Fig.~\ref{fig3}, the sketch is as follows.
We only need to consider the case where $ N $ is a sufficiently large odd integer. Without loss of generality, let $ m $ also be an odd number satisfying $m \sim \frac{N}{e^{\epsilon/2} + 1}$.
The symbol $\sim$ means the equivalence between two formulas for sufficiently large $N$.

Given any priori position $ k \leq N $ with $ k \in \mathbb{Z}^{+} $, the distances from any of $N$ positions (to be possibly reported) to $ k $ can be arranged in ascending order as the sequence $[0, 1, 1, 2, 2, \ldots, n, n, n+1, n+2, \ldots, N - 1 - n]$, where $0< 2n \leq N - 1 $ and $k=n+1$ or $N-n$. 

We analyze the local expected errors
$Q^{(k)}$  under two protocols. We define the normalized parameters: $m = cN, \,n=dN , \; \text{with} \; c_0 = \frac{1}{e^{\epsilon/2} + 1},\;c < \frac{1}{2} \; \text{and}\; 0\le d < \frac{1}{2}$.

For GRR mechanism, the  $Q^{(k)}(GRR)$ is computed as, 
$$
Q^{(k)}(GRR) \sim \frac{N^2}{N} \left[\frac{(1 - d)^2}{2} + \frac{d^2}{2}\right] = \frac{1}{2}(1 - 2d + 2d^2)N.
$$

For BRR mechanism, we distinguish two cases:

\begin{enumerate}[left=0pt]
    \item When $ m - 1 \leq 2n \leq N - 1 $, the $Q^{(k)}(BRR)$ is expressed as,
    \begin{equation*}
    Q^{(k)}(BRR)/N 
=\frac{1}{4(1 + (e^{\epsilon} - 1)c)} \left[c^2(e^{\epsilon} - 1) + 2(1 - 2d + 2d^2)\right].
    \end{equation*}
    
    \item When $ 0\le 2n <m$, the $Q^{(k)}(BRR)$ becomes,
    \begin{equation*}
    Q^{(k)}(BRR)/N = 
 \frac{1}{2}\cdot\frac{(1-2d+2d^2)e^{\epsilon}+(c-1)(c+1-2d)(e^{\epsilon}-1)}{(e^{\epsilon}-1)c+1}.
    \end{equation*}
\end{enumerate}
Afterward, the discussion on limits and bounds with $N\rightarrow +\infty$ and $c\to c_0$, gives the claims as required.
\end{proof}

The above bounds for ratios on local expected errors also give the upper and lower bounds for the ratio $\frac{Q_g(BRR)}{Q_g(GRR)}$ of global expected error. The gap between the bounds exists obviously due to the pointwise estimate of the local ratios. In order to close this gap, the asymptotically exact ratio is presented as follows.

\begin{thm}\label{thm:global_ratio_qloss}
For sufficiently large count $N$  in the equidistant point domain, 
the asymptotically exact ratio of BRR to GRR for global expected errors is
\begin{equation*}
  \lim\limits_{N \to \infty} \frac{Q_g(BRR,N)}{Q_g(GRR,N)} =\frac{7e^{\frac{\epsilon }{2} } +  9}{4\left( e^{\frac{\epsilon }{2} } +  1 \right)^2}.   
\end{equation*}
\end{thm}

\begin{proof}
 The detailed proof is presented in Appendix \ref{app:proof_part3}. 
The sketch is as follows.

Denote $n= dN$ as before.  Due to the symmetry, we can account for only  the case $k\le (N+1)/2$.
Since $N$ tends to infinity, we can assume, without loss of generality, that  both $N$  and $m = m(N)$ are odd.

We calculate the global expected error of BRR with the dynamic splitting number $m=cN$ as follows.
\begin{equation*}
         N\cdot Q_{g} (BRR)=2\left(\sum_{k=1}^{\frac{m+1}{2} } Q^{(k)} (BRR)+\sum_{k=\frac{m+3}{2} }^{\frac{N+1}{2} }  Q^{(k)} (BRR) \right)-Q^{(\frac{N+1}{2})}(BRR).
 \end{equation*}

Simplifying the above summation yields the limit of ratios as required when $N$ tends to infinity, by using $\lim_{N\rightarrow\infty}c(N)=c_0 = \frac{1}{e^{\epsilon/2} + 1}$.
\end{proof}

Now we are ready to close the four challenges presented in Section \ref{sect: prob_formu} as follows.

\begin{itemize}
\item In the local search phase, we increase the prob. weights in sequence of utility. When the weight of some point ascends, the (normalized) weighted average  error of reporting the other $N-1$ points is fixed. Then, the monotonicity of the utility function for varying the weight of only one point can easily deduce the bipartite distribution of weight values, only $1$ and $e^\epsilon$.

\item
In the global search phase, we set the minimal $m^{(i)}$ obtained in the local search phase as the uniform number $m$ of high probability weight $e^\epsilon$, which is certainly independent of the priori distribution. Indeed, BRR assgins the minimal $m^{(i)}$ among all priori points, $i=1,\ldots,N$, as the final $m$. This value corresponds to the priori point, starting with which the BRR first achieves its own smallest (local) expected error when simultaneously increasing the number of high probability $e^\epsilon$-weighted candidates for starting priori points. 

\item We present how to compute the uniform number $m$ using point-wise derivative formulas in the general case, as shown by Alg. \ref{alg2}. In the equidistant point domain, the explicit formulas are simply presented by Alg. \ref{alg:sequence}.

\item The impacts are amazing in theoretic and application aspects. Theoretically, the local search algorithm begins with the GRR, and the global setting $m=\min_i m^{(i)}$  ensures that the BRR outperforms  the GRR if $m>1$ and otherwise equals. In the equidistant point domain some significant theoretical results are presented in Theorems \ref{thm:BRR} to \ref{thm:global_ratio_qloss}  as demonstrated in Fig. \ref{fig:bQ}.
In applications to engineering including decision tree training, stochastic gradient descent,  location-based services, and  vector perturbation, a series of experimental results are shown in Section \ref{sect:exper}.

\end{itemize}

\section{Exploring the scenarios of BRR}
In this section, we explore the scenario of BRR mechanism for various applications.

\subsection{Adapting BRR in Continuous Domain}


We first consider the scalar case of continuous domain.
That is, the perturbation interval defined as the scalar range of query function is equally divided, which presents a discrete domain consisting of \( N \) discrete nodes,
\[
\mathcal{Y} = \{ y_1, y_2, \dots, y_N \}.
\]
Then, as shown in Algorithm~\ref{alg:continu_domain}, the actual query value \( x \) is firstly mapped to the closest discrete value \( y^\ast \in \mathcal{Y} \):
\begin{equation}\label{def:y_ast}
y^\ast = \arg\mathop{\min}_{y_i \in \mathcal{Y}} |x - y_i|.
\end{equation}

Afterwards, we make a bijection mapping the $N$ equidistant points to their labels $1,2,\ldots, N$ in order. The linear transformation helps to present the BRR mechanism on the domain of $\mathcal{Y}$. 
  The \( BRR \)  perturbs the approximate value \( y^\ast \) to generate a reporting value $y_j$. Since the \( BRR \) randomized response mechanism provides LDP guarantees, this method can flexibly improve the discretization accuracy by increasing the number of nodes while enforcing \( \epsilon \)-LDP.

\begin{algorithm}[tb]
  \renewcommand{\algorithmicrequire}{\textbf{Input:}}
  \renewcommand{\algorithmicensure}{\textbf{Output:}}
  \caption{BRR in one-dimensional interval range}
  \label{alg:continu_domain}
  \begin{algorithmic}[1]
\REQUIRE 
Privacy budget $\epsilon$, interval $[a,b]$, equidistant points count $N$, actual query value $x$

\STATE Generate set $\mathcal{Y}$ consisting of $y_i=a+\frac{j-1}{N-1}(b-a)$ for $j=1,\ldots,N$

\STATE Find the point $y^\ast$ closest to $x$ by \eqref{def:y_ast}


\STATE Confirm the number $m$ of high-weighted points by Algrithm  \ref{alg:sequence}
\STATE Construct the set $Y_m$ by the $m$ closest  points to $y^\ast$ in $\mathcal{Y}$
\STATE Generate a reported value $y_j$ in $\mathcal{Y}$ with the distribution by \eqref{def:BRR}
\ENSURE $m$ and reported value $y_j$
\end{algorithmic}
\end{algorithm}

The input of BRR already contains a certain level of noise during the mapping to discrete points. Specifically, since the original query value \( x \) is mapped to the closest point \( y^\ast \) in the discrete set \( \mathcal{Y} \), this mapping process introduces a deviation from the original \( x \), causing data distortion. Therefore, to ensure the approximation accuracy, we can reduce the mapping error by increasing the number of discrete points \( N \). As \( N \) increases, the smaller intervals \( L \) among the  discrete points become finer, making the mapped value \( y^\ast \) closer to the original data \( x \), thus maintaining higher data accuracy. 

However, we should note that increasing \( N \), while reducing mapping errors, also increases the number of discrete points that need to be handled in the subsequent \( BRR\) or \( GRR\) mechanism, thus expanding the spatial domain. This can lead to a sparser probability distribution in the release, potentially affecting the accuracy of the perturbation process. Therefore, although increasing \( N \) can improve the approximation of the mapping, it may destroy data utility. This violation really takes place for GRR but not for BRR due to its dynamic number selection of high probability candidates. This will be demonstrated in Section  \ref{sect:exper}.

For the vector case of query range, 
we can also develop the discrete domain in each dimension of the query range. More discussions will be demonstrated in Appendix~\ref{sect:vector-lap}.


\begin{figure*}[htbp]
\centering
\subfigure{
 \begin{minipage}{0.18\linewidth}
 \includegraphics[width=1\textwidth,height=3.5cm]{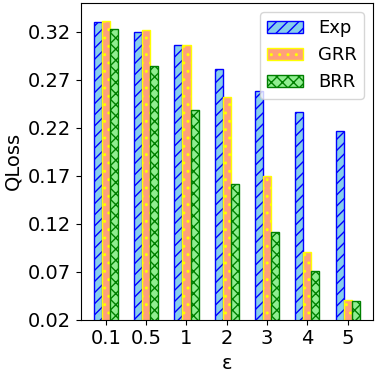} 
 \centerline{\small(a)\quad$N=20$ }
 \end{minipage}
}
\subfigure{
 \begin{minipage}{0.18\linewidth}
 \includegraphics[width=1\textwidth,height=3.5cm]{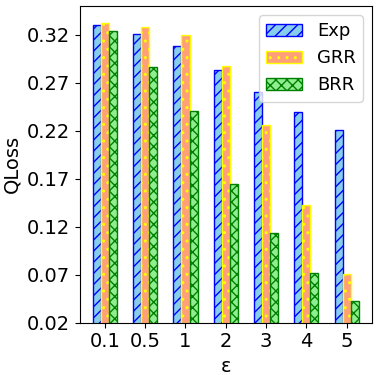} 
 \centerline{ \small(b)\quad$N=40$}
 \end{minipage}
}
\subfigure{
 \begin{minipage}{0.18\linewidth}
 \includegraphics[width=1\textwidth,height=3.5cm]{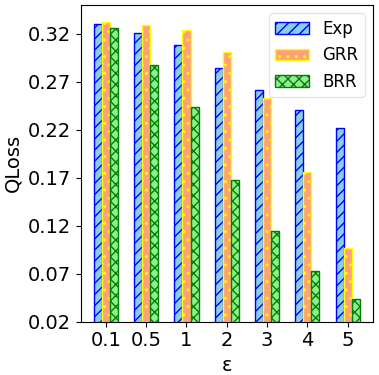} 
 \centerline{ \small(c)\quad$N=60$}
 \end{minipage}
}
\subfigure{
 \begin{minipage}{0.18\linewidth}
 \includegraphics[width=1\textwidth,height=3.5cm]{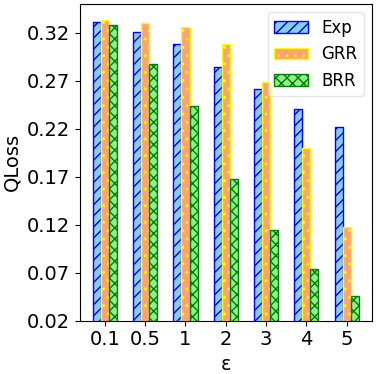} 
 \centerline{ \small(d)\quad$N=80$}
 \end{minipage}
}
\subfigure{
 \begin{minipage}{0.18\linewidth}
 \includegraphics[width=1\textwidth,height=3.5cm]{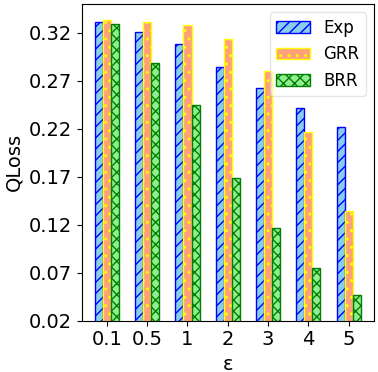} 
 \centerline{ \small(e)\quad$N=100$}
 \end{minipage}
}
\caption{Quality loss comparison of BRR, GRR, and Exponential mechanisms using Euclidean distance with varying privacy budgets and count of candidates.}
\label{fig:qloss_among}
\end{figure*}

\subsection{BRR in General Discrete Domain}
In this section, we consider to explore the scenario of BRR mechanism from two aspects: the distribution of points in the reporting domain, and the metric of data utility.


For stochastic distribution in the multi-dimensional discrete domain, or general metrics of data utility, one can perform the search in terms of Algorithm~\ref{alg2} with necessary adjustments on the ranking of $\lambda_i$ and the sign condition of partial derivatives of utility. Specifically, for any apriori point $x_k$ in the domain, the point $x_k$ itself has the highest utility $\lambda_1^{(k)}$ and the non-increasing sequence $\lambda_1^{(k)} \ge \lambda_2^{(k)} \ge \cdots \ge \lambda_N^{(k)}$, which presents the input of utility matrix $\{\lambda_j^{(k)}\}_{N\times N}$. Following this, we observe the partial derivative of utility and execute (Line 6):

\textbf{if} $\frac{\partial Q^{(k)}}{\partial s_i^{(k)}}>0$, \textbf{then}
        $s_i^{(k)} = e^\epsilon$, $m^{(k)}=i$,
      \textbf{otherwise}, break.

Due to the intrinsic design of BRR, we can make the claim on the comparison to GRR as follows.

\begin{thm}\label{thm:BRR_general}
Given general discrete domain and utility function, for any priori point in the  domain,  BRR at least equals GRR on the local utility. In addition, BRR outperforms GRR on global utility if $m>1$.
\end{thm}

We want to mention that in the two-dimensional domain case,
the BRR mechanism seems to be a simplification of the L-SRR mechanism \cite{WHXQH2022}, where instead of having multiple steps in the staircase, there are only two steps.
That is not the case. In fact, as claimed above,  developing BRR is  motivated by promoting GRR. The innovation  starts exactly from the local optimal search as shown in Algorithm~\ref{alg1}. We prove strictly that the two steps of probability weights give the optimal utility in the local search phase. The global search phase assigns the minimum $m^{(k)}$ as the uniform number $m$ independent of prior distribution. 
Our proposed BRR is fit for LDP applications in any discrete domain as optimal promotion of GRR  while L-SRR mechanism can only be used in two-dimensional discrete domains due to its involved hierarchical encoding scheme.

\section{Experimental Evaluations}\label{sect:exper}


In this section, the experimental evaluations consist of three parts: 
\begin{itemize}
\item
\textbf{Evaluations in discrete domains} (Sections 5.1 to 5.3).
We first consider the basic scenario of equidistant point domains as $\{1,2,\cdots,N\}$ for the alternative to GRR~\cite{KOV2016} and exponential mechanism ~\cite{MT2007} by BRR, assigning Euclidean distance and Jaccard similarity as utility, respectively. Afterwards, we illustrate the effect of BRR in 2-dimensional discrete domains for LBSs by comparing with GRR and the state-of-the-art L-SRR \cite{WHXQH2022}.

\item \textbf{Applications in continuous domains for machine learning}  (Appendix~\ref{append:GBDT} to \ref{sect:vector-lap}).
We evaluate the effect of applying BRR to machine learning for centralized DP, replacing Laplace mechanism in private GBDT~\cite{LWWH2020}, SGD~\cite{ACGMMTZ2016}, and vector perturbation for DNN, cf.~\cite{WGMJ2020}.

\item \textbf{Runtime}  (Appendix~\ref{append:runtime}).
We evaluate the runtime applying BRR for the server with varying the number of users.
\end{itemize}

These evaluations would confirm the robustness of BRR in a series of settings, particularly across multiple branches in machine learning, and illustrate that the effects of BRR are extensive and overwhelming in various scenarios of tasks and applications.



\subsection{Quality loss comparisons to GRR and Exponential mechanisms}
As we know, the data utility is negatively related to the distance, like Euclidean and Jaccard distances.
Following Section 3, we define Euclidean distance as the data  utility (loss) $\lambda$ here. Then we conduct the experiments in the equidistant point domain $\mathcal{X}=\{1,2,\cdots,N\}$ with varying $N$.
The standard metrics for comparisons among BRR, GRR and Exponential mechanisms include $Q_g$ ratios and $QLoss$, where
\begin{equation}\label{QLoss}
 QLoss= \frac{Q_g}{N-1}= \frac{{\textstyle \sum_{k=1}^{N}}  {\textstyle \sum_{i=1}^{N}} \lambda _{i}^{(k)}w _{i}  }{N(N-1)},   
\end{equation}
with the diameter normalization of domain $\mathcal{X}$.

\begin{figure}[tb]
\centering
\subfigure{
 \begin{minipage}{0.45\linewidth}
 \includegraphics[width=1.00\textwidth,trim=0.4cm 0.5cm 0.35cm 0.5cm, clip]{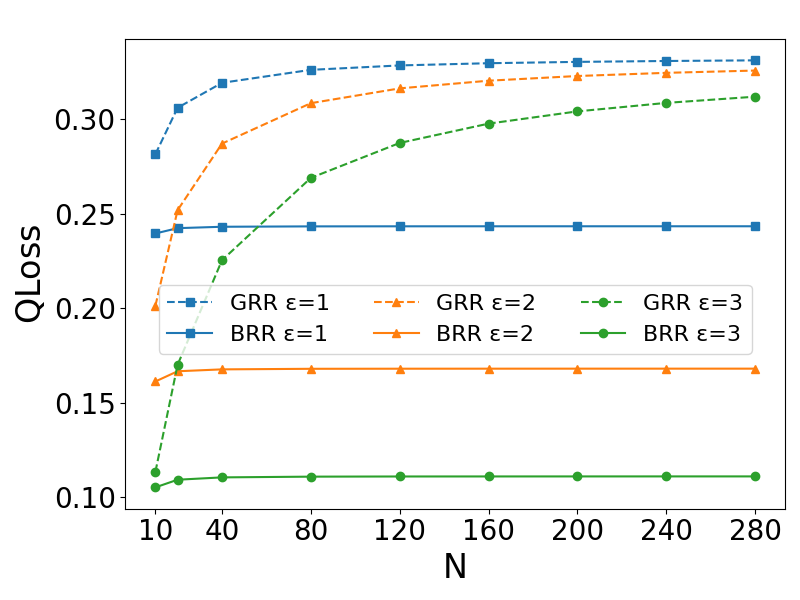} 
 \centerline{\fontsize{9}{10}\selectfont(a) Quality loss $(QLoss)$ 
 }
 \end{minipage}
}
\subfigure{
 \begin{minipage}{0.45\linewidth}
 \includegraphics[width=1.02\textwidth,trim=0.4cm 0.5cm 0.35cm 0.5cm, clip]{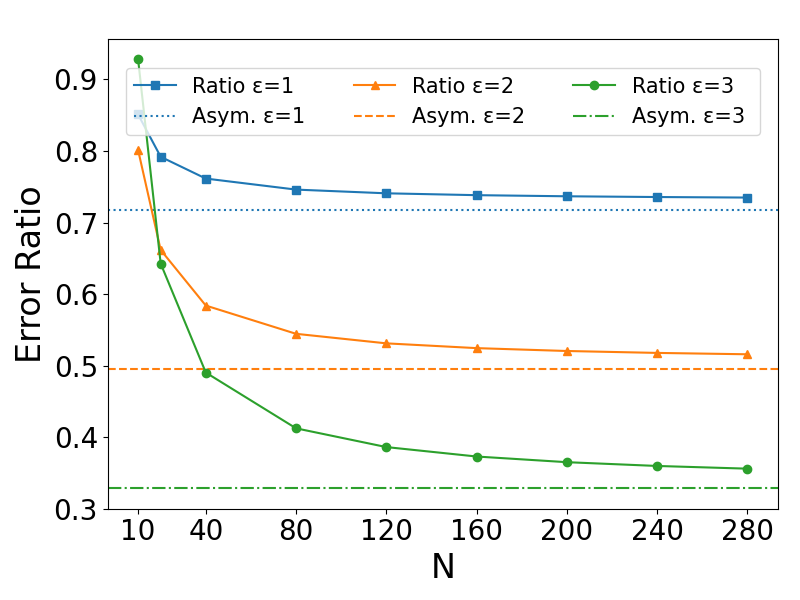} 
 \centerline{\fontsize{9}{10}(b)   $Q_g(BRR)$ / $Q_g(GRR)$}
 \end{minipage}
}
\caption{\fontsize{9}{10} Comparison of quality loss in the equidistant point domain: (a) Quality loss $(QLoss)$ 
of BRR and GRR mechanisms with varying total number $N$ of priori points.  (b) Global expected-error ratio of BRR to GRR as $N$ varies. 
The asymptotically exact ratios (``Asym.'') 
are plotted, cf. Theorem \ref{thm:global_ratio_qloss}.}
\label{fig:bQ}
\end{figure}

Fig.~\ref{fig:qloss_among} shows that, under the same privacy budget, the quality loss of BRR 
is consistently lower than that of the GRR and exponential mechanisms. 
This indicates that the BRR mechanism releases more effective candidates with higher probability given the privacy budget. 
We also observe, from Fig.~\ref{fig:bQ}(a), that as $N$ increases, the gap in utility loss under the GRR mechanism grows significantly, whereas the utility loss of the BRR and exponential mechanisms remains relatively stable across different values of $N$. The reason is that only the true value is assigned a high weight of $e^{\epsilon}$ in the GRR mechanism while the remaining $N-1$ values are assigned a weight of $1$, leading to increasing sparsity and dispersion of the probability distribution as $N$ grows. 
In contrast, the BRR mechanism assigns the high weight $e^{\epsilon}$ to $m$ values, and $m$ increases with a asymptotic proportion to $N$, thus maintaining a more stable probability of effective responses, as shown in Fig.~\ref{fig:bQ}(a). The exponential mechanism exhibits similar behavior, effectively mitigating the impact of distribution sparsity.

Furthermore, as demonstrated by Fig.~\ref{fig:bQ}(b), when $N$ tends to infinity, the asymptotically exact  quality-loss (and global expected-error) ratio of BRR to GRR verifies the assertions given by Theorem~\ref{thm:global_ratio_qloss} and decreases quickly as the privacy budget grows, e.g., approximately $0.5$ for $\epsilon=2$.
The superior performance of the BRR mechanism in maintaining data utility indicates its efficacy and potential for practical applications in privacy-preserving data releases.

\begin{figure*}[htbp]
\centering
\subfigure{
 \begin{minipage}{0.18\linewidth}
 \includegraphics[width=1\textwidth,height=3.5cm]{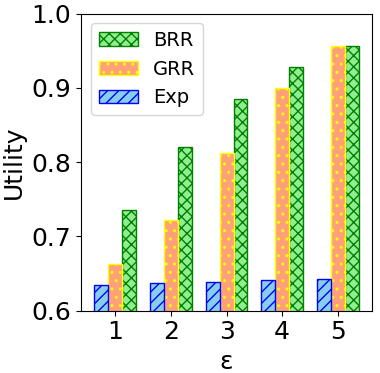} 
 \centerline{\small(a)\quad$N=20$ }
 \end{minipage}
}
\subfigure{
 \begin{minipage}{0.18\linewidth}
 \includegraphics[width=1\textwidth,height=3.5cm]{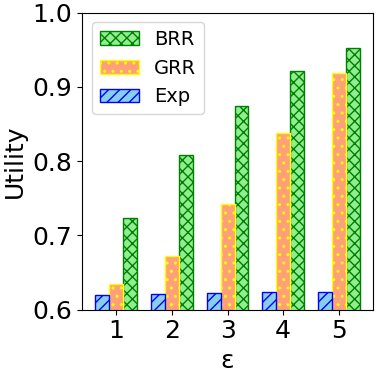} 
 \centerline{ \small(b)\quad$N=40$}
 \end{minipage}
}
\subfigure{
 \begin{minipage}{0.18\linewidth}
 \includegraphics[width=1\textwidth,height=3.5cm]{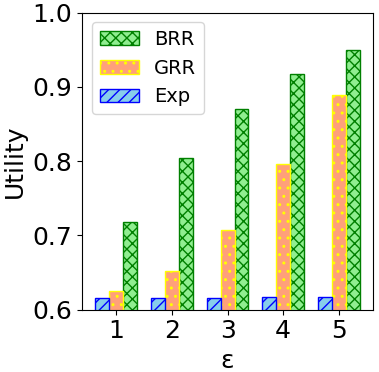} 
 \centerline{ \small(c)\quad$N=60$}
 \end{minipage}
}
\subfigure{
 \begin{minipage}{0.18\linewidth}
 \includegraphics[width=1\textwidth,height=3.5cm]{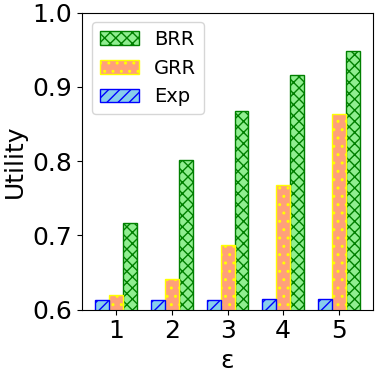} 
 \centerline{ \small(d)\quad$N=80$}
 \end{minipage}
}
\subfigure{
 \begin{minipage}{0.18\linewidth}
 \includegraphics[width=1\textwidth,height=3.5cm]{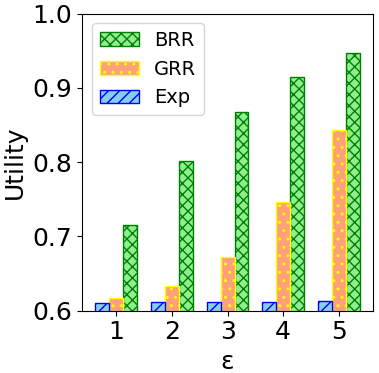} 
 \centerline{ \small(e)\quad$N=100$}
 \end{minipage}
}
\caption{Comparisons of BRR to GRR and Exponential mechanisms for Jaccard-similarity defined utility with varying privacy budgets and count of candidates.}
\label{fig4}
\end{figure*}

As shown in Fig.~\ref{fig:bQ}(b), Fig.~\ref{fig:qloss_among} demonstrates that the larger $\epsilon$ gives the smaller $QLoss$ (or $Q_g$) ratio of BRR to GRR that means greater improvements. However, this vigorous trend vanishes in the later tasks and applications. Particularly, for machine learning the obfuscation mechanisms are adopted as a part of the whole machine learning framework with centralized DP.
The larger $\epsilon$ gives the lower privacy level, the accuracy of iterative learning tasks would naturally approach that of non-private learning, and the smaller gaps on performances provide less chances to close. 

\subsection{BRR with Jaccard-Similarity Utility}\label{sect:jaccard}

Following the setting of Section 3,  we conduct the experiments in the equidistant point domain $\mathcal{X}=\{1,2,\cdots,N\}$ with varying $N$.
We denote data utility $\lambda$ by the pairwise generalized Jaccard similarity coefficient (Tinimoto coefficient), for any candidate $y$, given the true point $x$,

\begin{equation}\label{Jacc_sim}
\lambda(x, y) = \frac{xy}{x^2 + y^2 - xy}.
\end{equation}

The global utility $Q_g$ is defined by  \eqref{Q_g} for comparisons among BRR, GRR and Exponential mechanisms. In this scenario, the larger $Q_g$ gives the better performance.
We compare and analyze the performance of three different differential privacy protection mechanisms: BRR, GRR, and the Exponential mechanism, with privacy budgets and candidates number $N$.

From Fig.~\ref{fig4}, it can be observed that, for the same \(N\), the utility increases with growing privacy budgets across all mechanisms. This trend is inherent to the properties of differential privacy. 
When examined vertically, it is evident that the global expected error of the BRR mechanism remains relatively stable across different $N$ for given $\epsilon$ and is much lower than those of the GRR and exponential mechanisms. These findings indicate that the BRR mechanism enhances data utility largely. 
This substantiates our proposed assertions, cf. Theorem \ref{thm:BRR_general}.

\begin{figure}
\centering
\subfigure{
 \begin{minipage}{0.45\linewidth}
 \includegraphics[width=1\textwidth,height=3cm]{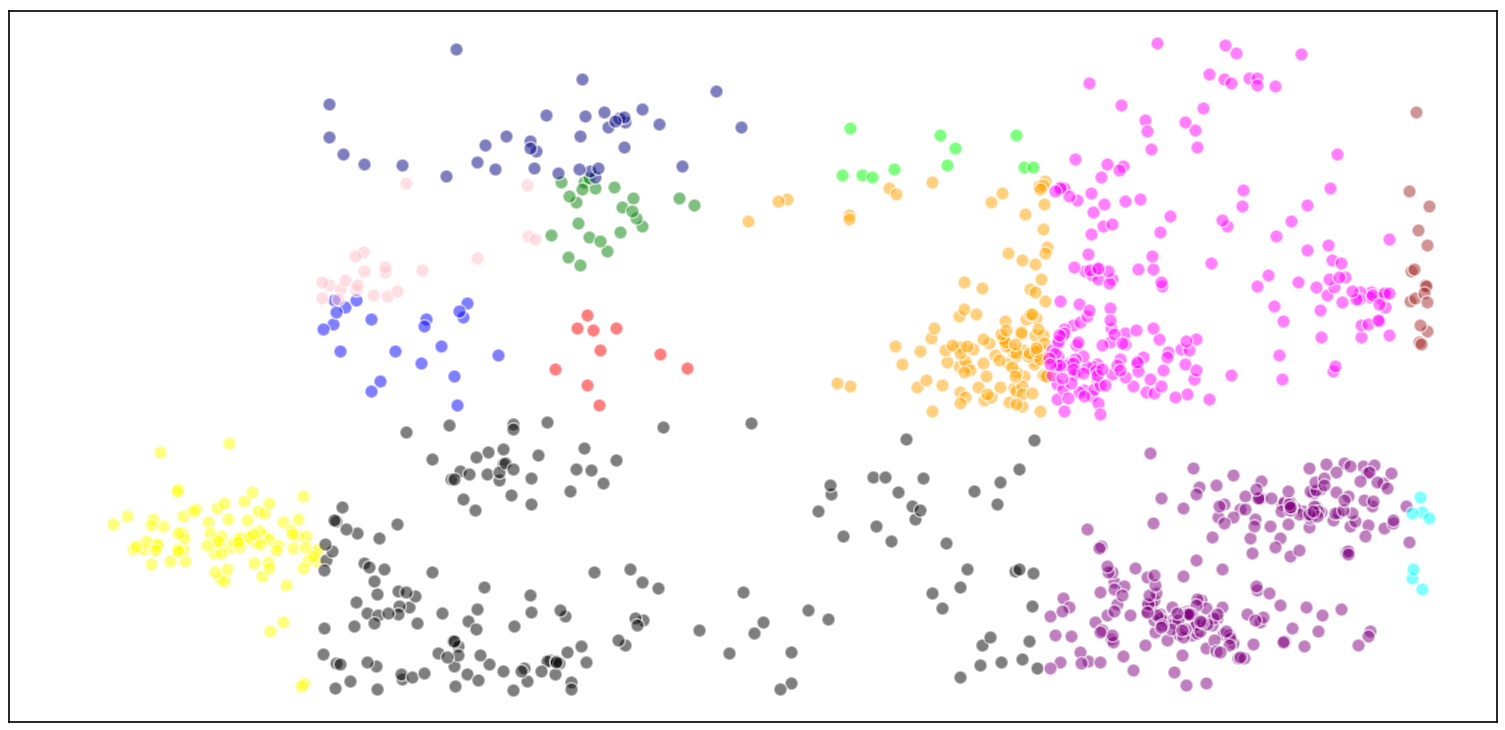}
 \centerline{\small(a)\quad Gowalla, 1000 locations}
 \label{h1}
 \end{minipage}
}
\subfigure{
 \begin{minipage}{0.45\linewidth}
 \includegraphics[width=1\textwidth,height=3cm]{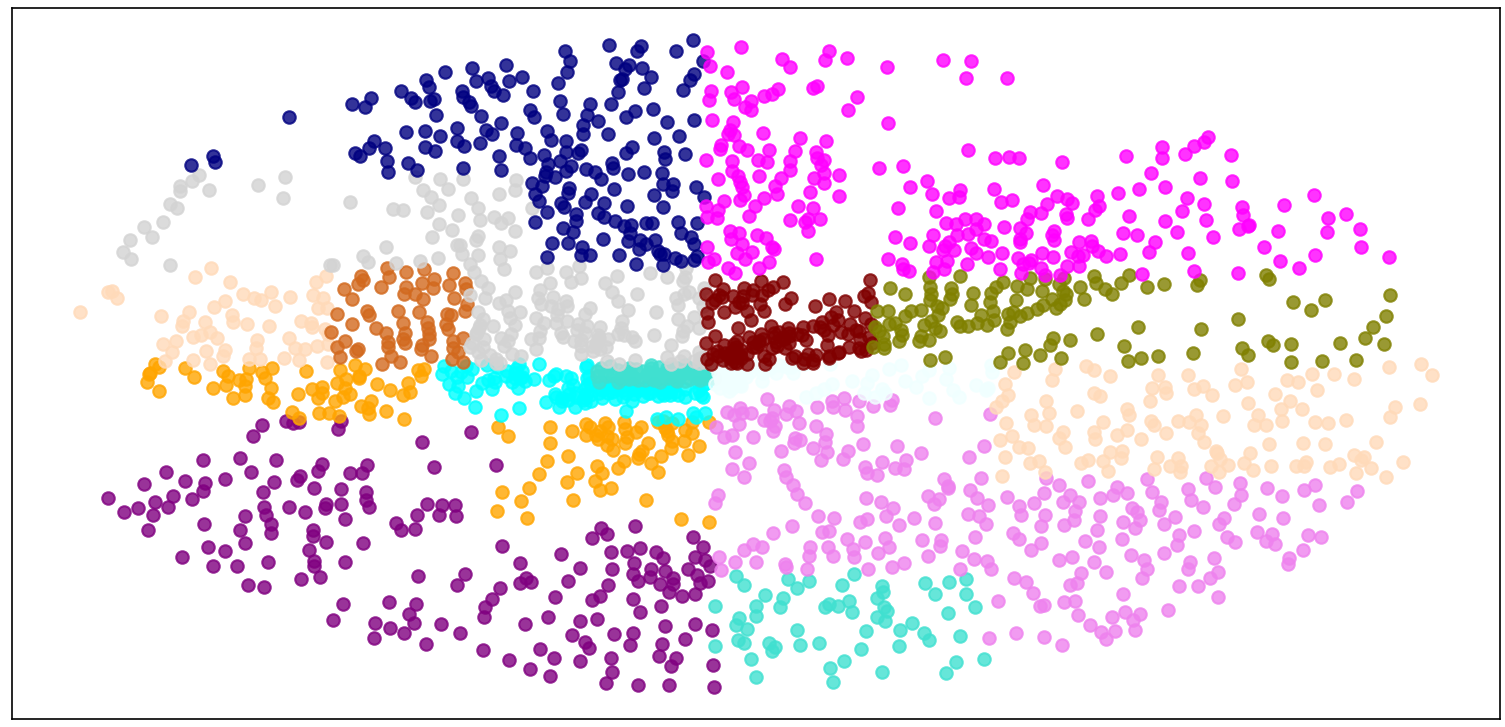} 
 \centerline{\small(b)\quad FourSquare, 3000 locations}
 \label{h2} 
 \end{minipage}
}
\caption{Two dataset maps are divided by the hierarchical
coding approach involved in L-SRR scheme \cite{WHXQH2022} into 13 and 14 regions, respectively, marked in colors. The
 smallest region contains 7 and 64 locations, respectively, while the largest region contains 235 and 278 locations, respectively.
 }
\label{distribution}
\end{figure}

\begin{figure}
\centering
\subfigure{
 \begin{minipage}{0.45\linewidth}
 \includegraphics[width=1\textwidth,height=3.8cm]{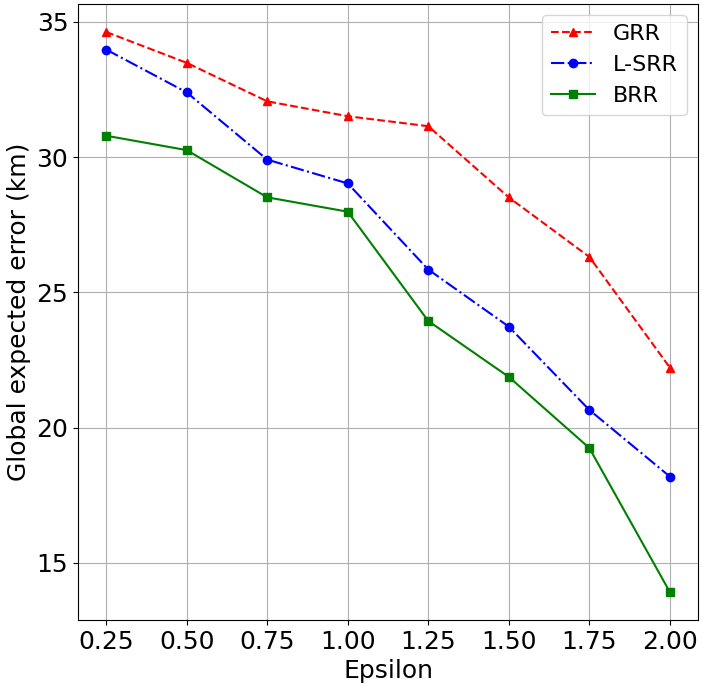}
 \centerline{\small(a)\quad Gowalla}
 \label{h5}
 \end{minipage}
}
\subfigure{
 \begin{minipage}{0.45\linewidth}
 \includegraphics[width=1\textwidth,height=3.8cm]{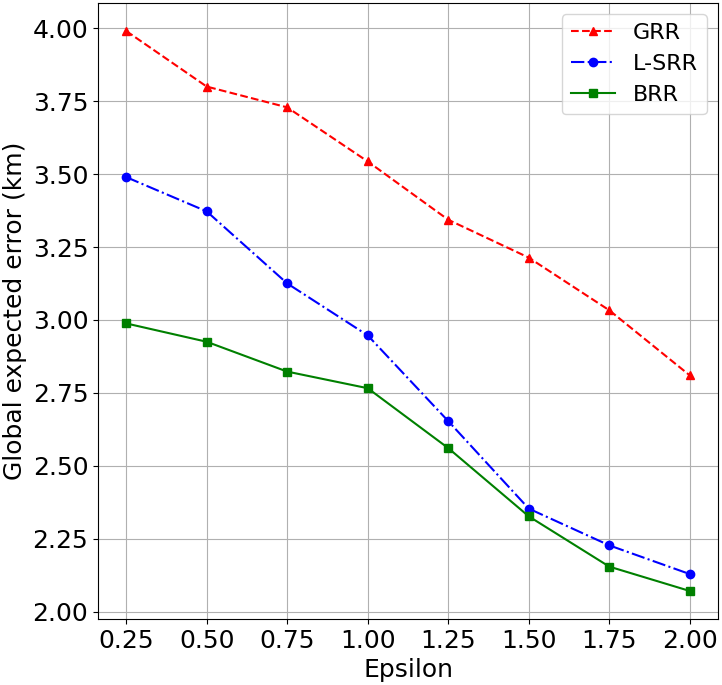} 
 \centerline{\small(b)\quad FourSquare}
 \label{h6} 
 \end{minipage}
}
\vspace{-10pt}
\caption{Global
expected errors of L-SRR \cite{WHXQH2022}, GRR \cite{KOV2016}, and BRR with varying  privacy budgets.}
\label{BRR in LBS}
\end{figure}

\subsection{BRR in Location-based Services}
In location-based services (LBS), ensuring user privacy while maintaining high utility is a significant challenge. The L-SRR framework proposed by Wang et al. ~\cite{WHXQH2022} employs a staircase randomized response mechanism to enhance utility  while achieving provable LDP. In this part, we make evaluations of BRR mechanism in the LBS scenario with comparison to the state-of-the-art L-SRR.

We conduct experiments on Gowalla\footnote{\url{https://snap.stanford.edu/data/loc-Gowalla.html}} and FourSquare~\cite{foursquare} datasets. We extract randomly 1,000 locations from the Gowalla dataset and 3,000 locations from the FourSquare dataset. 
To ensure consistency and facilitate comparison, following the hierarchical coding approach of the L-SRR scheme, we partition the datasets into regions based on common prefixes as shown in Fig.~\ref{distribution}. 
These three obfuscation mechanisms share the same regional partitioning for release. The release range originating from a prior location is identical to its located region, which prevents cross-region releases.

We compare the performance of three schemes (L-SRR, GRR, and BRR) under different privacy budgets ($\epsilon$ ranging from $0.25$ to $2$). As shown in Fig.~\ref{BRR in LBS}, we can observe that  BRR  performs the best across all privacy budgets. On the Gowalla dataset, the average global
expected error $Q_g$ of the BRR is reduced by $8.7\%$ than that of L-SRR, with the most significant reduce up to $23.4\%$ when $\epsilon = 2.0$. Similarly, on the FourSquare dataset, the average $Q_g$ of BRR decreases by $6.7\%$, with the highest improvement by $14.3\%$ at $\epsilon = 0.25$. Overall, when the privacy budget is low, the data utility of all schemes is significantly affected, but BRR still outperforms the other schemes, demonstrating its better noise resistance. As the privacy budget increases, the overall utility loss gradually decreases while BRR maintains relatively optimal accuracy across the entire range.

\section{Related Work}
This section mainly reviews the advancements in Locally Differentially
Private (LDP)  mechanisms and recent applications in machine learning, location-based Services and some other scenarios.

\subsection{Advancements in LDP Mechanisms}
Dwork et al. ~\cite{DMN2006} first systematically introduced the concept of differential privacy, theorizing the obfuscation mechanism based on sensitivity. However, some data providers do not trust the data collectors and analysts. This induces researchers to study the setting of local privacy. Local privacy dates back to Warner \cite{W1965}, who
designed in 1965 the randomized response method for individuals
responding to sensitive surveys. Kasiviswanathan et al. \cite{KLNRS11} first formalized the idea of local privacy, and later Duchi et al. ~\cite{DJW2013} studied the privacy-accuracy trade-offs in data tasks such as mean estimation, establishing a theoretical foundation for the application of LDP. These works rendered LDP prominent in the academic circle.
Holohan et al. ~\cite{HlM2017} provided an optimal mechanism for Warner's original RR technique, achieving the best application of differential privacy in randomized response surveys. 
 
With advances in technology, the $k$-ary randomized response ~\cite{KBR2016} ($k$-RR, also known as GRR) extended traditional binary response to satisfy differential privacy in multi-value data scenarios. 
Building on GRR, Arcolezi et al. ~\cite{AJBX2024} introduced a series of enhanced LDP protocols, such as L-GRR, OUE (Optimized Unary Encoding), and SUE (Symmetric Unary Encoding), further improving the applicability of privacy protection in longitudinal data collection. Makhlouf et al. ~\cite{KHSB2024} discussed the fairness issues of LDP under multiple sensitive attributes.
Chen et al. ~\cite{CLQKL2016} delved into the utility and complexity of personalized local differential privacy (PLDP), expanding the applicability of personalized privacy protection. Song et al. ~\cite{SLWL2020} further proposed the personalized randomized response (PRR) mechanism, enabling LDP to adapt more flexibly to individual differences by setting different privacy requirements for distinct data values. Wei et al. ~\cite{WBXYD2024} proposed the advanced adaptive additive (AAA) mechanism, which is a distribution-aware approach that addresses the classical mean estimation problem.  Wang et al. ~\cite{WBLJ2017} proposed a generalizable aggregation framework to optimize protocol parameters through a simplified aggregation algorithm, improving the utility of data collection. Zhang et al. ~\cite{ZWLHC2018} developed the CALM method, which selectively collects subset attributes to reduce noise impact, significantly enhancing the accuracy of marginal statistics in high-dimensional data.

\subsection{Applications of LDP}



\textbf{LDP in Machine Learning.} Through an in-depth analysis of existing studies, Wang et al. ~\cite{WWZW2023} and Chen et al.~\cite{CZZZY2024} summarized the applications and challenges of DP in deep learning and federated learning, respectively, including the trade-offs between where noise is introduced, accuracy, and privacy, and the types of DP noise added. Bai et al. ~\cite{BYYDXG2017} and Fletcher, Islam et al. ~\cite{FI2019} combined differential privacy with decision tree models, proposing a multilayer DP model based on Markov Chain Monte Carlo to effectively balance privacy protection and model performance. To optimize LDP’s frequency estimation performance, Fang et al. ~\cite{FCLG2023} proposed a convolution-based frequency estimation method that reduces LDP noise through deconvolution, significantly enhancing frequency statistics accuracy. Yang et al. ~\cite{YGZIZK2024} highlighted the importance and feasibility of LDP in protecting user privacy.

 Abadi et al. ~\cite{ACGMMTZ2016} proposed a method for training deep neural networks with differential privacy, providing privacy protection for machine learning on sensitive datasets. Bassily and Smith ~\cite{BS2015} proposed an efficient LDP protocol for accurately counting high-frequency items. Studies by Smith et al. ~\cite{STU2017} and Wang et al. ~\cite{WGX2018} focused on optimizing the interactive rounds and sample complexity of LDP, proposing methods to reduce computational costs and improve accuracy.

 \textbf{LDP in Location-based Services.} Wang et al. ~\cite{WHXQH2022} proposed a staircase randomized
response mechanism named L-SRR that introduced the hierarchical coding scheme. This makes L-SRR suitable only in two-dimensional domain.
Besides, L-SRR presented multiple steps in the
staircase while two steps are optimal in local search phase as proved in Section \ref{sect: local_BRR} of in this paper.
By incorporating expected inference error and geo-indistinguishability,
a regionalized LDP mechanism called DPIVE with personalized utility sensitivities was developed   for LBSs~\cite{ZDC24}. Wu et al. ~\cite{WFSK2020} investigated noise injection to protect privacy in distributed machine learning in multi-data-owner environments.  
Zhu et al.~\cite{ZHX2025} introduced 
termed location-discriminative geo-indistinguishability involving different sensitive levels of location privacy.
Cao et al. ~\cite{CJG2021} investigated data poisoning attacks on LDP protocols and proposed defense methods.  

 \textbf{LDP in other scenarios.} Gursoy et al. ~\cite{GTTWL2019} proposed the Cohesive Local Differential Privacy (CLDP) mechanism, which focuses privacy protection on data collection for smaller groups, improving applicability in network security for small-scale data. In high-dimensional data and IoT contexts, Mahawaga et al. ~\cite{MPBKLCA2020} improved privacy protection in convolutional neural network training with the LATENT mechanism, providing an efficient privacy solution for IoT devices. Arcolezi et al. ~\cite{ACAX2021} studied multidimensional data frequency estimation under LDP and proposed a multidimensional LDP scheme to address the combined privacy budget challenge.

Rachel et al. ~\cite{CKR2021} conducted a user survey to explore users' perceptions and expectations of differential privacy. Zhao et al. ~\cite{ZC2022} pointed out the importance of improving data utility while protecting privacy. Priyanka et al. ~\cite{PMRGE2023} effectively explained the core parameter $\epsilon$ to ordinary users to help them make more informed data-sharing decisions. Kairouz et al. ~\cite{KBR2016} presented new mechanisms, including hashed $k$-RR, which outperform existing mechanisms in utility across all privacy levels. Wang et al. ~\cite{WLJ2019} proposed the PEM protocol, optimizing utility and computational complexity by grouping users and reporting value prefixes. 
Murakami and Kawamoto ~\cite{MK2019} introduced Utility-optimized LDP to improve utility in contexts containing large amounts of non-sensitive data. 

To solve private release problems of multiple candidates, GRR mechanism has been the dominant method in recent years. However, GRR ignores the different degrees of utility for candidates and is not suitable for the cases of numerous candidates and high privacy levels requirements. This motivates us to design an adaptive
Bipartite Randomized Response  mechanism called BRR in this paper which significantly improve utility as desired.

\section{Conclusion}
This paper proposes a novel adaptive LDP mechanism, the Bipartite Randomized Response (BRR), which aims to balance the trade-off between privacy and data utility. 
In particular, the GRR mechanism treats all non-true items equally, 
which weakens the data utility in some scenarios.
To address this limitation, BRR optimizes the allocation of release probabilities by local and global search phases. 
This optimization enhances data utility while preserving privacy level. 
 We provide a series of theoretical results
 for the optimal number $m$ of high-weighted items and asymptotically exact ratio of BRR to GRR or global expected errors in the equidistant point domains.
 various empirical valuations verify that our BRR approach outperforms the state-of-the-art methods.
 


For the future work, we will investigate some reformed Bipartite Random Response mechanisms for minimizing the global (unconditional) expected inference error in general discrete domains under uniform and prior distributions.


\newpage

\bibliographystyle{ACM-Reference-Format}
\bibliography{refvldb}


\begin{thebibliography}{47}


\ifx \showCODEN    \undefined \def \showCODEN     #1{\unskip}     \fi
\ifx \showDOI      \undefined \def \showDOI       #1{#1}\fi
\ifx \showISBNx    \undefined \def \showISBNx     #1{\unskip}     \fi
\ifx \showISBNxiii \undefined \def \showISBNxiii  #1{\unskip}     \fi
\ifx \showISSN     \undefined \def \showISSN      #1{\unskip}     \fi
\ifx \showLCCN     \undefined \def \showLCCN      #1{\unskip}     \fi
\ifx \shownote     \undefined \def \shownote      #1{#1}          \fi
\ifx \showarticletitle \undefined \def \showarticletitle #1{#1}   \fi
\ifx \showURL      \undefined \def \showURL       {\relax}        \fi
\providecommand\bibfield[2]{#2}
\providecommand\bibinfo[2]{#2}
\providecommand\natexlab[1]{#1}
\providecommand\showeprint[2][]{arXiv:#2}

\bibitem[\protect\citeauthoryear{Abadi, Chu, Goodfellow, McMahan, Mironov, Talwar, and Zhang}{Abadi et~al\mbox{.}}{2016}]%
        {ACGMMTZ2016}
\bibfield{author}{\bibinfo{person}{Martin Abadi}, \bibinfo{person}{Andy Chu}, \bibinfo{person}{Ian Goodfellow}, \bibinfo{person}{H~Brendan McMahan}, \bibinfo{person}{Ilya Mironov}, \bibinfo{person}{Kunal Talwar}, {and} \bibinfo{person}{Li Zhang}.} \bibinfo{year}{2016}\natexlab{}.
\newblock \showarticletitle{{Deep learning with differential privacy}}. In \bibinfo{booktitle}{\emph{Proceedings of the 2016 ACM SIGSAC Conference on Computer and Communications Security}}. \bibinfo{pages}{308--318}.
\newblock


\bibitem[\protect\citeauthoryear{Arcolezi, Couchot, Al~Bouna, and Xiao}{Arcolezi et~al\mbox{.}}{2021}]%
        {ACAX2021}
\bibfield{author}{\bibinfo{person}{H{\'e}ber~H. Arcolezi}, \bibinfo{person}{Jean-Fran{\c{c}}ois Couchot}, \bibinfo{person}{Bechara Al~Bouna}, {and} \bibinfo{person}{Xiaokui Xiao}.} \bibinfo{year}{2021}\natexlab{}.
\newblock \showarticletitle{{Random sampling plus fake data: Multidimensional frequency estimates with local differential privacy}}. In \bibinfo{booktitle}{\emph{Proceedings of the 30th ACM International Conference on Information \& Knowledge Management}}. \bibinfo{pages}{47--57}.
\newblock


\bibitem[\protect\citeauthoryear{Arcolezi, Couchot, Bouna, and Xiao}{Arcolezi et~al\mbox{.}}{2024}]%
        {AJBX2024}
\bibfield{author}{\bibinfo{person}{Héber~H. Arcolezi}, \bibinfo{person}{Jean~François Couchot}, \bibinfo{person}{Bechara~Al Bouna}, {and} \bibinfo{person}{Xiaokui Xiao}.} \bibinfo{year}{2024}\natexlab{}.
\newblock \showarticletitle{Improving the utility of locally differentially private protocols for longitudinal and multidimensional frequency estimates}.
\newblock \bibinfo{journal}{\emph{Digital Communications and Networks}} \bibinfo{volume}{10}, \bibinfo{number}{2} (\bibinfo{year}{2024}), \bibinfo{pages}{369--379}.
\newblock
\showISSN{2352-8648}


\bibitem[\protect\citeauthoryear{Bai, Yao, Yuan, Deng, Xie, and Guan}{Bai et~al\mbox{.}}{2017}]%
        {BYYDXG2017}
\bibfield{author}{\bibinfo{person}{Xuanyu Bai}, \bibinfo{person}{Jianguo Yao}, \bibinfo{person}{Mingxuan Yuan}, \bibinfo{person}{Ke Deng}, \bibinfo{person}{Xike Xie}, {and} \bibinfo{person}{Haibing Guan}.} \bibinfo{year}{2017}\natexlab{}.
\newblock \showarticletitle{Embedding differential privacy in decision tree algorithm with different depths}.
\newblock \bibinfo{journal}{\emph{Science China Information Sciences}}  \bibinfo{volume}{60} (\bibinfo{year}{2017}), \bibinfo{pages}{1--15}.
\newblock


\bibitem[\protect\citeauthoryear{Bassily and Smith}{Bassily and Smith}{2015}]%
        {BS2015}
\bibfield{author}{\bibinfo{person}{Raef Bassily} {and} \bibinfo{person}{Adam Smith}.} \bibinfo{year}{2015}\natexlab{}.
\newblock \showarticletitle{Local, private, efficient protocols for succinct histograms}. In \bibinfo{booktitle}{\emph{Proceedings of the Forty-Seventh Annual ACM Symposium on Theory of Computing}}. \bibinfo{pages}{127--135}.
\newblock


\bibitem[\protect\citeauthoryear{Cao, Jia, and Gong}{Cao et~al\mbox{.}}{2021}]%
        {CJG2021}
\bibfield{author}{\bibinfo{person}{Xiaoyu Cao}, \bibinfo{person}{Jinyuan Jia}, {and} \bibinfo{person}{Neil~Zhenqiang Gong}.} \bibinfo{year}{2021}\natexlab{}.
\newblock \showarticletitle{Data poisoning attacks to local differential privacy protocols}. In \bibinfo{booktitle}{\emph{30th USENIX Security Symposium (USENIX Security 21)}}. \bibinfo{pages}{947--964}.
\newblock


\bibitem[\protect\citeauthoryear{Chen, Zhu, Zhang, Zhou, and Yu}{Chen et~al\mbox{.}}{2023}]%
        {CZZZY2024}
\bibfield{author}{\bibinfo{person}{Huiqiang Chen}, \bibinfo{person}{Tianqing Zhu}, \bibinfo{person}{Tao Zhang}, \bibinfo{person}{Wanlei Zhou}, {and} \bibinfo{person}{Philip~S. Yu}.} \bibinfo{year}{2023}\natexlab{}.
\newblock \showarticletitle{Privacy and Fairness in Federated Learning: On the Perspective of Tradeoff}.
\newblock \bibinfo{journal}{\emph{ACM Comput. Surv.}} \bibinfo{volume}{56}, \bibinfo{number}{2}, Article \bibinfo{articleno}{39} (\bibinfo{year}{2023}), \bibinfo{numpages}{37}~pages.
\newblock
\showISSN{0360-0300}


\bibitem[\protect\citeauthoryear{Chen, Li, Qin, Kasiviswanathan, and Jin}{Chen et~al\mbox{.}}{2016}]%
        {CLQKL2016}
\bibfield{author}{\bibinfo{person}{Rui Chen}, \bibinfo{person}{Haoran Li}, \bibinfo{person}{A.~Kai Qin}, \bibinfo{person}{Shiva~Prasad Kasiviswanathan}, {and} \bibinfo{person}{Hongxia Jin}.} \bibinfo{year}{2016}\natexlab{}.
\newblock \showarticletitle{Private spatial data aggregation in the local setting}. In \bibinfo{booktitle}{\emph{2016 IEEE 32nd International Conference on Data Engineering (ICDE)}}. IEEE, \bibinfo{pages}{289--300}.
\newblock


\bibitem[\protect\citeauthoryear{Cummings, Kaptchuk, and Redmiles}{Cummings et~al\mbox{.}}{2021}]%
        {CKR2021}
\bibfield{author}{\bibinfo{person}{Rachel Cummings}, \bibinfo{person}{Gabriel Kaptchuk}, {and} \bibinfo{person}{Elissa~M. Redmiles}.} \bibinfo{year}{2021}\natexlab{}.
\newblock \showarticletitle{{``I need a better description``: An Investigation Into User Expectations For Differential Privacy}}. In \bibinfo{booktitle}{\emph{Proceedings of the 2021 ACM SIGSAC Conference on Computer and Communications Security}}. \bibinfo{pages}{3037--3052}.
\newblock


\bibitem[\protect\citeauthoryear{Demelius, Kern, and Tr{\"u}gler}{Demelius et~al\mbox{.}}{2025}]%
        {DKT2025}
\bibfield{author}{\bibinfo{person}{Lea Demelius}, \bibinfo{person}{Roman Kern}, {and} \bibinfo{person}{Andreas Tr{\"u}gler}.} \bibinfo{year}{2025}\natexlab{}.
\newblock \showarticletitle{{Recent Advances of Differential Privacy in Centralized Deep Learning: A Systematic Survey}}.
\newblock \bibinfo{journal}{\emph{Comput. Surveys}} \bibinfo{volume}{57}, \bibinfo{number}{6} (\bibinfo{year}{2025}).
\newblock


\bibitem[\protect\citeauthoryear{Duchi, Jordan, and Wainwright}{Duchi et~al\mbox{.}}{2013}]%
        {DJW2013}
\bibfield{author}{\bibinfo{person}{John~C. Duchi}, \bibinfo{person}{Michael~I. Jordan}, {and} \bibinfo{person}{Martin~J. Wainwright}.} \bibinfo{year}{2013}\natexlab{}.
\newblock \showarticletitle{Local privacy and statistical minimax rates}. In \bibinfo{booktitle}{\emph{2013 IEEE 54th Annual Symposium on Foundations of Computer Science}}. IEEE, \bibinfo{pages}{429--438}.
\newblock


\bibitem[\protect\citeauthoryear{Dwork, McSherry, Nissim, and Smith}{Dwork et~al\mbox{.}}{2006}]%
        {DMN2006}
\bibfield{author}{\bibinfo{person}{Cynthia Dwork}, \bibinfo{person}{Frank McSherry}, \bibinfo{person}{Kobbi Nissim}, {and} \bibinfo{person}{Adam Smith}.} \bibinfo{year}{2006}\natexlab{}.
\newblock \showarticletitle{Calibrating noise to sensitivity in private data analysis}. In \bibinfo{booktitle}{\emph{Theory of Cryptography: Third Theory of Cryptography Conference, TCC 2006, New York, NY, USA, March 4-7, 2006. Proceedings 3}}. Springer, \bibinfo{pages}{265--284}.
\newblock


\bibitem[\protect\citeauthoryear{Erlingsson, Pihur, and Korolova}{Erlingsson et~al\mbox{.}}{2014}]%
        {EPK14}
\bibfield{author}{\bibinfo{person}{\'{U}lfar Erlingsson}, \bibinfo{person}{Vasyl Pihur}, {and} \bibinfo{person}{Aleksandra Korolova}.} \bibinfo{year}{2014}\natexlab{}.
\newblock \showarticletitle{RAPPOR: Randomized Aggregatable Privacy-Preserving Ordinal Response}. In \bibinfo{booktitle}{\emph{Proceedings of the 2014 ACM SIGSAC Conference on Computer and Communications Security}} \emph{(\bibinfo{series}{CCS '14})}. \bibinfo{publisher}{Association for Computing Machinery}, \bibinfo{address}{New York, NY, USA}, \bibinfo{pages}{1054–1067}.
\newblock
\showISBNx{9781450329576}


\bibitem[\protect\citeauthoryear{Fang, Chen, Liu, and Gao}{Fang et~al\mbox{.}}{2023}]%
        {FCLG2023}
\bibfield{author}{\bibinfo{person}{Huiyu Fang}, \bibinfo{person}{Liquan Chen}, \bibinfo{person}{Yali Liu}, {and} \bibinfo{person}{Yuan Gao}.} \bibinfo{year}{2023}\natexlab{}.
\newblock \showarticletitle{{Locally Differentially Private Frequency Estimation Based on Convolution Framework}}. In \bibinfo{booktitle}{\emph{2023 IEEE Symposium on Security and Privacy (S\&P)}}. IEEE, \bibinfo{pages}{2208--2222}.
\newblock


\bibitem[\protect\citeauthoryear{Fletcher and Islam}{Fletcher and Islam}{2019}]%
        {FI2019}
\bibfield{author}{\bibinfo{person}{Sam Fletcher} {and} \bibinfo{person}{Md~Zahidul Islam}.} \bibinfo{year}{2019}\natexlab{}.
\newblock \showarticletitle{{Decision Tree Classification with Differential Privacy}}.
\newblock \bibinfo{journal}{\emph{Comput. Surveys}} \bibinfo{volume}{52}, \bibinfo{number}{4} (\bibinfo{year}{2019}), \bibinfo{pages}{1--33}.
\newblock


\bibitem[\protect\citeauthoryear{Gursoy, Tamersoy, Truex, Wei, and Liu}{Gursoy et~al\mbox{.}}{2019}]%
        {GTTWL2019}
\bibfield{author}{\bibinfo{person}{Mehmet~Emre Gursoy}, \bibinfo{person}{Acar Tamersoy}, \bibinfo{person}{Stacey Truex}, \bibinfo{person}{Wenqi Wei}, {and} \bibinfo{person}{Ling Liu}.} \bibinfo{year}{2019}\natexlab{}.
\newblock \showarticletitle{Secure and utility-aware data collection with condensed local differential privacy}.
\newblock \bibinfo{journal}{\emph{IEEE Transactions on Dependable and Secure Computing}} \bibinfo{volume}{18}, \bibinfo{number}{5} (\bibinfo{year}{2019}), \bibinfo{pages}{2365--2378}.
\newblock


\bibitem[\protect\citeauthoryear{Holohan, Leith, and Mason}{Holohan et~al\mbox{.}}{2017}]%
        {HlM2017}
\bibfield{author}{\bibinfo{person}{Naoise Holohan}, \bibinfo{person}{Douglas~J. Leith}, {and} \bibinfo{person}{Oliver Mason}.} \bibinfo{year}{2017}\natexlab{}.
\newblock \showarticletitle{{Optimal Differentially Private Mechanisms for Randomised Response}}.
\newblock \bibinfo{journal}{\emph{IEEE Transactions on Information Forensics and Security}} \bibinfo{volume}{12}, \bibinfo{number}{11} (\bibinfo{year}{2017}), \bibinfo{pages}{2726--2735}.
\newblock


\bibitem[\protect\citeauthoryear{Kairouz, Bonawitz, and Ramage}{Kairouz et~al\mbox{.}}{2016a}]%
        {KBR2016}
\bibfield{author}{\bibinfo{person}{Peter Kairouz}, \bibinfo{person}{Keith Bonawitz}, {and} \bibinfo{person}{Daniel Ramage}.} \bibinfo{year}{2016}\natexlab{a}.
\newblock \showarticletitle{Discrete distribution estimation under local privacy}. In \bibinfo{booktitle}{\emph{International Conference on Machine Learning}}. PMLR, \bibinfo{pages}{2436--2444}.
\newblock


\bibitem[\protect\citeauthoryear{Kairouz, Oh, and Viswanath}{Kairouz et~al\mbox{.}}{2016b}]%
        {KOV2016}
\bibfield{author}{\bibinfo{person}{Peter Kairouz}, \bibinfo{person}{Sewoong Oh}, {and} \bibinfo{person}{Pramod Viswanath}.} \bibinfo{year}{2016}\natexlab{b}.
\newblock \showarticletitle{{Extremal Mechanisms for Local Differential Privacy}}.
\newblock \bibinfo{journal}{\emph{Journal of Machine Learning Research}} \bibinfo{volume}{17}, \bibinfo{number}{17} (\bibinfo{year}{2016}), \bibinfo{pages}{1--51}.
\newblock


\bibitem[\protect\citeauthoryear{Karima, H., Sami, Ghassen, and Catuscia}{Karima et~al\mbox{.}}{2024}]%
        {KHSB2024}
\bibfield{author}{\bibinfo{person}{Makhlouf Karima}, \bibinfo{person}{Arcolezi~Heber H.}, \bibinfo{person}{Zhioua Sami}, \bibinfo{person}{Ben~Brahim Ghassen}, {and} \bibinfo{person}{Palamidessi Catuscia}.} \bibinfo{year}{2024}\natexlab{}.
\newblock \showarticletitle{On the impact of multi-dimensional local differential privacy on fairness}.
\newblock \bibinfo{journal}{\emph{Data Mining and Knowledge Discovery}} \bibinfo{volume}{38}, \bibinfo{number}{4} (\bibinfo{year}{2024}), \bibinfo{pages}{2252--2275}.
\newblock
\showISSN{1384-5810}


\bibitem[\protect\citeauthoryear{Kasiviswanathan, Lee, Nissim, Raskhodnikova, and Smith}{Kasiviswanathan et~al\mbox{.}}{2011}]%
        {KLNRS11}
\bibfield{author}{\bibinfo{person}{Shiva~Prasad Kasiviswanathan}, \bibinfo{person}{Homin~K. Lee}, \bibinfo{person}{Kobbi Nissim}, \bibinfo{person}{Sofya Raskhodnikova}, {and} \bibinfo{person}{Adam~D. Smith}.} \bibinfo{year}{2011}\natexlab{}.
\newblock \showarticletitle{{What Can We Learn Privately?}}
\newblock \bibinfo{journal}{\emph{SIAM J. Comput.}} \bibinfo{volume}{40}, \bibinfo{number}{3} (\bibinfo{year}{2011}), \bibinfo{pages}{793--826}.
\newblock


\bibitem[\protect\citeauthoryear{Li, Lyu, Su, and Yang}{Li et~al\mbox{.}}{2017}]%
        {LLSY2017}
\bibfield{author}{\bibinfo{person}{Ninghui Li}, \bibinfo{person}{Min Lyu}, \bibinfo{person}{Dong Su}, {and} \bibinfo{person}{Weining Yang}.} \bibinfo{year}{2017}\natexlab{}.
\newblock \bibinfo{booktitle}{\emph{Differential Privacy: From Theory to Practice}}.
\newblock \bibinfo{publisher}{Springer}.
\newblock


\bibitem[\protect\citeauthoryear{Li, Wu, Wen, and He}{Li et~al\mbox{.}}{2020}]%
        {LWWH2020}
\bibfield{author}{\bibinfo{person}{Qinbin Li}, \bibinfo{person}{Zhaomin Wu}, \bibinfo{person}{Zeyi Wen}, {and} \bibinfo{person}{Bingsheng He}.} \bibinfo{year}{2020}\natexlab{}.
\newblock \showarticletitle{Privacy-preserving gradient boosting decision trees}. In \bibinfo{booktitle}{\emph{Proceedings of the AAAI Conference on Artificial Intelligence}}, Vol.~\bibinfo{volume}{34}. \bibinfo{pages}{784--791}.
\newblock


\bibitem[\protect\citeauthoryear{Ma, Zhang, Cai, and Yang}{Ma et~al\mbox{.}}{2023}]%
        {MZCY2023}
\bibfield{author}{\bibinfo{person}{Yuheng Ma}, \bibinfo{person}{Han Zhang}, \bibinfo{person}{Yuchao Cai}, {and} \bibinfo{person}{Hanfang Yang}.} \bibinfo{year}{2023}\natexlab{}.
\newblock \showarticletitle{Decision Tree for Locally Private Estimation with Public Data}.
\newblock \bibinfo{journal}{\emph{Advances in Neural Information Processing Systems}}  \bibinfo{volume}{36} (\bibinfo{year}{2023}), \bibinfo{pages}{43676--43705}.
\newblock


\bibitem[\protect\citeauthoryear{Mahawaga~Arachchige, Bertok, Khalil, Liu, Camtepe, and Atiquzzaman}{Mahawaga~Arachchige et~al\mbox{.}}{2020}]%
        {MPBKLCA2020}
\bibfield{author}{\bibinfo{person}{Pathum~Chamikara Mahawaga~Arachchige}, \bibinfo{person}{Peter Bertok}, \bibinfo{person}{Ibrahim Khalil}, \bibinfo{person}{Dongxi Liu}, \bibinfo{person}{Seyit Camtepe}, {and} \bibinfo{person}{Mohammed Atiquzzaman}.} \bibinfo{year}{2020}\natexlab{}.
\newblock \showarticletitle{{Local differential privacy for deep learning}}.
\newblock \bibinfo{journal}{\emph{IEEE Internet of Things Journal}} \bibinfo{volume}{7}, \bibinfo{number}{7} (\bibinfo{year}{2020}), \bibinfo{pages}{5827--5842}.
\newblock


\bibitem[\protect\citeauthoryear{McSherry and Talwar}{McSherry and Talwar}{2007}]%
        {MT2007}
\bibfield{author}{\bibinfo{person}{Frank McSherry} {and} \bibinfo{person}{Kunal Talwar}.} \bibinfo{year}{2007}\natexlab{}.
\newblock \showarticletitle{Mechanism design via differential privacy}. In \bibinfo{booktitle}{\emph{48th Annual IEEE Symposium on Foundations of Computer Science (FOCS'07)}}. IEEE, \bibinfo{pages}{94--103}.
\newblock


\bibitem[\protect\citeauthoryear{Murakami and Kawamoto}{Murakami and Kawamoto}{2019}]%
        {MK2019}
\bibfield{author}{\bibinfo{person}{Takao Murakami} {and} \bibinfo{person}{Yusuke Kawamoto}.} \bibinfo{year}{2019}\natexlab{}.
\newblock \showarticletitle{Utility-Optimized local differential privacy mechanisms for distribution estimation}. In \bibinfo{booktitle}{\emph{28th USENIX Security Symposium (USENIX Security 19)}}. \bibinfo{pages}{1877--1894}.
\newblock


\bibitem[\protect\citeauthoryear{Nanayakkara, Smart, Cummings, Kaptchuk, and Redmiles}{Nanayakkara et~al\mbox{.}}{2023}]%
        {PMRGE2023}
\bibfield{author}{\bibinfo{person}{Priyanka Nanayakkara}, \bibinfo{person}{Mary~Anne Smart}, \bibinfo{person}{Rachel Cummings}, \bibinfo{person}{Gabriel Kaptchuk}, {and} \bibinfo{person}{Elissa~M. Redmiles}.} \bibinfo{year}{2023}\natexlab{}.
\newblock \showarticletitle{What Are the Chances? Explaining the Epsilon Parameter in Differential Privacy}. In \bibinfo{booktitle}{\emph{32nd USENIX Security Symposium (USENIX Security 23)}}. \bibinfo{address}{Anaheim, CA}, \bibinfo{pages}{1613--1630}.
\newblock
\showISBNx{978-1-939133-37-3}


\bibitem[\protect\citeauthoryear{Smith, Thakurta, and Upadhyay}{Smith et~al\mbox{.}}{2017}]%
        {STU2017}
\bibfield{author}{\bibinfo{person}{Adam Smith}, \bibinfo{person}{Abhradeep Thakurta}, {and} \bibinfo{person}{Jalaj Upadhyay}.} \bibinfo{year}{2017}\natexlab{}.
\newblock \showarticletitle{Is interaction necessary for distributed private learning?}. In \bibinfo{booktitle}{\emph{2017 IEEE Symposium on Security and Privacy (S\&P)}}. IEEE, \bibinfo{pages}{58--77}.
\newblock


\bibitem[\protect\citeauthoryear{Song, Luo, Wang, and Li}{Song et~al\mbox{.}}{2020}]%
        {SLWL2020}
\bibfield{author}{\bibinfo{person}{Haina Song}, \bibinfo{person}{Tao Luo}, \bibinfo{person}{Xun Wang}, {and} \bibinfo{person}{Jianfeng Li}.} \bibinfo{year}{2020}\natexlab{}.
\newblock \showarticletitle{{Multiple Sensitive Values-Oriented Personalized Privacy Preservation Based on Randomized Response}}.
\newblock \bibinfo{journal}{\emph{IEEE Transactions on Information Forensics and Security}}  \bibinfo{volume}{15} (\bibinfo{year}{2020}), \bibinfo{pages}{2209--2224}.
\newblock


\bibitem[\protect\citeauthoryear{Tournier and Montjoye}{Tournier and Montjoye}{2022}]%
        {TM22}
\bibfield{author}{\bibinfo{person}{Arnaud~J. Tournier} {and} \bibinfo{person}{Yves-Alexandre~de Montjoye}.} \bibinfo{year}{2022}\natexlab{}.
\newblock \showarticletitle{Expanding the attack surface: Robust profiling attacks threaten the privacy of sparse behavioral data}.
\newblock \bibinfo{journal}{\emph{Science Advances}} \bibinfo{volume}{8}, \bibinfo{number}{33} (\bibinfo{year}{2022}).
\newblock
\newblock
\shownote{{eabl}6464.}


\bibitem[\protect\citeauthoryear{Wang, Gaboardi, and Xu}{Wang et~al\mbox{.}}{2018}]%
        {WGX2018}
\bibfield{author}{\bibinfo{person}{Di Wang}, \bibinfo{person}{Marco Gaboardi}, {and} \bibinfo{person}{Jinhui Xu}.} \bibinfo{year}{2018}\natexlab{}.
\newblock \showarticletitle{Empirical risk minimization in non-interactive local differential privacy revisited}.
\newblock \bibinfo{journal}{\emph{Advances in Neural Information Processing Systems}}  \bibinfo{volume}{31} (\bibinfo{year}{2018}).
\newblock


\bibitem[\protect\citeauthoryear{Wang, Hong, Xiong, Qin, and Hong}{Wang et~al\mbox{.}}{2022}]%
        {WHXQH2022}
\bibfield{author}{\bibinfo{person}{Han Wang}, \bibinfo{person}{Hanbin Hong}, \bibinfo{person}{Li Xiong}, \bibinfo{person}{Zhan Qin}, {and} \bibinfo{person}{Yuan Hong}.} \bibinfo{year}{2022}\natexlab{}.
\newblock \showarticletitle{{L-SRR}: Local Differential Privacy for Location-Based Services with Staircase Randomized Response}. In \bibinfo{booktitle}{\emph{Proceedings of the 2022 ACM SIGSAC Conference on Computer and Communications Security}}. \bibinfo{pages}{2809--2823}.
\newblock


\bibitem[\protect\citeauthoryear{Wang, Blocki, Li, and Jha}{Wang et~al\mbox{.}}{2017}]%
        {WBLJ2017}
\bibfield{author}{\bibinfo{person}{Tianhao Wang}, \bibinfo{person}{Jeremiah Blocki}, \bibinfo{person}{Ninghui Li}, {and} \bibinfo{person}{Somesh Jha}.} \bibinfo{year}{2017}\natexlab{}.
\newblock \showarticletitle{Locally differentially private protocols for frequency estimation}. In \bibinfo{booktitle}{\emph{26th USENIX Security Symposium (USENIX Security 17)}}. \bibinfo{pages}{729--745}.
\newblock


\bibitem[\protect\citeauthoryear{Wang, Li, and Jha}{Wang et~al\mbox{.}}{2019}]%
        {WLJ2019}
\bibfield{author}{\bibinfo{person}{Tianhao Wang}, \bibinfo{person}{Ninghui Li}, {and} \bibinfo{person}{Somesh Jha}.} \bibinfo{year}{2019}\natexlab{}.
\newblock \showarticletitle{Locally differentially private heavy hitter identification}.
\newblock \bibinfo{journal}{\emph{IEEE Transactions on Dependable and Secure Computing}} \bibinfo{volume}{18}, \bibinfo{number}{2} (\bibinfo{year}{2019}), \bibinfo{pages}{982--993}.
\newblock


\bibitem[\protect\citeauthoryear{Wang, Gu, Ma, and Jin}{Wang et~al\mbox{.}}{2020}]%
        {WGMJ2020}
\bibfield{author}{\bibinfo{person}{Yufeng Wang}, \bibinfo{person}{Min Gu}, \bibinfo{person}{Jianhua Ma}, {and} \bibinfo{person}{Qun Jin}.} \bibinfo{year}{2020}\natexlab{}.
\newblock \showarticletitle{DNN-DP: Differential Privacy Enabled Deep Neural Network Learning Framework for Sensitive Crowdsourcing Data}.
\newblock \bibinfo{journal}{\emph{IEEE Transactions on Computational Social Systems}} \bibinfo{volume}{7}, \bibinfo{number}{1} (\bibinfo{year}{2020}), \bibinfo{pages}{215--224}.
\newblock


\bibitem[\protect\citeauthoryear{Wang, Wang, Zhao, and Wang}{Wang et~al\mbox{.}}{2023}]%
        {WWZW2023}
\bibfield{author}{\bibinfo{person}{Yanling Wang}, \bibinfo{person}{Qian Wang}, \bibinfo{person}{Lingchen Zhao}, {and} \bibinfo{person}{Cong Wang}.} \bibinfo{year}{2023}\natexlab{}.
\newblock \showarticletitle{{Differential privacy in deep learning: Privacy and beyond}}.
\newblock \bibinfo{journal}{\emph{Future Generation Computer Systems}}  \bibinfo{volume}{148} (\bibinfo{year}{2023}), \bibinfo{pages}{408--424}.
\newblock


\bibitem[\protect\citeauthoryear{Warner}{Warner}{1965}]%
        {W1965}
\bibfield{author}{\bibinfo{person}{Stanley~L. Warner}.} \bibinfo{year}{1965}\natexlab{}.
\newblock \showarticletitle{{Randomized response: A survey technique for eliminating evasive answer bias}}.
\newblock \bibinfo{journal}{\emph{J. Amer. Statist. Assoc.}} \bibinfo{volume}{60}, \bibinfo{number}{309} (\bibinfo{year}{1965}), \bibinfo{pages}{63--69}.
\newblock


\bibitem[\protect\citeauthoryear{Wei, Bao, Xiao, Yang, and Ding}{Wei et~al\mbox{.}}{2024}]%
        {WBXYD2024}
\bibfield{author}{\bibinfo{person}{Fei Wei}, \bibinfo{person}{Ergute Bao}, \bibinfo{person}{Xiaokui Xiao}, \bibinfo{person}{Yin Yang}, {and} \bibinfo{person}{Bolin Ding}.} \bibinfo{year}{2024}\natexlab{}.
\newblock \showarticletitle{AAA: An adaptive mechanism for locally differentially private mean estimation}.
\newblock \bibinfo{journal}{\emph{Proc. VLDB Endow.}} \bibinfo{volume}{17}, \bibinfo{number}{8} (\bibinfo{date}{Apr.} \bibinfo{year}{2024}), \bibinfo{pages}{1843–1855}.
\newblock
\showISSN{2150-8097}


\bibitem[\protect\citeauthoryear{Wu, Farokhi, Smith, and Kaafar}{Wu et~al\mbox{.}}{2020}]%
        {WFSK2020}
\bibfield{author}{\bibinfo{person}{Nan Wu}, \bibinfo{person}{Farhad Farokhi}, \bibinfo{person}{David Smith}, {and} \bibinfo{person}{Mohamed~Ali Kaafar}.} \bibinfo{year}{2020}\natexlab{}.
\newblock \showarticletitle{The value of collaboration in convex machine learning with differential privacy}. In \bibinfo{booktitle}{\emph{2020 IEEE Symposium on Security and Privacy (S\&P)}}. IEEE, \bibinfo{pages}{304--317}.
\newblock


\bibitem[\protect\citeauthoryear{Yang, Qu, Yang, and Cudre-Mauroux}{Yang et~al\mbox{.}}{2019}]%
        {foursquare}
\bibfield{author}{\bibinfo{person}{Dingqi Yang}, \bibinfo{person}{Bingqing Qu}, \bibinfo{person}{Jie Yang}, {and} \bibinfo{person}{Philippe Cudre-Mauroux}.} \bibinfo{year}{2019}\natexlab{}.
\newblock \showarticletitle{Revisiting User Mobility and Social Relationships in LBSNs: A Hypergraph Embedding Approach}. In \bibinfo{booktitle}{\emph{The World Wide Web Conference}}. \bibinfo{pages}{2147--2157}.
\newblock


\bibitem[\protect\citeauthoryear{Yang, Guo, Zhu, Tjuawinata, Zhao, and Lam}{Yang et~al\mbox{.}}{2024}]%
        {YGZIZK2024}
\bibfield{author}{\bibinfo{person}{Mengmeng Yang}, \bibinfo{person}{Taolin Guo}, \bibinfo{person}{Tianqing Zhu}, \bibinfo{person}{Ivan Tjuawinata}, \bibinfo{person}{Jun Zhao}, {and} \bibinfo{person}{Kwok~Yan Lam}.} \bibinfo{year}{2024}\natexlab{}.
\newblock \showarticletitle{{Local differential privacy and its applications: A comprehensive survey}}.
\newblock \bibinfo{journal}{\emph{Computer Standards \&AMP; Interfaces}}  \bibinfo{volume}{89} (\bibinfo{year}{2024}), \bibinfo{pages}{103827}.
\newblock
\showISSN{0920-5489}


\bibitem[\protect\citeauthoryear{Ye, Hu, Meng, and Zheng}{Ye et~al\mbox{.}}{2019}]%
        {YHMZ2019}
\bibfield{author}{\bibinfo{person}{Qingqing Ye}, \bibinfo{person}{Haibo Hu}, \bibinfo{person}{Xiaofeng Meng}, {and} \bibinfo{person}{Huadi Zheng}.} \bibinfo{year}{2019}\natexlab{}.
\newblock \showarticletitle{PrivKV: Key-Value Data Collection with Local Differential Privacy}. In \bibinfo{booktitle}{\emph{2019 IEEE Symposium on Security and Privacy (S\&P)}}. IEEE, \bibinfo{pages}{317--331}.
\newblock


\bibitem[\protect\citeauthoryear{Zhang, Lan, Duan, Chen, Zhong, and Xiong}{Zhang et~al\mbox{.}}{2024}]%
        {ZDC24}
\bibfield{author}{\bibinfo{person}{Shun Zhang}, \bibinfo{person}{Pengfei Lan}, \bibinfo{person}{Benfei Duan}, \bibinfo{person}{Zhili Chen}, \bibinfo{person}{Hong Zhong}, {and} \bibinfo{person}{Neal~N. Xiong}.} \bibinfo{year}{2024}\natexlab{}.
\newblock \showarticletitle{{DPIVE}: A regionalized location obfuscation scheme with personalized privacy levels}.
\newblock \bibinfo{journal}{\emph{ACM Transactions on Sensor Networks}} \bibinfo{volume}{20}, \bibinfo{number}{2} (\bibinfo{year}{2024}).
\newblock
\newblock
\shownote{Article 35, 26 pages.}


\bibitem[\protect\citeauthoryear{Zhang, Wang, Li, He, and Chen}{Zhang et~al\mbox{.}}{2018}]%
        {ZWLHC2018}
\bibfield{author}{\bibinfo{person}{Zhikun Zhang}, \bibinfo{person}{Tianhao Wang}, \bibinfo{person}{Ninghui Li}, \bibinfo{person}{Shibo He}, {and} \bibinfo{person}{Jiming Chen}.} \bibinfo{year}{2018}\natexlab{}.
\newblock \showarticletitle{{CALM}: Consistent adaptive local marginal for marginal release under local differential privacy}. In \bibinfo{booktitle}{\emph{Proceedings of the 2018 ACM SIGSAC Conference on Computer and Communications Security}}. \bibinfo{pages}{212--229}.
\newblock


\bibitem[\protect\citeauthoryear{Zhao and Chen}{Zhao and Chen}{2022}]%
        {ZC2022}
\bibfield{author}{\bibinfo{person}{Ying Zhao} {and} \bibinfo{person}{Jinjun Chen}.} \bibinfo{year}{2022}\natexlab{}.
\newblock \showarticletitle{A Survey on Differential Privacy for Unstructured Data Content}.
\newblock \bibinfo{journal}{\emph{ACM Comput. Surv.}} \bibinfo{volume}{54}, \bibinfo{number}{10s}, Article \bibinfo{articleno}{207} (\bibinfo{year}{2022}), \bibinfo{numpages}{28}~pages.
\newblock
\showISSN{0360-0300}


\bibitem[\protect\citeauthoryear{Zhu, Hong, Xue, Lan, Zhang, and Xiang}{Zhu et~al\mbox{.}}{2025}]%
        {ZHX2025}
\bibfield{author}{\bibinfo{person}{Youwen Zhu}, \bibinfo{person}{Yuanyuan Hong}, \bibinfo{person}{Qiao Xue}, \bibinfo{person}{Xiao Lan}, \bibinfo{person}{Yushu Zhang}, {and} \bibinfo{person}{Yong Xiang}.} \bibinfo{year}{2025}\natexlab{}.
\newblock \showarticletitle{{LDGI}: Location-Discriminative Geo-Indistinguishability for Location Privacy}.
\newblock \bibinfo{journal}{\emph{IEEE Transactions on Knowledge and Data Engineering}} \bibinfo{volume}{37}, \bibinfo{number}{3} (\bibinfo{year}{2025}), \bibinfo{pages}{1282--1293}.
\newblock


\end{thebibliography}

\clearpage

\appendix

\section{Proofs}


\subsection{Proof of Theorem \ref{thm:BRR_qloss}}\label{app:proof_part2}

We only need to consider the case where $ N $ is a sufficiently large odd integer. Without loss of generality, let $ m $ also be an odd number satisfying $m \sim \frac{N}{e^{\epsilon/2} + 1}$.
The symbol $\sim$ means the equivalence between two formulas for sufficiently large $N$.

Given any priori position $ k \leq N $ with $ k \in \mathbb{Z}^{+} $, the distances from any of $N$ positions (to be possibly reported) to $ k $ can be arranged in ascending order. As illustrated in each row of Fig.~\ref{fig3}, the general form of this ordered sequence is $[0, 1, 1, 2, 2, \ldots, n, n, n+1, n+2, \ldots, N - 1 - n],$ where $0< 2n \leq N - 1 $ and $k=n+1$ or $N-n$. We now analyze the local expected errors
$Q^{(k)}$ under two protocols.

For GRR mechanism, the  $Q^{(k)}(GRR)$ is computed as, 
$$
Q^{(k)}(GRR) = \left[\frac{(N - n)(N - n - 1)}{2} + \frac{n(n + 1)}{2}\right] \cdot \frac{1}{e^{\epsilon} + N - 1}.
$$

For BRR mechanism, we distinguish three cases:

\begin{enumerate}[left=0pt]
    \item When $ m - 1 \leq 2n \leq N - 1 $, the $Q^{(k)}(BRR)$ is expressed as,
    \begin{equation*}
    \begin{split}
    &Q^{(k)}(BRR) = \frac{m - 1}{2} \left(\frac{m - 1}{2} + 1\right) \cdot \frac{e^{\epsilon}}{N + (e^{\epsilon} - 1)m}+\\
    &\left[\frac{(N - n)(N - n - 1)}{2} + \frac{n(n + 1)}{2} - \frac{m - 1}{2} \frac{m + 1}{2}\right] \cdot \frac{1}{N + (e^{\epsilon} - 1)m}.
    \end{split}
    \end{equation*}
    
    \item When $  2n <m$, the $Q^{(k)}(BRR)$ becomes,
    \begin{equation*}
    \begin{split}
    &Q^{(k)}(BRR) = \left[\frac{(m - 1 - n)(m - n)}{2} + \frac{n(n + 1)}{2}\right] \cdot \frac{e^{\epsilon}}{N + (e^{\epsilon} - 1)m} +\\
    & \left[\frac{(N - n)(N - n - 1)}{2} - \frac{(m - 1 - n)(m - n)}{2}\right] \cdot \frac{1}{N + (e^{\epsilon} - 1)m}.
    \end{split}
    \end{equation*}
   \item When $n=0$, the priori  point is located at the extreme position $k=1$\ or\ $N$, and the sorted distance sequence is $[0, 1, 2, \ldots, N-1]$, where $0$ appears once. The $Q^{(k)}(BRR)$ is identical to the above formula with $n=0$ in case (2).
\end{enumerate}

In the following, we proceed with the computation  employing asymptotic equivalent substitutions for sufficiently large $N$. We define the normalized parameters: $m = cN, \,n=dN , \; \text{with} \; c_0 = \frac{1}{e^{\epsilon/2} + 1},\;c < \frac{1}{2} \; \text{and}\; 0\le d < \frac{1}{2}$.

Substituting the $Q^{(k)}(GRR)$ expression and simplifying yields, with considering only the factors affecting coefficients of $N^2$,
$$
Q^{(k)}(GRR) \sim \frac{N^2}{N} \left[\frac{(1 - d)^2}{2} + \frac{d^2}{2}\right] = \frac{1}{2}(1 - 2d + 2d^2)N.
$$

For $Q^{(k)}(BRR)$, we discuss the asymptotic equivalent substitution of order in two cases for sufficiently large $N$.

 (1) Let's consider the  case $ m - 1 \leq 2n \leq N - 1 $. Then,
    \begin{equation*}
    \begin{split}
Q^{(k)}(BRR)/N = &\frac{c}{2}  \cdot \frac{c}{2}  \cdot \frac{e^{\epsilon}}{1 + (e^{\epsilon} - 1)c}\\
& +\left[\frac{(1 - d)^2}{2} + \frac{d^2}{2} - \frac{c^2}{4}\right] \cdot \frac{1}{1 + (e^{\epsilon} - 1)c}\\
=&\frac{1}{4(1 + (e^{\epsilon} - 1)c)} \left[c^2(e^{\epsilon} - 1) + 2(1 - 2d + 2d^2)\right].
    \end{split}
    \end{equation*}  
  Considering the inequality $ cN - 1 \leq 2dN \leq N - 1 $, we have $\frac{c}{2} \leq d < \frac{1}{2}$ for sufficiently large $N$.
Moreover, $\frac{1}{2} \leq 1 - 2d + 2d^2 = 2\left(d - \frac{1}{2}\right)^2+ \frac{1}{2} \leq 1 - c + \frac{c^2}{2}.$ Substituting $ c=c_0 = \frac{1}{e^{\epsilon/2} + 1} $ and $ d = \frac{1}{2}$, we obtain an upper bound on the expected-error ratio, for sufficiently large $N$,
\begin{equation*}
\frac{Q^{(k)}(BRR)}{Q^{(k)}(GRR)} \leq \frac{1}{e^{\epsilon/2}} \left( \frac{e^{\epsilon/2} - 1}{e^{\epsilon/2} + 1} + 1 \right) = \frac{2}{e^{\epsilon/2} + 1}.
\end{equation*}
For the lower bound, substituting  $ c=c_0 = \frac{1}{e^{\epsilon/2} + 1} $ and $ d = \frac{c}{2}$, we have, for sufficiently large $N$,
\begin{equation*}
\frac{Q^{(k)}(BRR)}{Q^{(k)}(GRR)} \ge
\frac{1}{e^{\epsilon/2}} \left[\frac{c^2(e^{\epsilon} - 1)}{2(1 - 2d + 2d^2)} + 1\right]
=\frac{3e^{\epsilon/2} + 2}{2e^{\epsilon} + 2e^{\epsilon/2} + 1}.
\end{equation*}
Combining both bounds, we have, for sufficiently large $N$,
\begin{equation*}
\frac{3e^{\epsilon/2} + 2}{2e^{\epsilon} + 2e^{\epsilon/2} + 1}
\leq
\frac{Q^{(k)}(BRR)}{Q^{(k)}(GRR)}
\leq
\frac{2}{e^{\epsilon/2} + 1}.
\end{equation*} 

(2) Next, we consider the second case: $2n < m$, which gives $ 0 \leq d < \frac{c}{2} < \frac{1}{4}$.
We derive the asymptotic equivalent expression for the local expected error of BRR, for sufficiently large $N$, 
\begin{equation*}
\begin{split}
Q^{(k)}(BRR)/N = &\left[ \frac{(c - d)^2}{2} +\frac{d^2}{2} \right] \cdot \frac{e^{\epsilon}}{1 + (e^{\epsilon} - 1)c} \\& +\left[ \frac{(1 - d)^2}{2} - \frac{(c - d)^2}{2} \right] \cdot \frac{1}{1 + (e^{\epsilon} - 1)c}\\
 =&\frac{1}{2}\cdot\frac{(1-2d+2d^2)e^{\epsilon}+(c-1)(c+1-2d)(e^{\epsilon}-1)}{(e^{\epsilon}-1)c+1}.   
\end{split}
\end{equation*}

Then, the expected-error ratio of BRR to GRR is,
$$\frac{Q^{(k)}(BRR)}{Q^{(k)}(GRR)} 
\sim \frac{e^{\epsilon}}{(e^{\epsilon} - 1)c + 1} +\frac{ (c - 1)(c + 1 - 2d)(e^{\epsilon} - 1)}{(1 - 2d + 2d^2)\left[(e^{\epsilon} - 1)c + 1\right]}.$$
For the expression $\frac{c+1-2d}{1-2d+2d^2} $, we take the derivative of the expression with respect to $d$ that gives $\frac{2c - 4d(c + 1) + 4d^2}{(1 - 2d + 2d^2)^2}$. Let $2c-4(c+1)d+4d^2=0$, then $\Delta =16(c^2+1)$ with respect to $d$. The left root of the equation is $0 < d_1 = \frac{c + 1 - \sqrt{c^2 + 1}}{2} < \frac{c}{2}$. We observe that the function is monotonic over the intervals $(0, d_1)$ and $(d_1, \frac{c}{2})$, increasing first and then decreasing. Now we estimate the ratio $\frac{Q^{(k)}(BRR)}{Q^{(k)}(GRR)}$ at the stationary point and the two endpoints as follows.
\begin{itemize}
    \item When $d = 0$, corresponding to the above case (3), we have,  as $N\rightarrow +\infty$,  $c\to c_0$ and
\begin{equation*}
\begin{split}
        \frac{Q^{(k)}(BRR)}{Q^{(k)}(GRR)} &\to \frac{1}{e^{\epsilon/2}} \left[ e^{\epsilon} + (c_0 + 1)(c_0 - 1)(e^{\epsilon} - 1) \right]\\& = e^{\frac{\epsilon }{2} } -\frac{e^{\frac{\epsilon }{2} }-1}{e^{\frac{\epsilon }{2} }+1} (\sqrt{e^{\epsilon } +2e^{\epsilon }+1} +1)= \frac{2}{e^{\epsilon/2} + 1}.
\end{split}
\end{equation*}
    \item When $d = d_1$, 
  we have,  as $N\rightarrow +\infty$,   
    $$\frac{(c + 1 - 2d)(c - 1)}{1 - 2d + 2d^2} = (c - 1)\left( \sqrt{c^2 + 1} + c \right),
    $$ 
    and then, $c\to c_0$ gives
\begin{equation*}
    \begin{split}
            \frac{Q^{(k)}(BRR)}{Q^{(k)}(GRR)} &\to \frac{1}{e^{\epsilon/2}} \left[ e^{\epsilon} + (c_0 - 1)(\sqrt{c_0^2 + 1} + c_0)(e^{\epsilon} - 1) \right] \\&= e^{\frac{\epsilon }{2} } -\frac{e^{\frac{\epsilon }{2} }-1}{e^{\frac{\epsilon }{2} }+1} (\sqrt{e^{\epsilon } +2e^{\epsilon }+2} +1).
    \end{split}
\end{equation*}
    \item When $d = \frac{c}{2}\rightarrow \frac{c_0}{2}$, as in case (1),   we have,   as $N\rightarrow +\infty$, 
    $$
    \frac{Q^{(k)}(BRR)}{Q^{(k)}(GRR)} \to \frac{3e^{\epsilon/2} + 2}{2e^{\epsilon} + 2e^{\epsilon/2} + 1}.
    $$
\end{itemize}

By the monotonicity, the expected-error ratio reaches its lower bound when $ d = d_1 $. 
As for the upper bound, the asymptotic order of the two settings, $d=0$ and $d=c/2$, are identical to the upper and lower bounds of case 1), respectively. 

Hence, for each priori number $k$, the upper and lower bounds for the local expected-error ratio are obtained.
For the gap between the upper and lower bounds, we have
\begin{equation*}
    \begin{split}
    \text{\rm sup-inf}
    =&\frac{e^{\frac{\epsilon}{2}} -1}{e^{\frac{\epsilon}{2}} +1}\left(\sqrt{e^{\epsilon} +2e^{\frac{\epsilon}{2}} +2} -\sqrt{e^{\epsilon} +2e^{\frac{\epsilon}{2}} +1}\right)\\
    =&\frac{e^{\frac{\epsilon}{2}} -1}{e^{\frac{\epsilon}{2}} +1}\left(\sqrt{e^{\epsilon} +2e^{\frac{\epsilon}{2}} +2} +e^{\frac{\epsilon}{2}} +1\right)^{-1}\\<&\frac12\left(e^{\frac{\epsilon}{2}} +1\right)^{-1}.
    \end{split}
\end{equation*}

\subsection{Proof of Theorem \ref{thm:global_ratio_qloss}}\label{app:proof_part3}

Denote $n= dN$ as before.  Due to the symmetry, we can account for only  the case $k\le (N+1)/2$.
Since $N$ tends to infinity, we can assume that  both $N$  and $m = m(N)$ are odd.

As demonstrated above, for a priori position $k$, $Q^{(k)}(GRR,N) \sim (1 - 2d + 2d^2)N/2$. Then,
\begin{equation*}
    \begin{split}
        N\cdot Q_g(GRR)&=\sum_{k=1}^{N} Q^{(k)} (GRR)\sim ·2\sum_{n=1}^{\frac{N-1}{2} } \left(\frac{N}{2}-n+\frac{n^{2} }{N} \right) \\&=
        2\left(\frac{N^{2} -N}{4}-\frac{N^{2} -1}{8} +  \frac{N^{2}-1 }{24}  \right)\\&=\frac{1}{3} N^{2} -\frac{1}{2} N+\frac{1}{6} .
    \end{split}
\end{equation*}

Denote $m_0=\frac{N}{e^{\frac{\epsilon}{2}} +1},\ c_0=\frac {m_0}N=\frac{1}{e^{\frac{\epsilon}{2}} +1}$ and $c=c(N)=\frac {m}N$. Due to
$\lim_{N\rightarrow\infty} m/N=\frac{1}{e^{\frac{\epsilon}{2}} +1}$ by Theorem \ref{thm:BRR}, we have $\lim_{N\rightarrow\infty}c(N)=c_0$, that is,
\[\lim_{N\rightarrow\infty}\frac{m-m_0}{N}=\lim_{N\rightarrow\infty}\frac{m-\frac{N}{e^{\frac{\epsilon}{2}} +1}}{N}=\lim_{N\rightarrow\infty}\frac m{N}-\frac{1}{e^{\frac{\epsilon}{2}} +1}=0. \]

We are now ready to calculate the global expected error of BRR with the dynamic splitting number $m=cN$ as follows.

\begin{equation}\label{eq:N_Qg}
    \begin{split}
        &N\cdot Q_{g} (BRR)=\sum_{k=1}^{N} Q^{(k)} (BRR)\\ =& 2\left(\sum_{k=1}^{\frac{m+1}{2} } Q^{(k)} (BRR)+\sum_{k=\frac{m+3}{2} }^{\frac{N+1}{2} }  Q^{(k)} (BRR) \right)-Q^{(\frac{N+1}{2})}(BRR)\\ =&
        \sum_{n=0}^{\frac{m-1}{2} } \frac{ \dfrac{2e^{\epsilon}}{N} n(n+1) - 2n  + c(e^{\epsilon} - 1)(Nc-1-2n) + N - 1 }{1 + (e^{\epsilon} - 1)c } \\ &+
        \sum_{n=\frac{m+1}{2} }^{\frac{N-1}{2} }\frac{N\left[\left({c}^{2}-\frac{1}{N^2}\right)(e^{\epsilon } -1)+2-\frac{4n}{N}+\frac{4n}{N^2}+\frac{4n^{2} }{N^{2} }-\frac{2}{N}\right] }{2(1+(e^{\epsilon } -1)c)} \\ &-
        \frac{N\cdot \left(\left({c}^{2}-\frac{1}{N^2}\right)(e^{\epsilon } -1)+2-\frac{1}{N^{2} }-1+\frac{2}{N}-\frac{2}{N}   \right)}{4(1+(e^{\epsilon } -1)c)}.
    \end{split}
\end{equation}

Denote $A=\frac{1}{(e^{\epsilon }-1 )c+1}$, then

\begin{equation*}
    \begin{split}
        &N\cdot Q_{g} (BRR)\\=&A\sum_{n=0}^{\frac{cN-1}{2} }[c(e^{\epsilon} - 1)(Nc-1) + N - 1]\\&+A\sum_{n=0}^{\frac{cN-1}{2}}[(\frac{2}{N}-2)e^{\epsilon }-2(c-1)(e^{\epsilon }-1)]n + A\sum_{n=0}^{\frac{cN-1}{2}}\frac{2e^{\epsilon }}{N}{n}^{2}\\&+\frac{A}{2}\sum_{n=\frac{cN+1}{2} }^{\frac{N-1}{2} }N\left(\left({c}^{2}-\frac{1}{N^2}\right)(e^{\epsilon }-1)+2-\frac{2}{N}\right)+\frac{A}{2}\sum_{n=\frac{cN+1}{2}}^{\frac{N-1}{2}}[(-4)n+\frac{4n}{N}]\\&+\frac{A}{2}\sum_{n=\frac{cN+1}{2}}^{\frac{N-1}{2}}\frac{4}{N}{n}^{2}-
        \frac{A}{4}N\cdot \left(\left({c}^{2}-\frac{1}{N^2}\right)(e^{\epsilon } -1)+1-\frac{1}{N^{2} }   \right).
    \end{split}
\end{equation*}

By simplifying the above summation, the result is
\begin{equation}
    \begin{split}
        &f(c)=\frac{(e^\epsilon - 1)c^3 + 3(e^\epsilon - 1)c^2 -3e^\epsilon+ 4}{12\left[(e^\epsilon - 1)c + 1\right]} N^2- \frac{1}{12}
        \\
        &=\frac{(e^\epsilon - 1)N^2c^3 + 3(e^\epsilon - 1)N^2c^2 -(e^\epsilon - 1)c -3 e^\epsilon- 1+ 4N^2}{12\left[(e^\epsilon - 1)c + 1\right]}.
        \label{fun:N_Qg}
    \end{split}
  \end{equation}

This approaches $\frac{7e^{\frac{\epsilon }{2} } +  9}{12 \left ( e^{\frac{\epsilon }{2} } +  1 \right )^2 }$ when $N$ tends to infinity, by using $\lim_{N\rightarrow\infty}c(N)=c_0 = \frac{1}{e^{\epsilon/2} + 1}$.

Thus, we come to the final conclusion,

\begin{equation*}
    \begin{split}
        \lim_{N \to \infty} \frac{Q_g(BRR,N)}{Q_g(GRR,N)} =\lim_{N \to \infty} 
        \frac{N^{2}(7e^{\epsilon /2 }+9 )}{12(e^{\epsilon /2 }+1)^2} 
        \bigg/
        \frac{N^{2}}{3}
        =
        \frac{7e^{\epsilon /2 }+9 }{4(e^{\epsilon /2 }+1)^2}.
    \end{split}
\end{equation*}

We would like to mention that, the limit of ratios depends directly on the  terms of order $N^2$ in the expansion of \eqref{eq:N_Qg}, then the case of even $N$ or $m$ remains the final result while the slight changes on the upper and lower limits of the sums take place, like replacing $\frac{m-1}2$ by $\frac{m}2-1$.

\begin{figure*}[htbp]
\centering
\subfigure{
 \begin{minipage}{0.25\linewidth}
 \includegraphics[width=1\textwidth,height=4cm,trim=6.5cm 2cm 7cm 1cm ,clip]{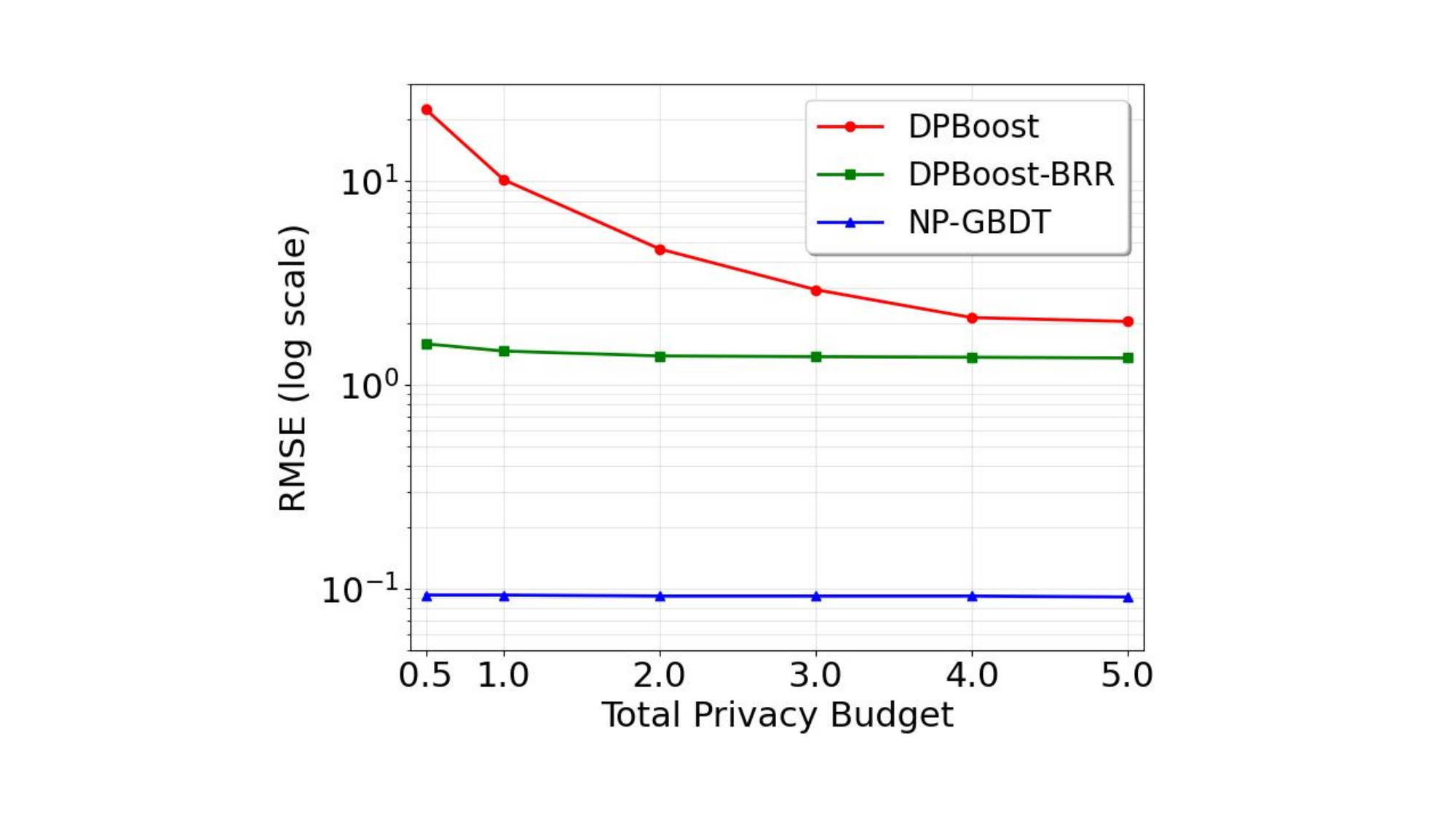} 
 \centerline{\small(a)\quad Bias}
 \label{Bias} 
 \end{minipage}
}
\subfigure{
 \begin{minipage}{0.25\linewidth}
 \includegraphics[width=1\textwidth,height=4cm,trim=7.5cm 2.08cm 6cm 1.8cm,clip]{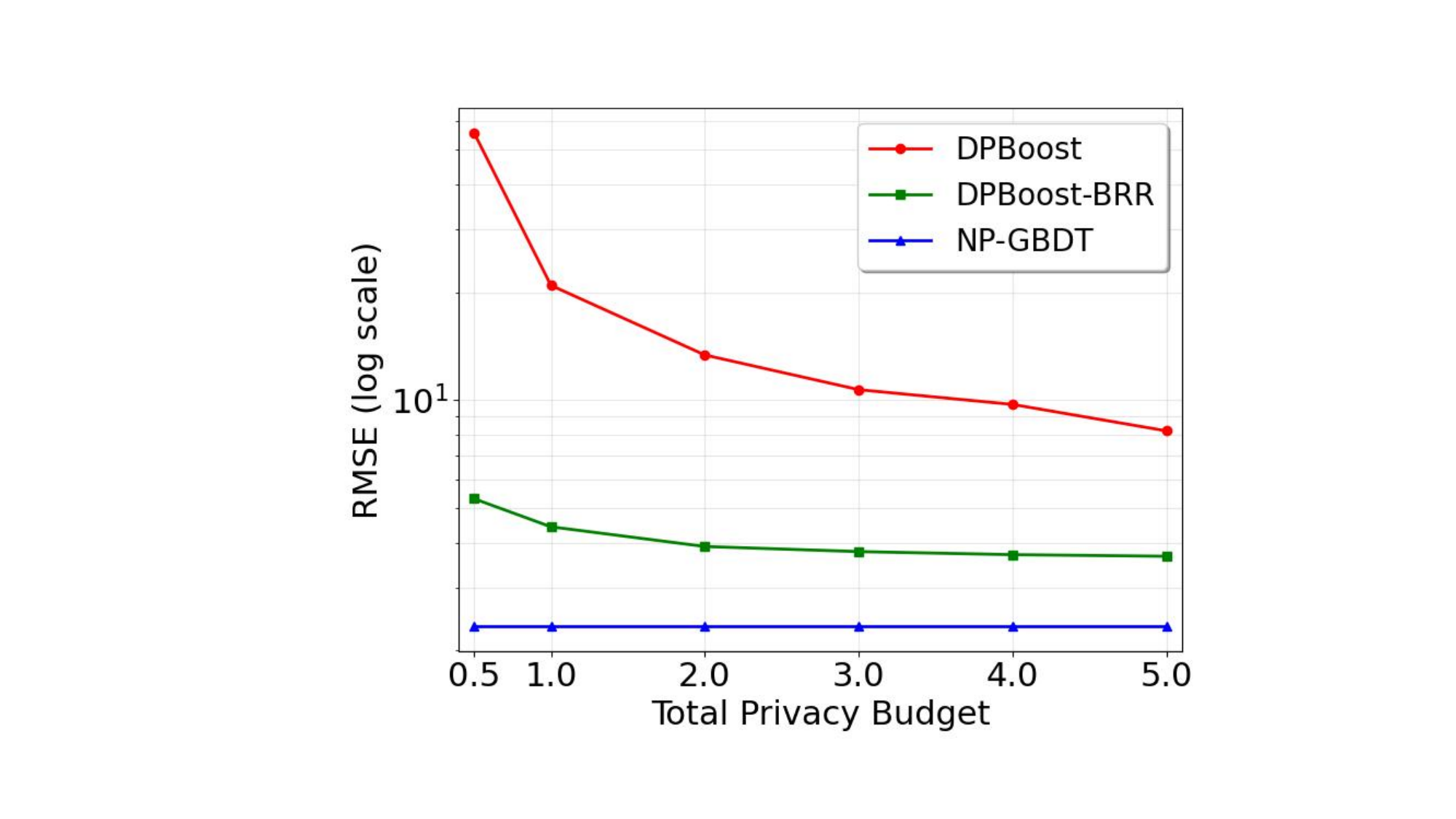} 
 \centerline{\small(b)\quad abalone}
 \label{aba} 
 \end{minipage}
}  
\subfigure{
 \begin{minipage}{0.25\linewidth}
 \includegraphics[width=1\textwidth,height=4cm,trim=7cm 2.2cm 6cm 1.8cm,clip]{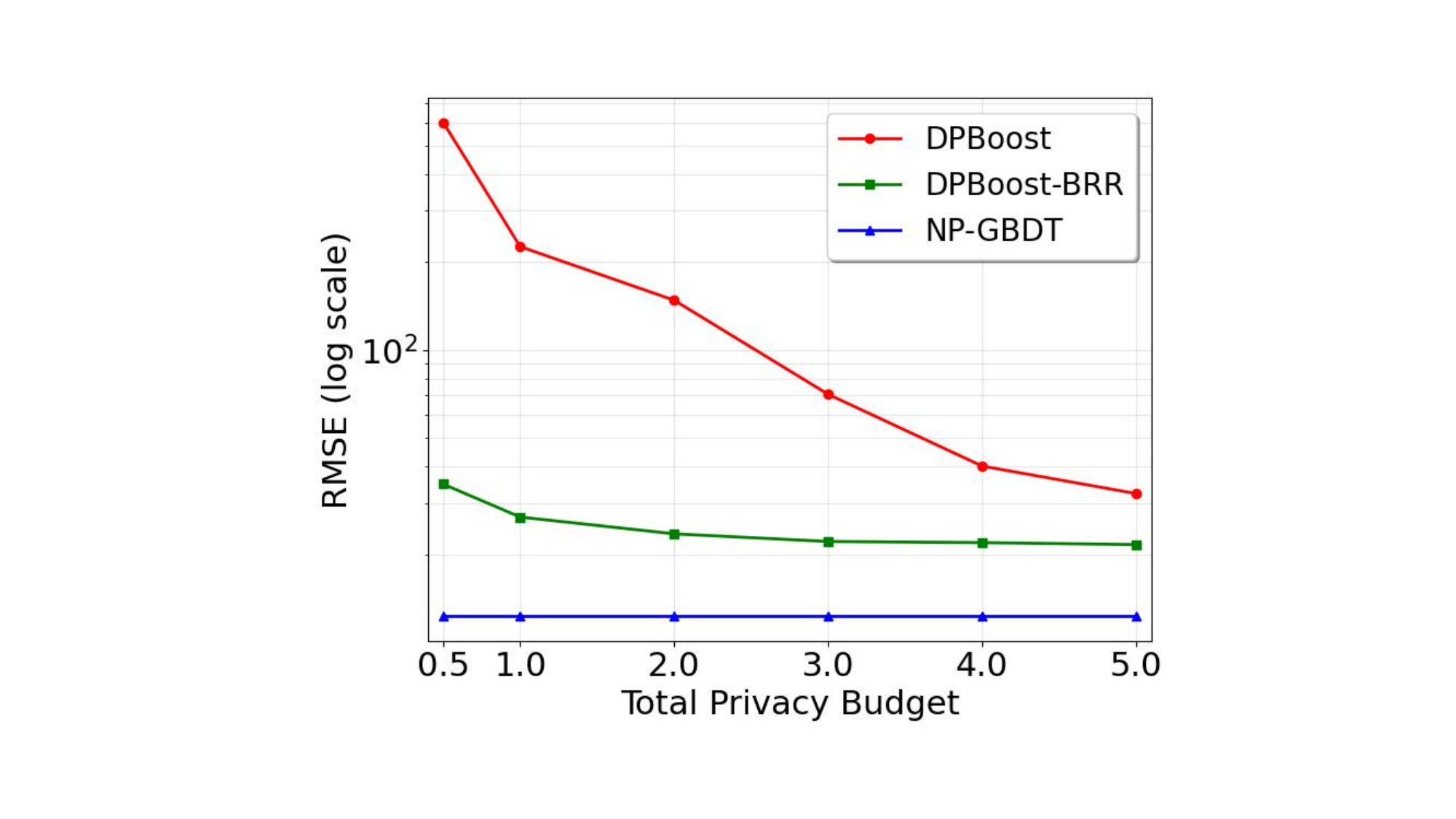} 
  \centerline{\small(c)\quad sup}
 \label{e3} 
 \end{minipage}
}
\vspace{-10pt}
\caption{RMSE comparison of BRR and Laplace mechanisms in privacy-preserving decision tree models  with varying total privacy budgets, cf. DPBoost \cite{LWWH2020} (using Laplace), DPBoost-BRR (using BRR), NP-GBDT (the vanilla GBDT)  \cite{LWWH2020}.}
\label{BRR in privacy-preserving decision tree}
\end{figure*}

\begin{figure*}[htbp]
\centering
\subfigure{
 \begin{minipage}{0.25\linewidth}
 \includegraphics[width=1\textwidth,height=3.6cm,trim=5.5cm 0.2cm 6cm 2cm,clip]{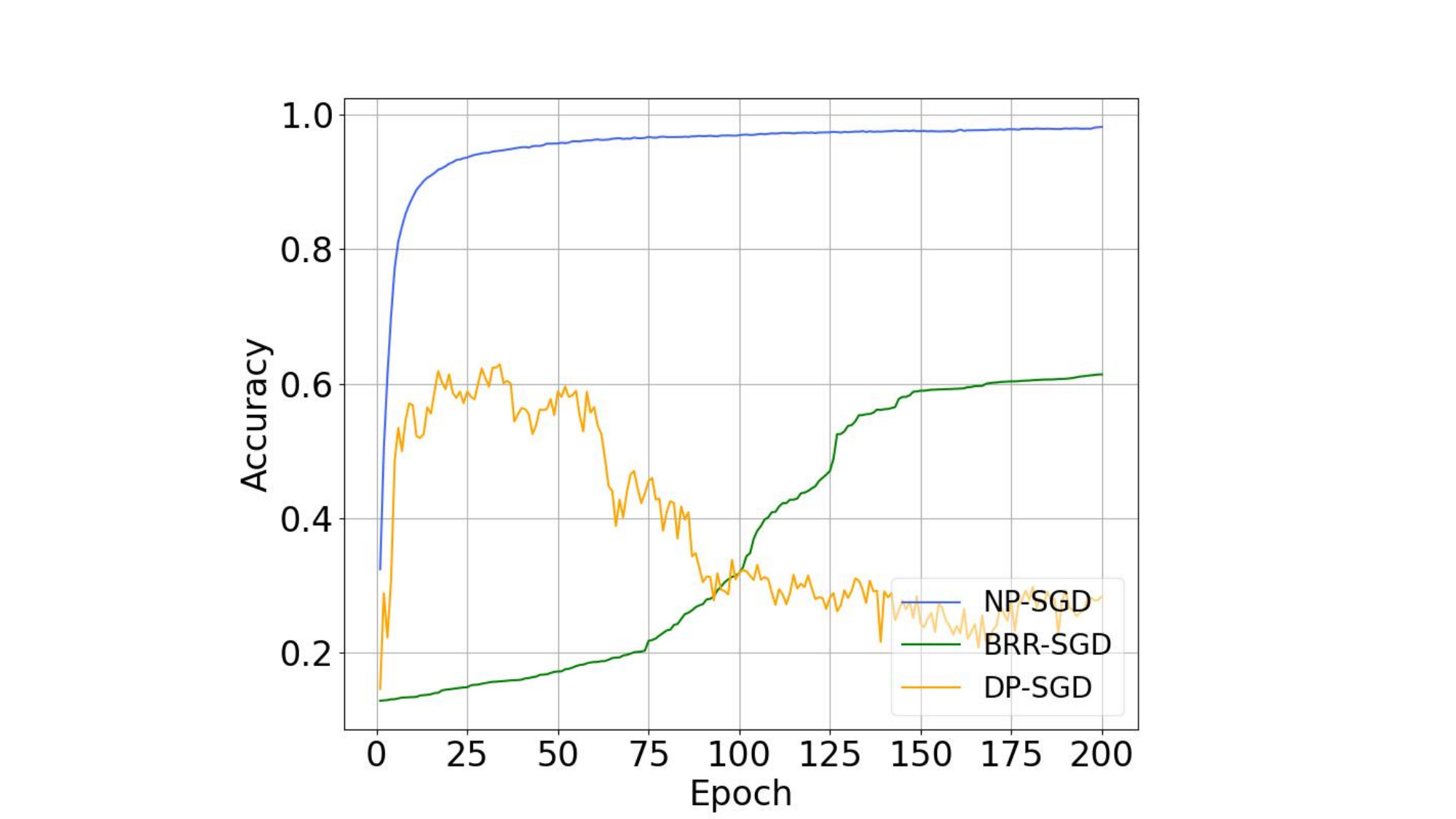} 
 \centerline{\small(a)\quad$\epsilon=0.1$}
 \label{0.1} 
 \end{minipage}
}
\subfigure{
 \begin{minipage}{0.25\linewidth}
 \includegraphics[width=1\textwidth,height=3.6cm,trim=5.5cm 0.2cm 6cm 2cm,clip]{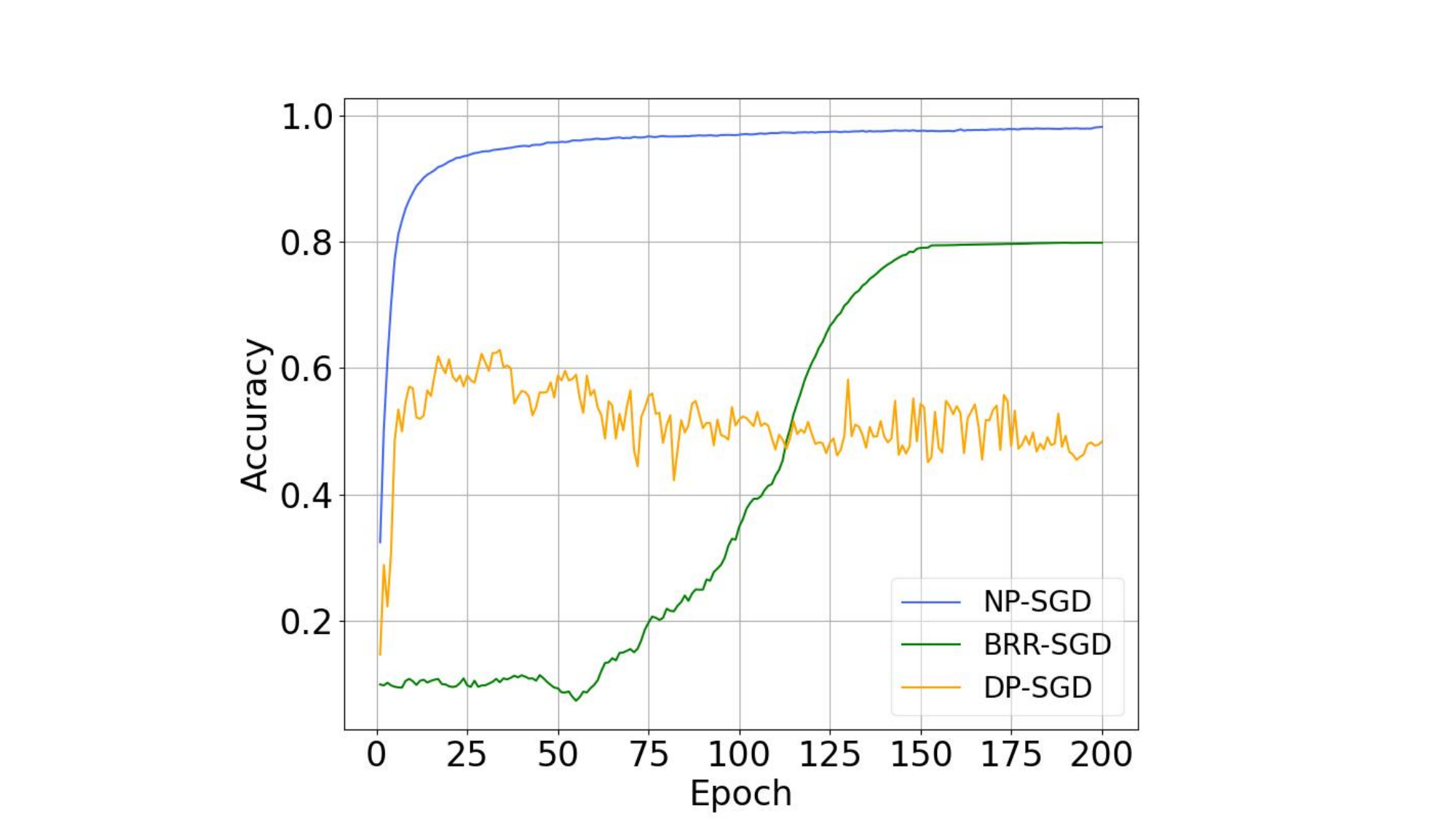} 
 \centerline{\small(b)\quad$\epsilon=0.3$}
 \label{0.3} 
 \end{minipage}
}  
\subfigure{
 \begin{minipage}{0.25\linewidth}
 \includegraphics[width=1\textwidth,height=3.6cm,trim=5cm 0.1cm 6cm 2cm,clip]{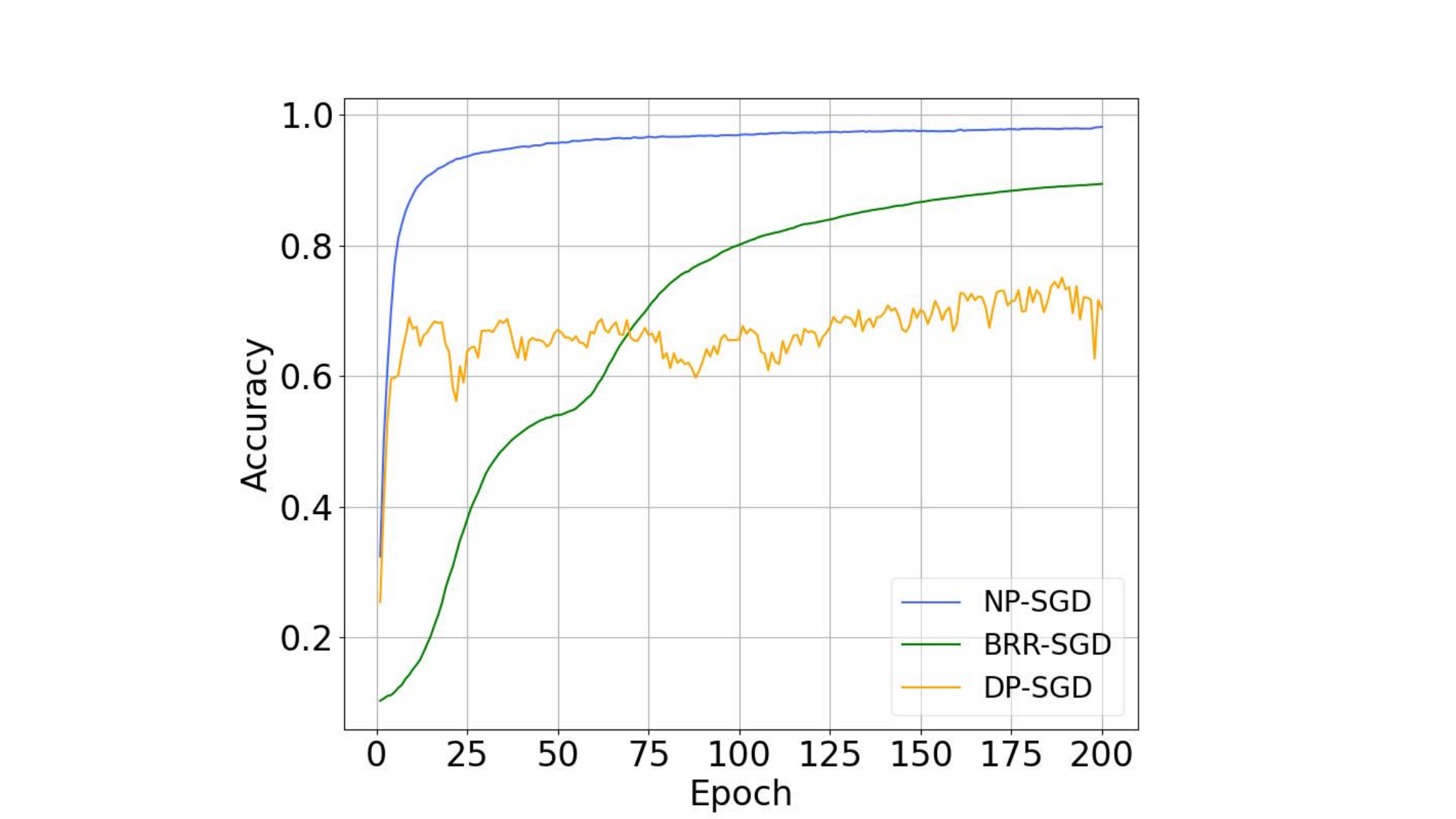} 
  \centerline{\small(c)\quad$\epsilon=0.5$}
 \label{0.5} 
 \end{minipage}
}
\subfigure {
 \begin{minipage}{0.25\linewidth}
 \includegraphics[width=1\textwidth,height=3.6cm,trim=5.5cm 0.2cm 6cm 2cm,clip]{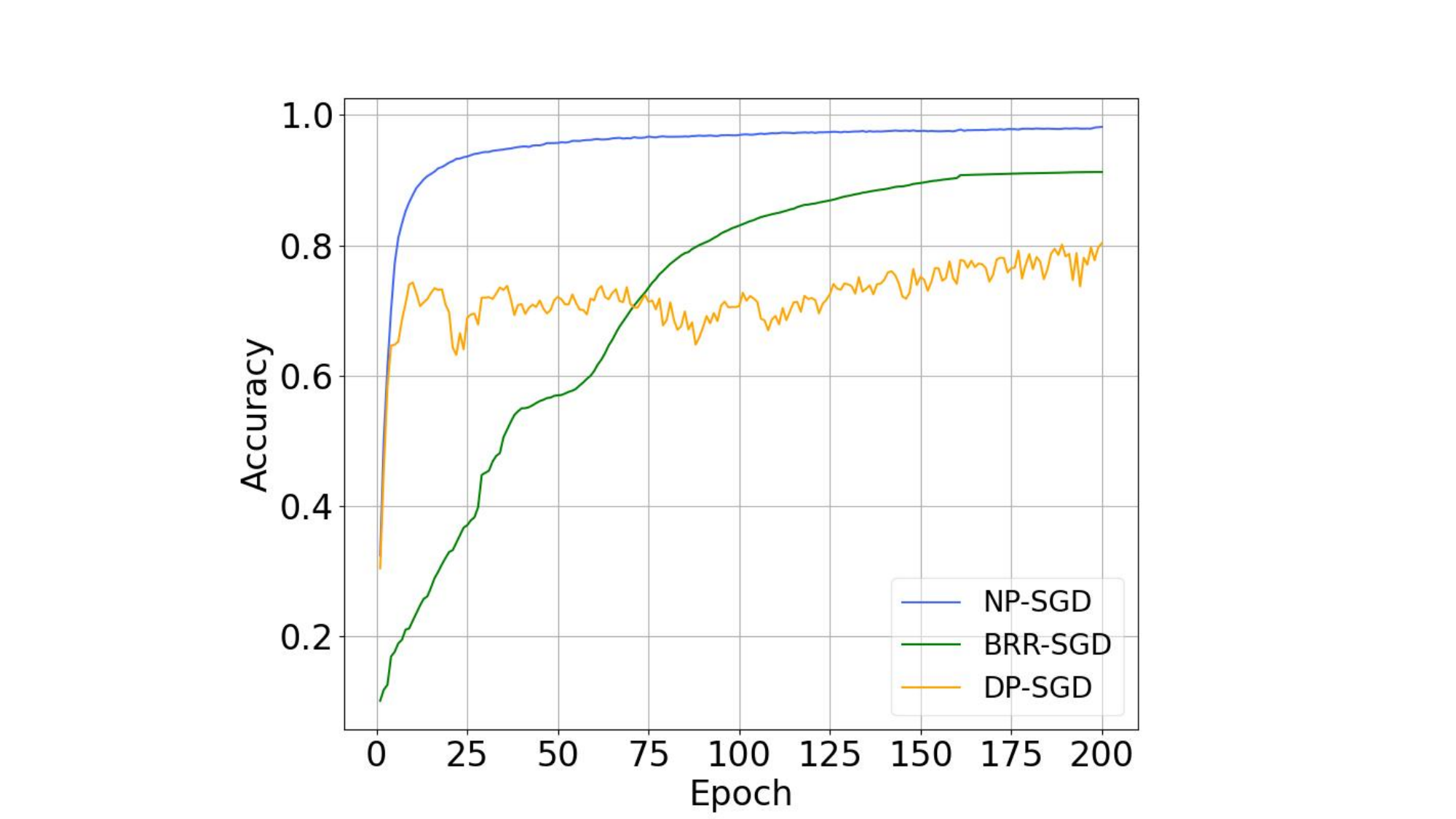} 
  \centerline{\small(d)\quad$\epsilon=0.7$}
 \label{0.7} 
 \end{minipage}
}
\subfigure {
 \begin{minipage}{0.25\linewidth}
 \includegraphics[width=1\textwidth,height=3.6cm,trim=5.5cm 0.2cm 6cm 2cm,clip]{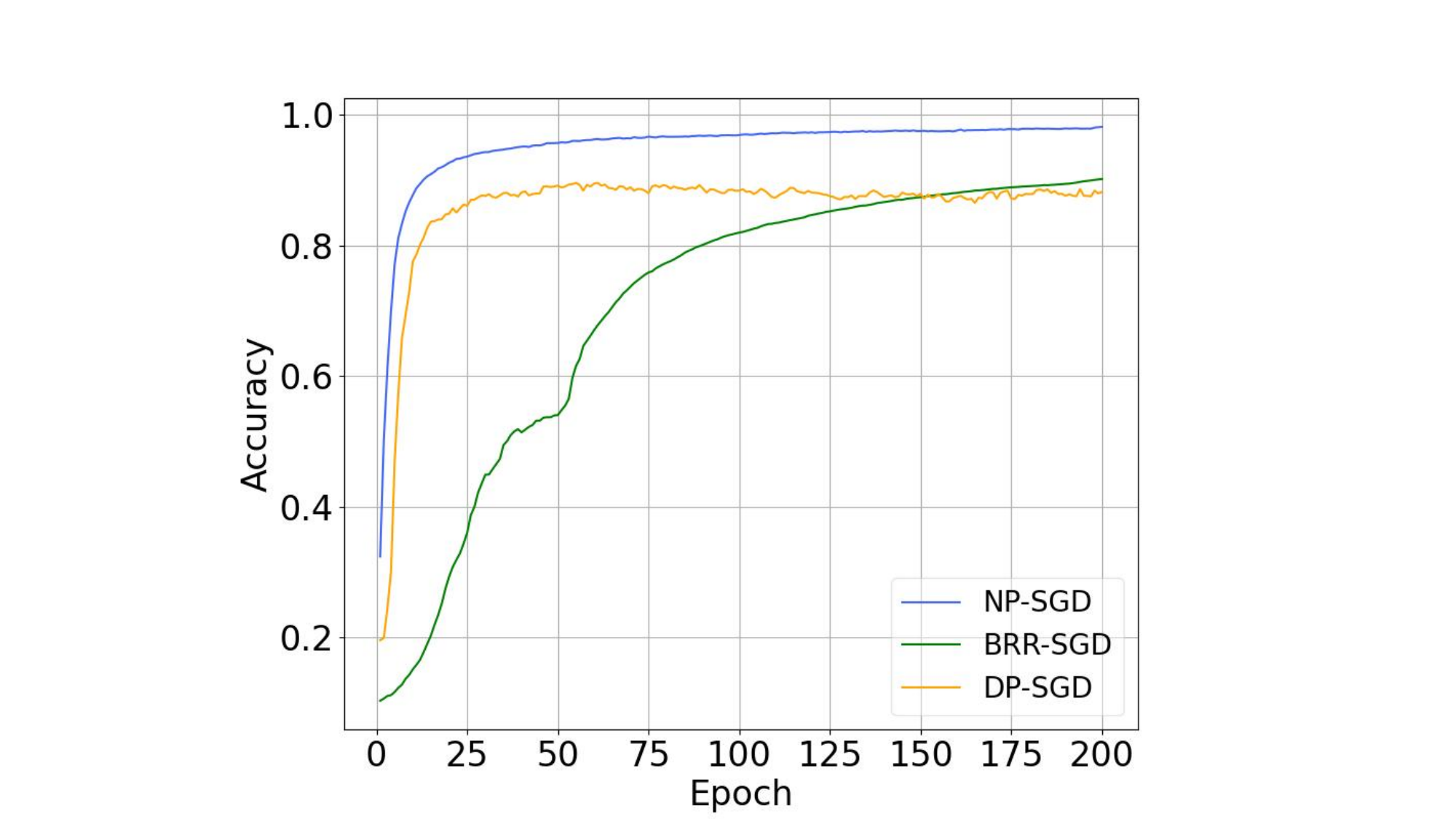}
  \centerline{\small(e)\quad$\epsilon=0.9$}
 \label{0.9} 
 \end{minipage}
}
\subfigure {
 \begin{minipage}{0.25\linewidth}
 \includegraphics[width=1\textwidth,height=3.6cm,trim=5.5cm 0cm 5cm 0cm, clip]{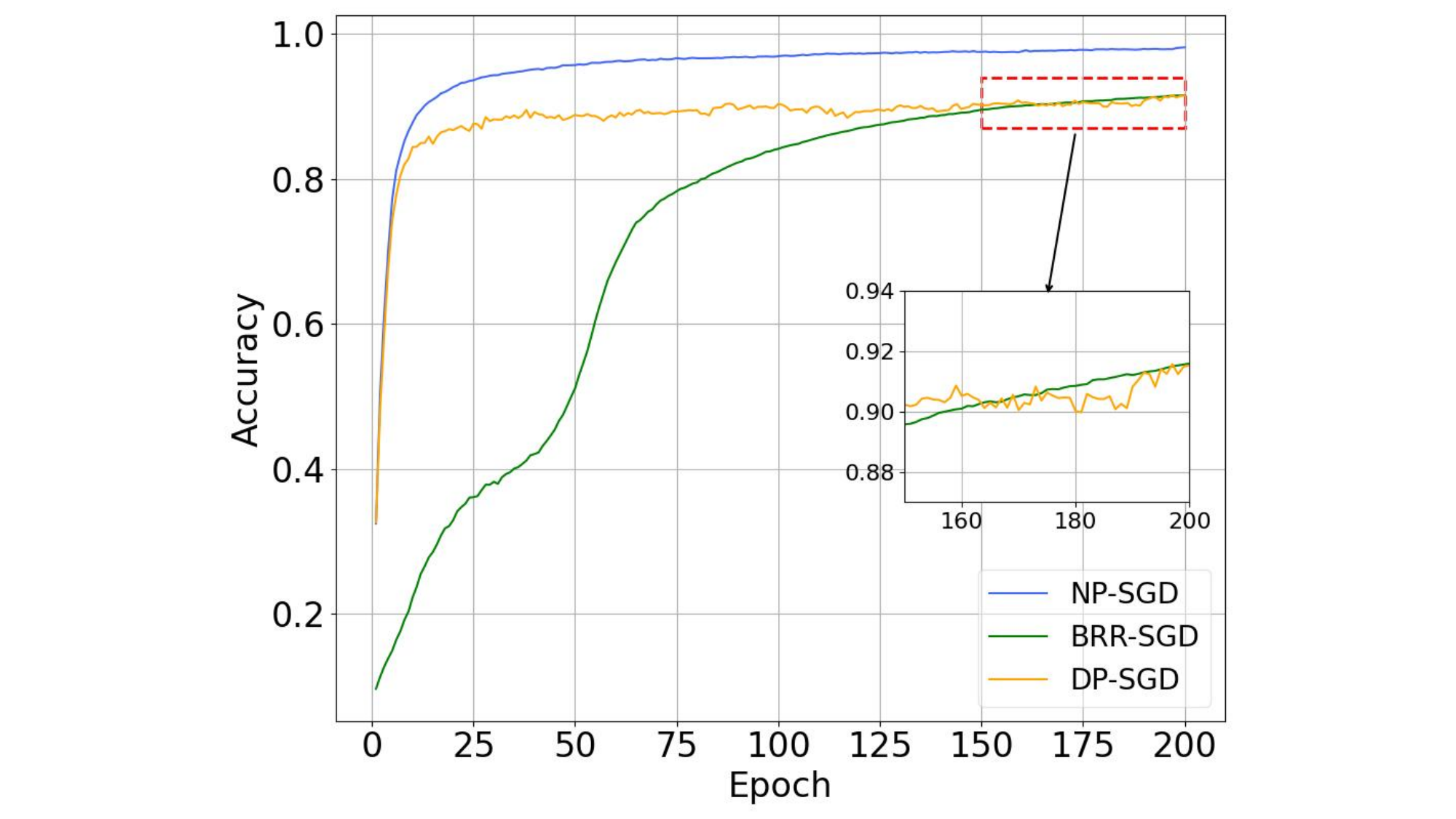} 
  \centerline{\small(f)\quad$\epsilon=1$}
 \label{1} 
 \end{minipage}
}
\vspace{-10pt}
\caption{Accuracy dynamics of differentially private SGD models with Laplace~\cite{ACGMMTZ2016} and BRR  mechanisms, respectively, compared to Non-Private SGD (NP-SGD) across training iterations under varying total privacy budgets, cf. DP-SGD \cite{ACGMMTZ2016} (using Laplace), BRR-SGD (using BRR).}
\label{BRR in SGD}
\end{figure*}

\section{Additional Experimental Evaluation}

\subsection{BRR in Private GBDT}\label{append:GBDT}




Now,  we turn to the applications of BRR in continuous domains for machine learning.
The current experiments are
conducted by deploying BRR mechanism as an alternative to the Laplace mechanism for perturbing 
each leaf value
in DPBoost by Li et al.~\cite{LWWH2020}. Originally, for training DPBoost, 
the sensitivity of the noise added after each geometric leaf clipping is $\Delta q=\min(\frac{g_{l}^\ast}{1+\lambda},2g_{l}^\ast(1-\eta)^{t-1})$. In the training of new DP-GBDT called DPBoost-BRR due to the above alternative, the perturbation range of leaf clipping should be \([-c, c]\) with $c=\Delta q/2$. 
Then  the  range \([-c, c]\) is equally divided for suitable $L$ that gives $n=\frac{2c}{L}+1$  discrete equidistant points. The leaf node values are first adjusted to the closest candidate. Finally, applying the BRR mechanism gives the perturbed leaf values.

For the regression tasks on the Bias, abalone and superconduct (sup) datasets, RMSE (Root Mean Squared Error) is used as the evaluation metric. The number of instances and features for these datasets are $(4177, 8)$, $(7588, 23)$, and $(21263,82)$, respectively. We compare the performance of applying BRR and Laplace mechanisms, respectively, in DPBoost under different privacy budgets. As shown in Fig.~\ref{BRR in privacy-preserving decision tree}, DPBoost-BRR consistently outperforms DPBoost in terms of regression performance. Moreover, the BRR mechanism demonstrates greater stability across various privacy budgets, highlighting its advantage in balancing privacy preservation and model accuracy. For small privacy budgets, the RMSE of DPBoost is excessively high, weakening its practicality in prediction. 
Even at $\epsilon=5$, the RMSE of DPBoost-BRR achieves about $66\%$, $45\%$, and $67\%$ of that for DPBoost (with Laplace) on three datasets, respectively. 

\subsection{BRR in Private SGD}\label{append:SGD}

The experiment is conducted based on the differential private SGD (DP-SGD)~\cite{ACGMMTZ2016} involving Laplace mechanism. The BRR mechanism is compared with the Laplace mechanism by their alternative in DP-SGD, with mainly analyzing how these two obfuscation mechanisms affect model accuracy under varying privacy budgets. Following the procedure in \cite{ACGMMTZ2016}, gradients at each iteration are clipped with the scalar norm constriction at most $C$, thereby bounding the sensitivity of the algorithm by $C$. To better adapt to the changing gradient scale during training, the clipping threshold $C$ is made adaptive. Specifically, let $ \mathbf{g}_t $ be the mean gradient norm of the current mini-batch at iteration $ t $, then $ C $ is updated as $
C_t = \max\left( \rho \cdot C_{t-1} + (1 - \rho) \cdot \mathbf{g}_t, \; C_{\min} \right),$ 
where $ \rho \in (0,1) $ is the decay rate for exponential moving averaging, and $ C_{\min} $ is a predefined minimum threshold to prevent $ C_t $ from being too small and causing numerical instability. NP-SGD is the base scheme where no noise is added on SGD. For the BRR-SGD, after the gradient clipping, the interval $[-C_t, 0]$ or $[ 0, C_t]$ is  split accordingly for the discrete probability distribution on equidistant points of BRR as before.

The experiments are carried out on the MNIST dataset, a widely-used benchmark for handwritten digit recognition, for the classification task. The model performance is evaluated using accuracy as the primary metric, the prediction results of the model are compared with the real labels of the test set and the accuracy is output, with the aim of assessing the effectiveness of these noise mechanisms across different privacy budget conditions.

As shown in Fig.~\ref{BRR in SGD}, with smaller privacy budgets, the injected noise increases, resulting in a decline in the accuracy of model updates. Under these conditions, the Laplace mechanism introduces larger noise magnitudes, leading to significant fluctuations in accuracy and difficulty in maintaining stability. In contrast, the BRR mechanism demonstrates better stability, although its convergence is slower. As the privacy budget increases, the noise intensity decreases, reducing its impact on model updates. The Laplace mechanism can achieve its potential accuracy more quickly; however, its fluctuations persist. On the other hand, the BRR mechanism exhibits a more stable convergence process and eventually reaches, or even surpasses, the accuracy of the Laplace mechanism. When the privacy budget is low, BRR-SGD is significantly better than DP-SGD, and the accuracy indicators are improved by about $28\%$, $21\%$, and $12\%$,  for the privacy budgets $0.3$, $0.5$, and $0.7$, respectively. At higher privacy budgets, both mechanisms inject smaller amounts of noise, and the model’s performance approximates a noise-free scenario. At this phase, the gap in accuracy between DP-SGD and BRR-SGD narrows, and both ultimately achieve similarly high accuracy levels.

Overall, the BRR mechanism proves to be more advantageous in terms of stability and final accuracy of SGD, particularly on balancing privacy protection with maintaining strong model performance.

\begin{figure}[tb]
\centering
\subfigure{
 \begin{minipage}{0.45\linewidth}
 \includegraphics[width=1\textwidth,height=3.5cm]{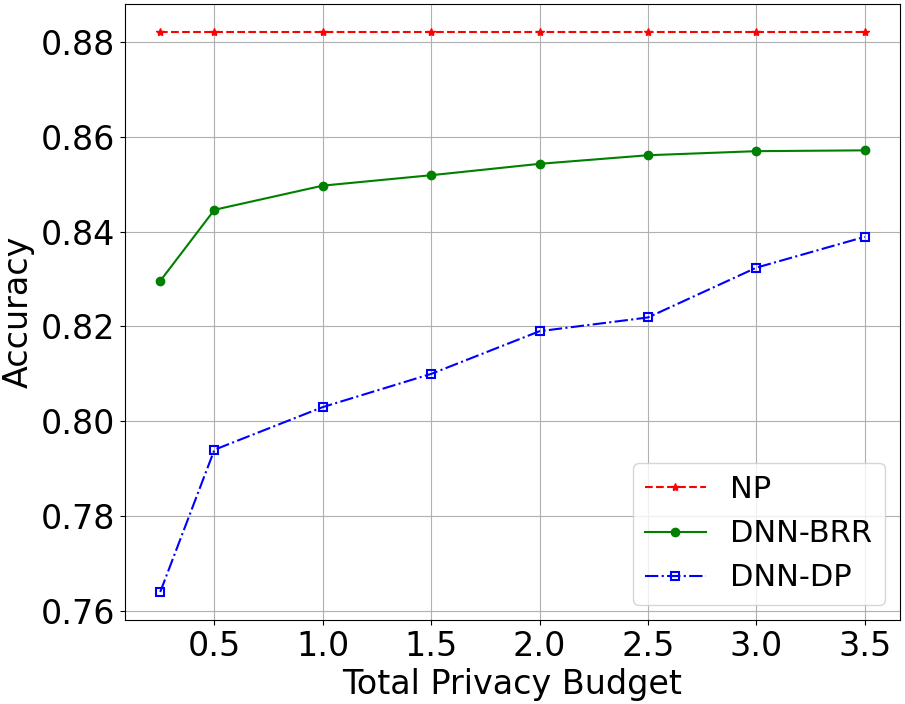} 
 \centerline{\small(a)\quad US Census}
 \label{US}
 \end{minipage}
}
\subfigure{
 \begin{minipage}{0.45\linewidth}
 \includegraphics[width=1\textwidth,height=3.5cm]{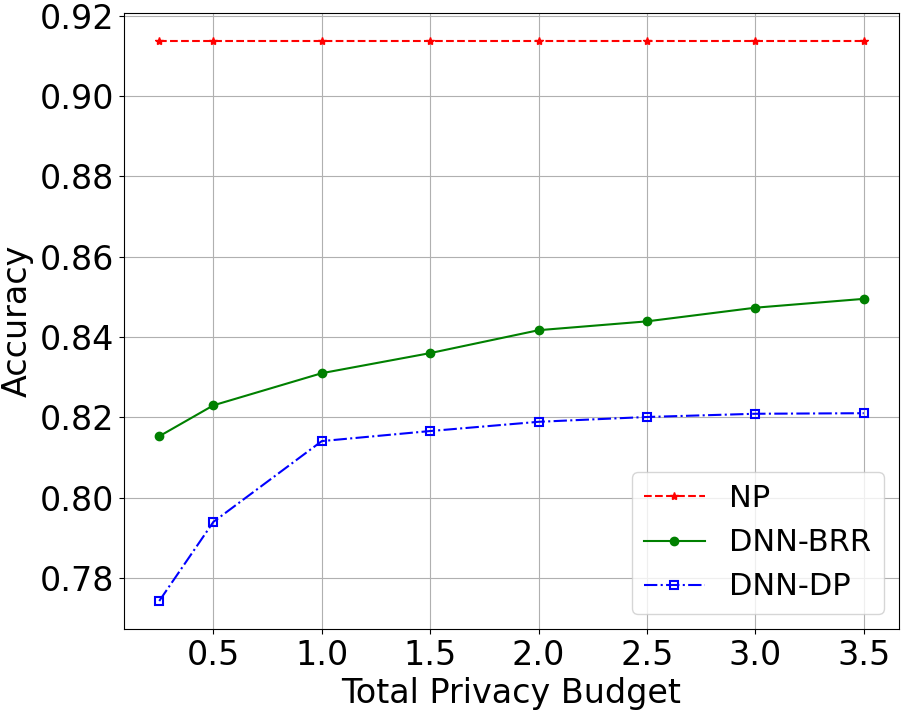} 
 \centerline{\small(b)\quad Bank Marketing}
 \label{bank} 
 \end{minipage}
}
\vspace{-10pt}
\caption{Classification-accuracy  of DNN-DP approaches using the Laplace ~\cite{WGMJ2020} and BRR mechanisms (DNN-BRR), respectively, and the vanilla DNN (NP), on the US Census and bank datasets, with varying total privacy budgets, cf. DNN-DP \cite{WGMJ2020} (using Laplace), DNN-BRR (using BRR), NP \cite{WGMJ2020}.}
\label{BRR_DNN}
\end{figure}

\subsection{BRR for Vector Perturbation in DNN}\label{sect:vector-lap}
To evaluate the effectiveness of the BRR mechanism in vector perturbation, we adopted the DNN-DP framework proposed by Wang et al. ~\cite{WGMJ2020} as a comparative baseline. 
To ensure a fair comparison, we replace the original Laplace noise by 
BRR mechanism while keeping the rest of DNN-DP unchanged. This allows us to compare the two mechanisms' abilities by uniform hyper-parameter optimization to preserve model's performance under privacy constraints.

The experiments are conducted on two publicly available datasets: US Census dataset\footnote{\url{https://archive.ics.uci.edu/ml/datasets/Adult}}, and  Bank Marketing  dataset\footnote{\url{https://archive.ics.uci.edu/ml/datasets/Bank+Marketing}} which contains records from a direct marketing campaign conducted by a Portuguese bank. The former contains $14$ features that were used by ~\cite{WGMJ2020} in the initial DNN-DP experiments, while the latter has been widely adopted in taxonomic studies and contains $20$ features. In order to ensure the fairness and comparability of experimental results, we adopt a uniform hyperparameter optimization strategy. We vary the privacy budget within the range of $0.25$ to $3.5$ to systematically evaluate the classification accuracy of the two mechanisms under different levels of privacy protection.

As shown in Fig.~\ref{BRR_DNN}, the experimental results on the US Census demonstrate that the BRR-based perturbation approach DNN-BRR consistently outperforms the traditional DNN-DP framework 
and also the results shown in \cite{WGMJ2020} as the privacy budget increases from $0.25$ to $3.5$. On average, BRR achieves $3.22\%$ improvement in accuracy over DNN, with a particularly notable gain of $5.29\%$ at the lowest privacy budget ($\epsilon$=0.25). 
Similar results can be observed on the bank dataset where the BRR mechanism also exhibits superior performance compared to the Laplace-based DNN-DP method. 

\subsection{Runtime}\label{append:runtime}

\begin{figure}[tb]
\centering
\subfigure{
 \begin{minipage}{0.45\linewidth}
 \includegraphics[width=1\textwidth,height=3.5cm]{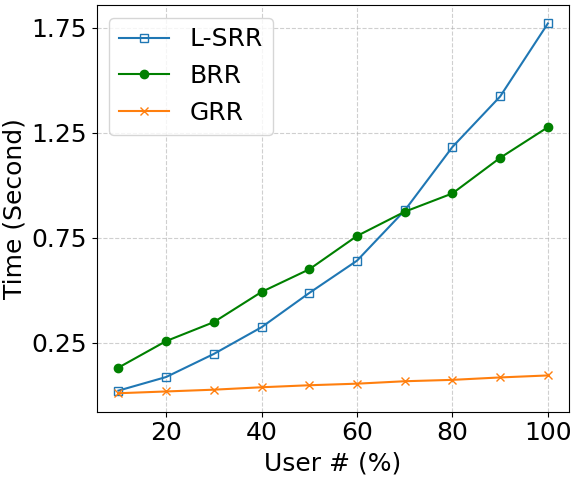} 
 \centerline{\small(a)\quad Gowalla}
 \end{minipage}
}
\subfigure{
 \begin{minipage}{0.45\linewidth}
 \includegraphics[width=1\textwidth,height=3.5cm]{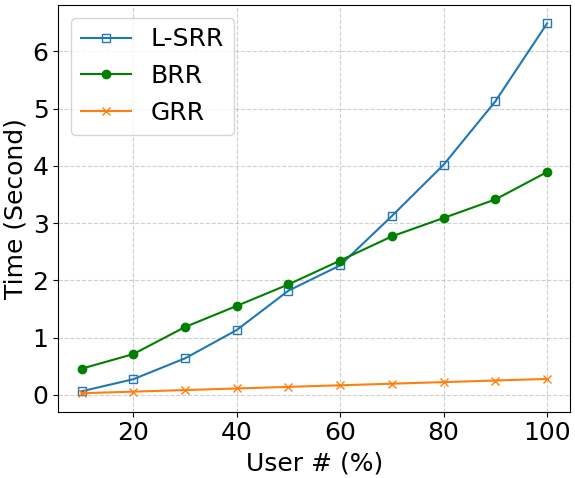} 
 \centerline{\small(b)\quad FourSquare}
 \end{minipage}
}
\vspace{-10pt}
\caption{Runtime of L-SRR \cite{WHXQH2022}, GRR \cite{KOV2016}, and BRR for the server with varying the number of users.}
\label{fig:runtime}
\end{figure}

In this section for runtime, the testing environment consists of the Intel64 Family 6 Model 158 CPU running at 3.41 GHz, 16 GB main memory, Windows 10 version 21H2, and Python version 3.12.

Since each user perturb their locations locally, the user-side
runtime is only $0.03$ second for each user on average. Then we mainly report the server-side runtime by testing $10\%$ to $100\%$ of each dataset demonstrated in Fig.~\ref{distribution}  with comparisons to GRR and L-SRR. For example, the ``$10\%$'' means that the number of users whose locations are to be perturbed accounts to $10\%$ of locations in the whole domain for each dataset.
The division of the whole domain can be finished ahead and is not involved in this testing runtime. 

Fig.~\ref{fig:runtime} shows that GRR takes a little time due to the simple process while L-SRR consumes some more time than BRR due to its complex computation for multiple location groups.
The runtime
of BRR  grows more slowly than L-SRR as the number of users reaches
$3000$ (e.g., $4$ seconds for FourSquare dataset), which is acceptable in real-life applications.

\section{Open Science appendix}

\noindent \textbf{Artifacts provided.}  
To support the reproducibility of our experimental evaluation of the Bipartite Randomized Response (BRR) mechanism and its integration with other differential privacy (DP) techniques, we provide the following artifacts:
\begin{itemize}
    \item Source code implementing the BRR mechanism (\texttt{BRR/BRR.py}) and the local utility optimization algorithm (\texttt{BRR/LocalRS.py})
    \item Implementation of DP-enhanced decision trees (\texttt{DPBoost/})
    \item Reproduction code for the L-SRR\cite{WHXQH2022} scheme by Wang et al. (\texttt{L-SRR/encode.py}, \texttt{L-SRR/partition.py})
    \item Implementations of DP mechanisms applied to Stochastic Gradient Descent and Deep Neural Networks (\texttt{SGDandDNN/})
    \item All experimental datasets are included in the repository, organized within each module’s directory (e.g., DPBoost/data for regression experiments)
    \item Documentation in the form of inline comments and this README, specifying input/output formats, parameter configurations, and execution instructions
\end{itemize}

\noindent \textbf{Access during double-blind review.}  
All artifacts are hosted in an anonymous, publicly accessible repository: 

\url{https://anonymous.4open.science/r/BipartiteRR-5B8E}.  

The repository is self-contained and can be executed independently in each subfolder. We recommend Python 3.12 and standard scientific libraries (e.g., NumPy, Scikit-learn, PyTorch). No authentication or special permissions are required.

\noindent \textbf{Justification for withheld artifacts.}  
All core artifacts necessary to reproduce the results presented in the paper are included in the repository. The experimental datasets are derived from publicly available sources, and no proprietary or sensitive data is used. Therefore, no materials are withheld.
\end{document}